\newcommand{\mytitle}{Placing Green Bridges Optimally for Robust Habitat Reconnection}

\def\thxPACSs{Funded by Deutsche Forschungsgemeinschaft (DFG, German
	Research Foundation), project PACS (FL~1247/1-1, 522475669).}
\def\HUaffil{Humboldt-Universität zu Berlin,
	Department of Computer Science, Algorithm Engineering Group, Germany}

\documentclass[pagebackref]{article}
\usepackage[margin=3cm]{geometry}

\usepackage{authblk}

\title{\Large\bf \mytitle}
\author[1]{Gero Ellmies}
\author[1]{Till Fluschnik\footnote{\thxPACSs}}

\affil[1]{\small\HUaffil%
	\\
	\texttt{$\{$gero.ellmies,till.fluschnik$\}$@hu-berlin.de}
}

\date{}
\usepackage{amsthm}
\usepackage{thmtools}

\usepackage{url}
\usepackage[x11names, table]{xcolor}
\usepackage[utf8]{inputenc}

\usepackage{graphicx}
\usepackage{amsmath}
\usepackage{booktabs}
\usepackage{tabularray}
\UseTblrLibrary{booktabs}
\usepackage{nicefrac}

\usepackage[T1]{fontenc}
\usepackage{amssymb}
\usepackage[outline]{contour}
\usepackage{enumerate}
\usepackage{dsfont}
\usepackage{tabularx}
\usepackage[ruled,vlined,linesnumbered]{algorithm2e}

\usepackage[square,numbers]{natbib}
\usepackage[linkcolor=red!50!black,citecolor=green!40!black,colorlinks]{hyperref}
\usepackage[capitalize,nameinlink]{cleveref}

\usepackage{doi}
\usepackage{paralist}
\usepackage{tikz}
\usetikzlibrary{calc,positioning}
\usetikzlibrary{arrows.meta,positioning}
\usetikzlibrary{patterns.meta}
\usepackage{pgfplots}
\pgfplotsset{compat=1.5}
\pgfplotsset{major grid style={very thin,gray!20!white}} %
\usepackage{pgfplotstable}
\usepackage{mathtools}

\usepackage{multirow}

\usepackage{pifont}%

\newcommand{\appsymb}{$\bigstar$}
\newcommand{\appref}[1]{\ifappendix{}{\hyperref[proof:#1]{\appsymb}}\else{}\appsymb\fi{}}

\usepackage{etoolbox}

\usepackage{xargs}
\newcommandx{\set}[2][1=1]{\ensuremath{\{#1,\ldots,#2\}}}

\newcommandx{\mydefenv}[6][2=A,4=A,6=A]{%
  \ifstrequal{#2}{A}{\newtheorem{#1}{#3}}{}
  \ifstrequal{#4}{A}{\crefname{#1}{#3}{#3s}}{\crefname{#1}{#3}{#4}}%
  \ifstrequal{#6}{A}{\Crefname{#1}{#5{.}}{#5s{.}}}{\Crefname{#1}{#5{.}}{#5{.}}}
}

\mydefenv{theorem}[A]{Theorem}{Thm}
\mydefenv{lemma}[A]{Lemma}{Lem}
\mydefenv{proposition}[A]{Proposition}{Pro}
\mydefenv{conjecture}[A]{Conjecture}{Conj}
\mydefenv{corollary}[A]{Corollary}[Corollaries]{Cor}
\mydefenv{observation}[A]{Observation}{Obs}[B]
\mydefenv{fact}{Fact}{Fact}
\theoremstyle{definition}
\mydefenv{problem}{Problem}{Prob}
\mydefenv{definition}{Definition}{Def}
\mydefenv{construction}{Construction}{Constr}
\mydefenv{rrule}{Reduction Rule}{DRR}
\mydefenv{xrule}{Rule}{Rule}
\mydefenv{myalgo}{Algorithm}{Algo}
\theoremstyle{remark}
\mydefenv{example}{Example}{Ex}
\mydefenv{remark}{Remark}{Rem}
\mydefenv{idea}{Idea}{Idea}
\mydefenv{intuition}{Intuition}{Intuition}

\newcommandx{\decprob}[6][3=Input,5=Question]{
  \begin{problem}[{#1}]\label{prob:#2}
	\textbf{Given} #4,
	the \textbf{question} is whether #6.
  \end{problem}
}

\newcommand{\probref}[1]{\cref{prob:#1}}

\newcommand{\calC}{\mathcal{C}}

\newcommand{\calG}{\mathcal{G}}
\newcommand{\calH}{\mathcal{H}}
\newcommand{\calI}{\mathcal{I}}
\newcommand{\calJ}{\mathcal{J}}

\newcommand{\calS}{\mathcal{S}}

\newcommand{\yes}{\textnormal{\texttt{yes}}}
\newcommand{\no}{\textnormal{\texttt{no}}}
\newcommand{\RD}{$(\Rightarrow)\quad$}
\newcommand{\LD}{$(\Leftarrow)\quad$}
\newcommand{\cqed}{\hfill$\diamond$}

\newcommand{\N}{\mathbb{N}}
\newcommand{\Nzero}{\mathbb{N}_0}

\newcommand{\prob}[1]{\textnormal{\textsc{#1}}}
\newcommand{\hcvcTsc}{\prob{Hamiltonian Cubic Vertex Cover}}
\newcommand{\hcvcAcr}{\prob{HCVC}}
\newcommand{\cvcTsc}{\prob{Cubic Vertex Cover}}
\newcommand{\cvcAcr}{\prob{CVC}}
\newcommand{\tttsatTsc}{\prob{(2,2)-3 Satisfiability}}
\newcommand{\tttsatAcr}{\prob{(2,2)-3SAT}}

\newcommand{\cocl}[1]{\ensuremath{\operatorname{#1}}}
\newcommand{\NP}{\cocl{NP}}

\newcommand{\poly}{\cocl{poly}}
\newcommand{\OPT}{\mathrm{OPT}}

\newcommand{\true}{\ensuremath{\top}}
\newcommand{\false}{\ensuremath{\bot}}

\newcommand{\xcase}[2]{\textit{Case~#1}: #2.}
\newcommand{\xsubcase}[2]{\textit{Subcase~#1}: #2.}

\DeclareMathOperator*{\bigland}{\bigwedge}

\newcommand{\ceq}{\ensuremath{\coloneqq}}
\newcommand{\cif}{\text{if }}
\newcommand{\etal}{et~al.}%

\definecolor{lilla}{HTML}{750787}
\newcommand{\thecolor}{black}%

\usepackage{siunitx}

\usepackage{makecell}

\usepackage{etoolbox}
\newcommand{\ExternalLink}{%
\tikz[x=1.2ex, y=1.2ex, baseline=-0.05ex]{%
    \begin{scope}[x=1ex, y=1ex]
        \clip (-0.1,-0.1) --++ (-0, 1.2) --++ (0.6, 0) --++ (0, -0.6) --++ (0.6, 0) --++ (0, -1);
        \path[draw, line width = 0.5, rounded corners=0.5] (0,0) rectangle (1,1);
    \end{scope}
    \path[draw, line width = 0.5] (0.5, 0.5) -- (1, 1);
    \path[draw, line width = 0.5] (0.6, 1) -- (1, 1) -- (1, 0.6);
    }
}

\newcommand{\rgbpTsc}{\prob{Robust Green Bridges Placement}}
\newcommand{\rgbpAcr}{\prob{RGBP}}

\newcommand{\brgbpAcr}{\prob{B-RGBP}}

\newcommand{\Wlog}{W.l.o.g.}

\newcommand{\Fin}{\ensuremath{F_{\rm in}}}
\newcommand{\bhg}{basic habitat graph}
\newcommand{\Bhg}{Basic habitat graph}

\newcommand{\tikzpreamble}{%
  \tikzstyle{xnode}=[circle,fill,scale=0.5,draw]
  \tikzstyle{xnodeA}=[circle,fill=blue,scale=0.5,draw]
  \tikzstyle{xnodeB}=[circle,fill=orange,scale=0.5,draw]
  \tikzstyle{xnodeC}=[circle,fill=magenta,scale=0.5,draw]
  \tikzstyle{xnodeD}=[circle,fill=teal,scale=0.5,draw]
  \tikzstyle{xnodeF}=[circle,fill=green,scale=0.5,draw]
  \tikzstyle{xcnode}=[circle,fill=blue!10,inner sep=1pt,draw,font=\scriptsize]
  \tikzstyle{xedge}=[-,color=black]
  \tikzstyle{xedgeA}=[-,color=blue]
  \tikzstyle{xedgeB}=[-,color=orange]
  \tikzstyle{xedgeC}=[-,color=magenta]
  \tikzstyle{xedgeD}=[-,color=cyan]
  \tikzstyle{xfedge}=[thick,-,color=red!75!black]
  \tikzstyle{xsedge}=[ultra thick,-,color=green!66!black]
  \tikzstyle{xbedge}=[ultra thick,-,color=blue]
  \tikzstyle{xxedge}=[ultra thick,-,color=\thecolor]
  \tikzstyle{xedgedot}=[thick,-,dotted,color=white]
  \tikzstyle{xhab}=[-,opacity=0.2, line width=3pt, line cap=round, rounded corners]
  \tikzstyle{xhabA}=[xhab,color=cyan]
  \tikzstyle{xhabB}=[xhab,color=orange]
  \tikzstyle{xhabC}=[xhab,color=magenta]
  \tikzstyle{xhabD}=[xhab,color=teal]
  \tikzstyle{xhabE}=[xhab,color=blue]
  \tikzstyle{xhabF}=[xhab,color=green]

  \tikzstyle{xhili}=[circle,opacity=0.2,scale=1.5,fill,draw=none]
  \tikzstyle{xhiliA}=[xhili,color=cyan]
  \tikzstyle{xhiliB}=[xhili,color=orange]
  \tikzstyle{xhiliC}=[xhili,color=magenta]

	\tikzstyle{xemark}=[midway,inner sep=1pt,fill=white,opacity=0.8,font=\scriptsize,sloped]
	\tikzstyle{xemarkz}=[midway,inner sep=1pt,fill=white,font=\scriptsize,sloped]
	\tikzstyle{xemarkb}=[midway,below,inner sep=1pt,fill=white,opacity=0.4,font=\scriptsize,sloped]
	\tikzstyle{xemarka}=[midway,above,inner sep=1pt,fill=white,opacity=0.4,font=\scriptsize,sloped]

  \tikzstyle{xpath}=[->,>=latex,line width=0.35em,opacity=0.25]
  \tikzstyle{xtypeA}=[color=green!66!black]
  \tikzstyle{xtypeB}=[color=blue]
  \tikzstyle{xtypeC}=[color=orange!80!black]
  \tikzstyle{xpathA}=[xpath,xtypeA]
  \tikzstyle{xpathB}=[xpath,xtypeB]
  \tikzstyle{xpathC}=[xpath,xtypeC]
  \tikzstyle{xpathD}=[xpath,color=magenta]
  \tikzstyle{xpathE}=[xpath,color=red]
  \tikzstyle{xpathF}=[xpath,color=brown]
}

\def\teps{0.2*\xr*\yr}
\newcommand{\xn}[1]{($(#1)+(+0,+\teps)$)}
\newcommand{\xne}[1]{($(#1)+(+\teps,+\teps)$)}
\newcommand{\xe}[1]{($(#1)+(+\teps,+0)$)}
\newcommand{\xse}[1]{($(#1)+(+\teps,-\teps)$)}
\newcommand{\xs}[1]{($(#1)+(+0,-\teps)$)}
\newcommand{\xsw}[1]{($(#1)+(-\teps,-\teps)$)}
\newcommand{\xw}[1]{($(#1)+(-\teps,+0)$)}
\newcommand{\xnw}[1]{($(#1)+(-\teps,+\teps)$)}


\newcommand{\myabstract}{%
We study the problem of robustly reconnecting habitats
via the placement of green bridges at minimum total cost.
Habitats are fragmented into patches
and we seek to reconnect each habitat such that it remains connected
even if any of its patches becomes unavailable.
Formally,
we are given an undirected graph
with edge costs,
a set of fixed green bridges represented as a subset of the graph's edges,
a set of habitats represented as vertex subsets,
and some budget.
We decide whether there exists a subset of the graph's edges containing all fixed green bridges
such that,
for each habitat,
the induced subgraph on the solution edges is 2-vertex-connected,
and the total cost does not exceed the budget.
We also study the 2-edge-connectivity variant,
modeling the case where any single reconnecting green bridge may fail.
We analyze the computational complexity of these problems,
focusing on the boundary between NP-hardness and polynomial-time solvability
when the maximum habitat size and maximum vertex degree are bounded by constants.
We prove that for each constant maximum habitat size of at least four there exists a small constant maximum degree for which the problems are NP-hard,
and complement this with polynomial-time algorithms yielding
partial dichotomies for bounded habitat size and degree.
}

\begin{document}

\maketitle

\begin{abstract}
	\myabstract{}
\end{abstract}

\section{Introduction}

When reconnecting fragmented habitats via green bridges,
one is concerned with how robustly a habitat is connected.
Robustness here refers to the scenario in which either a habitat patch or a green bridge becomes unavailable.
A habitat patch can become unavailable due to forest fires, flooding,
or invasive species.
A green bridge can become unavailable due to
defects or maintenance.
In graph-theoretic terms,
the unavailability of habitat patches corresponds to vertex-connectivity,
while the unavailability of green bridges corresponds to edge-connectivity.
It is well-known that vertex-connectivity is stronger
in the sense that any $k$-vertex-connected
graph is also~$k$-edge connected,
while the converse does not hold.
We investigate the cases of 2-vertex-
and 2-edge-connectivity,
modelling the simplest case where a single failure
does not disconnect the habitat.
We present the case of 2-vertex-connectivity
and address the consequences for the case of 2-edge-connectivity
briefly.

We model the problem in line with the literature~\cite{FluschnikK24,HerkenrathFGK22,WallischFK25}
via an undirected graph~$G=(V,E)$.
The habitat patches are represented by the vertex set~$V$,
and the edges in~$E$ represent possible green bridges.
Since the cost of green bridges can differ,
we equip each edge with a cost~$c\colon E\to \Nzero$.
Moreover,
to represent already built bridges,
we provide a set~$\Fin\subseteq E$ of so-called \emph{forced} edges.
We next state our central problem.
Herein,
for a graph~$G=(V,E)$,
vertex subset~$H\subseteq V$ and edge subset~$F\subseteq E$,
we denote by $G[H,F]=(H,\{\{u,v\}\in F\mid u,v\in H\})$
the graph with vertex set~$H$ and all edges from~$F$ with both endpoints in~$H$.
\decprob{\rgbpTsc~(\rgbpAcr)}{rgbp}
{an undirected graph~$G=(V,E)$ with edge costs~$c\colon E\to \Nzero$,
a set~$\Fin\subseteq E$, a set~$\calH$ of subsets of~$V$,
and an integer~$k\in\N$}
{there exists a subset~$F\subseteq E$ with~$\Fin\subseteq F$ and~$\sum_{e\in F} c(e)\leq k$ such that for every~$H\in\calH$ we have that~$G[H,F]$ is 2-vertex-connected}
This work aims to study the
boundary between \NP-hardness and polynomial-time solvability
when the maximum habitat size~$\eta$ and maximum vertex degree~$\Delta$ are bounded by constants.
Concretely,
for each fixed constant value of $\eta$ and~$\Delta$,
we wonder whether \rgbpAcr{} is polynomial-time solvable or \NP-hard.

\subsection{Our Contributions}

By
\brgbpAcr{}
we denote the special case of \rgbpAcr{} with unit costs and no forced edges,
i.e.,
where~$c(e)=1$ for all~$e\in E$ and~$F^*=\emptyset$
(hence, we drop~$c$ and~$F^*$ from the input).
Our main contribution is the following
(depicted in~\cref{fig:results}).
\begin{figure}[t]
  \centering
	\begin{tikzpicture}

	\def\xr{1}
	\def\yr{0.375}
	\def\xsh{8}
	\def\tepsx{\xr*0.5}
	\def\tepsy{\yr*0.5}

	\def\xmin{2}
	\def\xmax{7}
	\def\ymin{3}
	\def\ymax{11}

	\newcommand{\thegrid}[1]{%
		\foreach \x in {\xmin,...,\xmax} {
				\draw[gray!20] (\x*\xr-0.5*\xr,\ymin*\yr-0.5*\yr) -- (\x*\xr-0.5*\xr,\ymax*\yr+0.5*\yr);
				\node at (\x*\xr,\ymin*\yr-1*\yr)[]{\x};
		}
		\foreach \y in {\ymin,...,\ymax} {
		    \def\zscl{1}
				\ifnum\y<7\def\bez{\y}\fi
				\ifnum\y=7\def\bez{$\vdots$}\def\zscl{0.6}\fi
				\ifnum\y=8\def\bez{13}\fi
				\ifnum\y=9\def\bez{14}\fi
				\ifnum\y=10\def\bez{$\vdots$}\def\zscl{0.6}\fi
				\ifnum\y>10\def\bez{22}\fi
				\draw[gray!20] (\xmin*\xr-0.5*\xr,\y*\yr-0.5*\yr) -- (\xmax*\xr+0.5*\xr,\y*\yr-0.5*\yr);
				\node at (\xmin*\xr-0.75*\xr,\y*\yr)[scale=\zscl,inner sep=0pt]{\bez};
		}

		\foreach \x in {\xmin,...,\xmax}
				\foreach \y in {\ymin,...,\ymax}
						\node[circle, inner sep=1pt, fill=red!30, opacity=0.0, draw=none] (n\x-\y) at (\x*\xr,\y*\yr) {}; %

		\node at (\xmin*\xr-0.66*\xr,\ymin*\yr-1.125*\yr)[inner sep=0pt]{$\nicefrac{\eta}{\Delta}$};
	}

	\newcommand{\thelabels}[1]{%
		\def\xlblfnt{\tiny}
		\node at (n7-4)[font=\xlblfnt]{\Cref{prop:H4D7:NPhard}};
		\node at (n5-4)[font=\xlblfnt]{\Cref{prop:H4D5:P}};
		\node at (n5-6)[font=\xlblfnt]{\Cref{prop:H6D5:NPhard}};
		\node at (n3-11)[font=\xlblfnt]{\Cref{prop:D3D4:NPhard}\eqref{item:H22D3:NPhard}};
		\node at (n3-5)[font=\xlblfnt]{\Cref{prop:H5D3:P}};
		\node at (n3-4)[font=\xlblfnt]{\Cref{prop:H4D4:P}};
		\node at (n4-4)[font=\xlblfnt]{\Cref{prop:H4D4:P}};
		\node at (n3-6)[font=\xlblfnt]{\Cref{prop:H6D3:P}};
		\ifnum#1=0
			\node at (n4-8)[font=\xlblfnt]{\Cref{prop:D3D4:NPhard}\eqref{item:H13D4:NPhard}};
			\node at (n4-5)[font=\xlblfnt]{\Cref{prop:H5D4:P}};
			\node at (n6-5)[font=\xlblfnt]{\Cref{prop:H5D6:NPhard}};
		\else
			\node at (n4-9)[font=\xlblfnt]{\Cref{prop:D3D4:NPhard}\eqref{item:H13D4:NPhard}$^\dagger$};
		\fi

	}

	\newcommand{\thedots}{
		\foreach \y in {7,10}{
		 \foreach \x in {2,...,7}{
			\node at (n\x-\y)[scale=0.6]{{\boldmath\(\vdots\)}};
		 }
		}
	}

	\begin{scope}
	 \thegrid{0}

	 \draw[fill=green!50!black, opacity=0.5] ($(n2-3)+(-\tepsx,-\tepsy)$) to ($(n2-11)+(-\tepsx,+\tepsy)$) to ($(n2-11)+(+\tepsx,+\tepsy)$) to ($(n2-3)+(+\tepsx,+\tepsy)$) to ($(n7-3)+(+\tepsx,+\tepsy)$)  to ($(n7-3)+(+\tepsx,-\tepsy)$) to cycle;

	 \draw[fill=green, opacity=0.5] ($(n2-3)+(-\tepsx,-\tepsy)$) to ($(n2-11)+(-\tepsx,+\tepsy)$) to ($(n2-11)+(+\tepsx,+\tepsy)$)
	 to ($(n2-6)+(+\tepsx,+\tepsy)$) to ($(n3-6)+(+\tepsx,+\tepsy)$)  to ($(n3-5)+(+\tepsx,+\tepsy)$)
	 to ($(n4-5)+(+\tepsx,+\tepsy)$) to ($(n4-4)+(+\tepsx,+\tepsy)$) to ($(n5-4)+(+\tepsx,+\tepsy)$) to ($(n5-3)+(+\tepsx,+\tepsy)$) to ($(n7-3)+(+\tepsx,+\tepsy)$) to ($(n7-3)+(+\tepsx,-\tepsy)$) to cycle;

		\draw[fill=red, opacity=0.5] ($(n7-4)+(-\tepsx,-\tepsy)$) to ($(n7-5)+(-\tepsx,-\tepsy)$) to ($(n6-5)+(-\tepsx,-\tepsy)$) to ($(n6-6)+(-\tepsx,-\tepsy)$)
		to ($(n5-6)+(-\tepsx,-\tepsy)$) to ($(n5-8)+(-\tepsx,-\tepsy)$) to ($(n4-8)+(-\tepsx,-\tepsy)$) to ($(n4-10)+(-\tepsx,+\tepsy)$) to ($(n3-10)+(-\tepsx,+\tepsy)$) to ($(n3-11)+(-\tepsx,+\tepsy)$) to ($(n7-11)+(+\tepsx,+\tepsy)$) to ($(n7-4)+(+\tepsx,-\tepsy)$)
		to cycle;

		\thelabels{0}
		\thedots{}
	\end{scope}

	\begin{scope}[xshift=1*\xsh*\xr cm]
	 \thegrid{1}

	 \draw[fill=green!50!black, opacity=0.5] ($(n2-3)+(-\tepsx,-\tepsy)$) to ($(n2-11)+(-\tepsx,+\tepsy)$) to ($(n2-11)+(+\tepsx,+\tepsy)$)
	 to ($(n2-3)+(+\tepsx,+\tepsy)$) to ($(n7-3)+(+\tepsx,+\tepsy)$)  to ($(n7-3)+(+\tepsx,-\tepsy)$) to cycle;

	 \draw[fill=green, opacity=0.5] ($(n2-3)+(-\tepsx,-\tepsy)$) to ($(n2-11)+(-\tepsx,+\tepsy)$) to ($(n2-11)+(+\tepsx,+\tepsy)$) to ($(n2-6)+(+\tepsx,+\tepsy)$) to ($(n3-6)+(+\tepsx,+\tepsy)$) to ($(n3-4)+(+\tepsx,+\tepsy)$) to ($(n5-4)+(+\tepsx,+\tepsy)$) to ($(n5-3)+(+\tepsx,+\tepsy)$) to ($(n7-3)+(+\tepsx,+\tepsy)$) to ($(n7-3)+(+\tepsx,-\tepsy)$) to cycle;

		\draw[fill=red, opacity=0.5] ($(n7-4)+(-\tepsx,-\tepsy)$) to ($(n7-6)+(-\tepsx,-\tepsy)$)
		to ($(n5-6)+(-\tepsx,-\tepsy)$) to ($(n5-9)+(-\tepsx,-\tepsy)$) to ($(n4-9)+(-\tepsx,-\tepsy)$) to ($(n4-10)+(-\tepsx,+\tepsy)$) to ($(n3-10)+(-\tepsx,+\tepsy)$) to ($(n3-11)+(-\tepsx,+\tepsy)$) to ($(n7-11)+(+\tepsx,+\tepsy)$) to ($(n7-4)+(+\tepsx,-\tepsy)$)
		to cycle;

		\thelabels{1}
		\thedots{}
	\end{scope}
	\end{tikzpicture}
	\caption{Overview of our results for (left) \rgbpAcr{} and (right)
	the edge-connectivity version of \rgbpAcr{}.
	Green corresponds to polynomial-time solvability,
	red to~\NP-hardness even with unit costs and no forced edges.
	($^\dagger$: cf.~\Cref{rem:H14D4:NPhard})}
	\label{fig:results}
\end{figure}
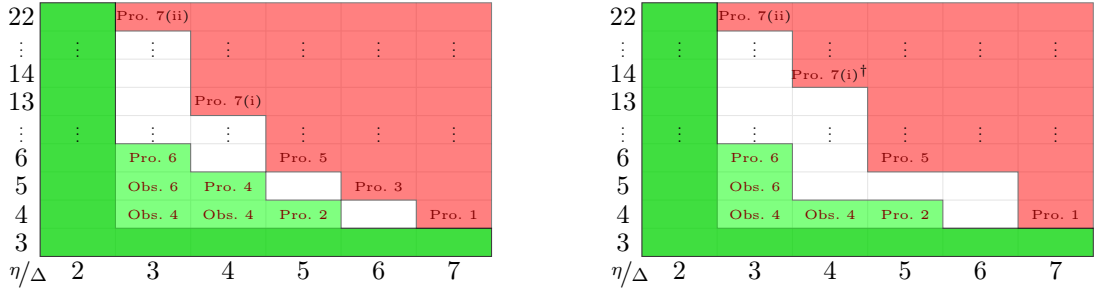

\begin{theorem}
 \label{thm:main}
 \rgbpAcr{} is polynomial-time solvable if~$\Delta\leq 2$ or
 $\eta\leq 3$
 Moreover:
 \begin{compactenum}[(i)]
  \item \rgbpAcr{} is polynomial-time solvable if~$\eta\leq 4$ and~$\Delta\leq 5$,
  but~\brgbpAcr{} is \NP-hard even if~$\eta\geq 4$ and~$\Delta\geq 7$.\label{thm:main:4}
  \item \rgbpAcr{} is polynomial-time solvable if~$\eta\leq 5$ and~$\Delta\leq 4$,
  but~\brgbpAcr{} is \NP-hard even if~$\eta\geq 5$ and~$\Delta\geq 6$.
  \item \rgbpAcr{} is polynomial-time solvable if~$\eta\leq 6$ and~$\Delta\leq 3$,
  but~\brgbpAcr{} is \NP-hard even if~$\eta\geq 6$ and~$\Delta\geq 5$.
  \item \brgbpAcr{} is \NP-hard even if~$\eta\geq 13$ and~$\Delta\geq 4$ or if $\eta\geq 22$ and~$\Delta\geq 3$.
 \end{compactenum}
\end{theorem}
\subsection{Related Work}

\paragraph{Placing Green Bridges Optimally.}
Fluschnik and Kellerhals~\cite{FluschnikK24}
introduced the problem of cost-optimal placement of green bridges
under habitat constraints
and studied \prob{$d$-Reach GBP} (connectivity),
\prob{$d$-Diam GBP} (diameter~$d$),
and \prob{$d$-Closed GBP} (clique-structure).
Herkenrath~\etal~\cite{HerkenrathFGK22}
focus on \prob{$1$-Reach GBP} with small habitats in degree bounded planar graphs,
and provide experiments on real-world landscapes fragmented by highways.
Wallisch~\etal~\cite{WallischFK25}
study \prob{$2$-Diam GBP}
with the same focus as ours,
i.e.,
on small habitats and bounded degrees.
For every constant~$\Delta$ there exist a constant~$\eta_{\Delta}$ such that \prob{$2$-Diam GBP} on graphs of maximum degree~$\Delta$ is trivially polynomial-time solvable if a habitat has size at least~$\eta_{\Delta}$;
we prove that this does not hold for~\rgbpAcr{}.

\paragraph{Spanning 2-connected subgraph and survivable network design.}
\rgbpAcr{} with only one habitat is equivalent to the \prob{Minimum $2$-connected Subgraph (M2cS)} problem,
where one seeks to select a minimum-cost (spanning) edge subset from a graph while keeping it 2-connected.
\prob{M2cS}
is \NP-hard~\cite{GareyJ79},
and the reduction is from \prob{Hamiltonian Cycle (HC)},
which is \NP-hard on cubic graphs~\cite{GareyJT76}.
Since the graph in reduction remains unchanged,
it follows that \rgbpAcr{} is already \NP-hard for only one habitat and maximum degree three.
Note that if the only habitat is of constant size,
then the problem is trivial.
Cheriyan and V\'{e}gh~\cite{CheriyanVegh13} proves a 6-approximation for the problem.

Ene and Vakilian~\cite{EneVakilian14} studies the Survivable Network Design (SNDP) problem
where one seeks a minimum-cost subgraph fulfilling degree and connectivity requirements.
For 2-edge-connectivity,
Jain~\cite{Jain01} proves a 2-approximation.

\paragraph{Community-aware network sparsification.}
The green bridge placement framework of Fluschnik and Kellerhals~\cite{FluschnikK24} is closely related to community-aware network sparsification~\cite{AngluinAR15,GionisRTT17,HerrendorfKMS24},
where habitats correspond to communities and constraints include connectivity~\cite{AngluinAR15,GionisRTT17,HerrendorfKMS24},
star containment~\cite{GionisRTT17,HerrendorfKMS24},
and density~\cite{GionisRTT17}.
Note that density constraints,
even for large~$\alpha<1$,
do not guarantee 2-connectivity.
For a comprehensive overview,
see~\cite{HerrendorfKMS24}.
Cohen~\etal~\cite{CohenHMSW18} study the cost-free clique case,
implying polynomial-time solvability for~$\eta=4$ but \NP-hardness for~$\eta>4$.
In contrast,
we consider bounded-degree graphs,
complementing their clique-based setting.

\paragraph{Group Steiner Trees.}
Khandekar~\etal~\cite{KhandekarKN12} studies the problem of connecting 
(at least one member of) 
each given group via two internally vertex- or edge-disjoint paths. 
Their groups are assumed to be disjoint, 
and not internally 2-connected.
Other works on this are~\cite{ChalermsookDELV18}.

\section{Preliminaries}
\label{sec:prelims}

An undirected graph~$G=(V,E)$ consists of a vertex set~$V$ and edge set~$E\subseteq \binom{V}{2}$.
We also write~$V(G)$ and~$E(G)$ to refer to the vertex and edge set of~$G$,
respectively.
For a vertex subset~$W\subseteq V$,
we denote by~$G[W]\ceq (W,E\cap \binom{W}{2})$ the graph~$G$ induced by~$W$.
For a vertex subset~$W\subseteq V$ or edge subset~$E'\subseteq E$,
we denote the deletion from~$G$ by~$G-W\ceq G[V\setminus W]$ and by~$G-E'\ceq (V,E\setminus E')$
(for singletons, we often write~$G-v$ and~$G-e$).
By~$N_G(v)=\{w\in V\mid \{v,w\}\in E\}$ we define the neighborhood of vertex~$v\in V$.
The degree of vertex~$v$ is~$\deg_G(v)=|N_G(v)|$.
The maximum degree of~$G$ is $\max_{v\in V} \deg_G(v)$.
Two different vertices~$s,t\in V$ are connected if there is a sequence
$(v_1,\dots,v_q)$ of pairwise different vertices in~$V$ with~$s=v_1$ and~$t=v_q$ such that~$\{v_i,v_{i+1}\}\in E$ for all~$i\in\set{q-1}$.
Such a sequence is called a \emph{path}
(a path on~$n$ vertices is also denoted by~$P_n$).
A graph is connected if every pair of vertices is connected.
A subset~$W\subseteq V$ forms a \emph{(connected) component} in~$G$ if every pair of vertices in~$W$ is connected and no vertex~$v\in V\setminus W$ is connected to a vertex in~$W$.
A graph is 2-vertex-connected
(2-connected for short)
if it has at least three vertices, and
for every vertex~$v\in V$,
$G-v$ is connected.
A graph is 2-edge-connected
if it has at least two edges, and
for every edge~$e\in E$,
$G-e$ is connected.\footnote{%
Thus,
any instance~$I=(G,c,\Fin,\calH,k)$ of \rgbpAcr{},
where~$G$ has maximum degree~2,
is a trivial \no-instance
if $|\calH|\geq 1$
or~$c(\Fin)>k$,
and a trivial \yes-instance otherwise.}

\subsection{Decision Problems}

We use the following \cref{prob:hcvc} and \cref{prob:ttts},
which are \NP-complete~\cite{FLEISCHNER20102742,DarmannD21},
in our polynomial-time many-one reductions.

\decprob{\hcvcTsc~(\hcvcAcr)}{hcvc}
{an undirected graph~$G=(V,E)$, a Hamiltonian cycle~$C$ in~$G$, and an integer~$p\in\N$}
{there is vertex subset~$W\subseteq V$ of size at most~$p$ such that every edge in~$E$ has an endpoint in~$W$}

\decprob{\tttsatTsc~(\tttsatAcr)}{ttts}
{a set~$X$
of variables and a (2,2)-3-CNF formula~$\phi=\bigland_{i=1}^M C_i$ of clauses over~$X$, i.e., a 3-CNF formula of clauses over~$X$ each of size exactly three such that each variable appears exactly twice negated and twice unnegated}
{there is a satisfying truth assignment~$\alpha\colon X\to\{\true,\false\}$}

\subsection{Preprocessing}
\newcommand{\therrs}{rr:trivialno,rr:redundantH,rr:forcededge,rr:solvedH,rr:deledge}

To solve an instance of \rgbpAcr{},
it will be useful to preprocess the instance using the following \cref{\therrs}.
We call an instance of \rgbpAcr{} \emph{reduced} if and only if
none of
\cref{\therrs}
is applicable.
Since their correctness and polynomial-time applicability
are immediate,
the corresponding proofs are brief.
Our first rule detects trivial \no-instances.

\begin{rrule}
 \label{rr:trivialno}
 If~$G[H]$ is not 2-connected for some~$H\in \calH$,
 return \no.
\end{rrule}
\begin{proof}
 For every edge set~$F\subseteq E$,
 it holds that~$G[H,F]$ is subgraph of~$G[H]$.
 Thus,
 since $G[H]$ is not 2-connected,
 neither is $G[H,F]$ for every $F\subseteq E$.

 We can apply the rule in polynomial-time
 since for each habitat we can check whether its induced graph is 2-connected in polynomial-time.
\end{proof}
We can safely assume that there are no duplicated habitats.

\begin{rrule}
 \label{rr:redundantH}
 If there are~$H,H'\in \calH$ with~$H=H'$,
 delete one of them.
\end{rrule}
\begin{proof}
 Let~$F\subseteq E$ be a solution.
 Then,
 $G[H,F]$ is 2-connected
 if and only if~$G[H',F]$ is connected.

 We can apply the rule in polynomial-time
 since checking whether two finite sets are equal can be done in polynomial time.
\end{proof}
Every edge required for $G[H]$ to be 2-connected can be safely enforced.

\begin{rrule}
 \label{rr:forcededge}
 If for a habitat~$H$
 there is an edge~$e\in E(G[H])$ such that~$G[H]-e$ is not 2-connected,
 then add~$e$ to~$\Fin$.
\end{rrule}
\begin{proof}
 For every~$F\subseteq E\setminus \{e\}$,
 we have that~$G[H,F]$ is a subgraph of~$G[H]-e$.
 Thus,
 if~$G[H]-e$ is not 2-connected,
 neither is~$G[H,F]$ for every $F\subseteq E\setminus \{e\}$.

 We can apply the rule in polynomial-time
 since for each edge in each habitat
 we check whether the graph without the edge is 2-connected in polynomial-time.
\end{proof}
Note that this implies that if for a habitat~$H$
the graph $G[H]$ contains a vertex~$v$ of degree two,
then its two adjacent edges in~$G[H]$ are added to~$\Fin$.
In particular,
all edges of a 2-connected size-3 habitat (which is a triangle) are forced.

Since~$\Fin\subseteq F$ for every solution~$F$,
we have:

\begin{rrule}
 \label{rr:solvedH}
 If~$G[H,\Fin]$ is 2-connected for some habitat~$H$,
 delete~$H$.
\end{rrule}
\begin{proof}
 Since~$\Fin\subseteq F$ for every solution~$F$,
 we have that if~$G[H,\Fin]$ is 2-connected,
 so is~$G[H,F]$.

 We can apply the rule in polynomial-time
 since for each habitat we can check whether its induced graph on~$\Fin$ is 2-connected in polynomial-time.
\end{proof}
An edge induced by no habitat can be deleted
(budget update may be required).

\begin{rrule}
 \label{rr:deledge}
 If there is an edge~$e\in E$ not contained in~$G[H]$ for every~$H\in \calH$,
 then delete~$e$
 and if~$e\in \Fin$,
 then set~$k\ceq k-c(e)$.
\end{rrule}
\begin{proof}
 Since~$e$ is not in any habitat's induced graph,
 for every minimal solution~$F$
 we have that~$e\in F$ if and only if~$e\in \Fin$.

 We can apply the rule in polynomial-time
 since we can check in polynomial time whether an edge is part of a habitat and~$\Fin$.
\end{proof}

\subsection{\Bhg{}}

Most of our polynomial-time algorithms
rely on the following auxiliary graph that reflects how the habitats intersect in the input graph
(see~\cref{fig:bhg} for an~example).
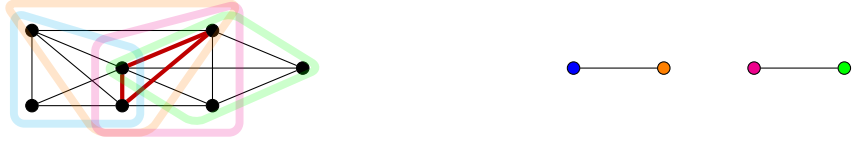
\begin{figure}[t]
 \centering
 \begin{tikzpicture}
  \def\xr{1.2}
  \def\yr{1}
  \tikzpreamble{}

  \begin{scope}
		\foreach \x in {1,3}{\node (x\x) at (\x*\xr,1*\yr)[xnode]{};}
		\foreach \x in {1,2,3}{\node (y\x) at (\x*\xr,0*\yr)[xnode]{};}
		\node (x2) at (2*\xr,0.5*\yr)[xnode]{};
		\node (y4) at (4*\xr,0.5*\yr)[xnode]{};

		\foreach \x in {1,2,3}{\draw[xedge] (x\x) to (y\x);}
		\foreach \x/\y in {1/2,2/3,3/4}{
			\draw[xedge] (x\x) to (y\y);
			\ifnum\x>0\draw[xedge] (y\x) to (y\y);
				\ifnum\y<4\draw[xedge] (y\x) to (x\y);\fi\fi
			\ifnum\y<4\draw[xedge] (x\x) to (x\y);\fi
		}
		\draw[xedge] (x2) to (y4);
		\draw[xedge] (x1) to (x3);

		\draw[xfedge,ultra thick] (x2) to (y2);
		\draw[xfedge,ultra thick] (x2) to (x3);
		\draw[xfedge,ultra thick] (y2) to (x3);

		\def\teps{0.2*\xr*\yr}
		\draw[xhabA] \xnw{x1} -- \xne{x2} -- \xse{y2} -- \xsw{y1} -- cycle;
		\def\teps{0.3*\xr*\yr}
		\draw[xhabB] \xnw{x1} -- \xne{x3} -- \xse{y2} -- \xsw{y2} -- cycle;
		\def\teps{0.3*\xr*\yr}
		\draw[xhabC] \xnw{x2} -- \xne{x3} -- \xse{y3} -- \xsw{y2} -- cycle;
		\def\teps{0.2*\xr*\yr}
		\draw[xhabF] \xw{x2} -- \xne{x3} -- \xe{y4} -- \xsw{y3} -- cycle;
		\def\teps{0.2*\xr*\yr}
  \end{scope}

	\begin{scope}[xshift=7*\xr cm, yshift=0.5*\yr cm]
	 \node (A) at (0*\xr,0*\yr)[xnodeA]{};
	 \node (B) at (1*\xr,0*\yr)[xnodeB]{};
	 \node (C) at (2*\xr,0*\yr)[xnodeC]{};
	 \node (D) at (3*\xr,0*\yr)[xnodeF]{};

	 \foreach \x/\y in {A/B,C/D}{\draw[xedge] (\x) to (\y);}
	\end{scope}
 \end{tikzpicture}
 \caption{Illustrative example to \cref{def:HG}. (Left) A graph $G$, a habitat set~$\calH$ with a blue, orange, magenta, and green size-4 habitat (contained vertices are encircled), and the set~$\Fin$ of forced edges consisting of the thick/red edges.
 (Right) The corresponding basic habitat graph~$\calG_{G,\calH,\Fin}$ with two components each being a path of length one.}
 \label{fig:bhg}
\end{figure}

\begin{definition}
 \label{def:HG}
 Given a graph~$G=(V,E)$ with forced edges~$\Fin\subseteq E$ and a set~$\calH$ of habitats,
 the \emph{\bhg{}} $\calG_{G,\Fin,\calH}=(\calH,Y)$
 is the graph with vertex set~$\calH$ and edge set~$Y=\{\{H,H'\}\mid (E(G[H])\cap E(G[H']))\setminus \Fin \neq \emptyset\}$.
\end{definition}
Wallisch~\etal~\cite{WallischFK25} introduced the \emph{habitat (intersection) graph},
which is similar to our \bhg{}
but also requires that no third habitat ``lies in between'' two adjacent habitats.
We distinguish from their notion via the
suffix ``basic''.
We prove,
in analogue to~\cite[Lemma 12]{WallischFK25},
that the problem essentially decomposes across connected components of the habitat graph.

\begin{observation}
 \label{obs:HG:components}
 Let~$I=(G,c,\Fin,\calH,k)$ be an instance of \rgbpAcr{}.
 Let $C_1,\dots,C_q$ form the connected components in~$\calG_{G,\Fin,\calH}$.
 For each~$p\in\set{q}$,
 let $G_p\ceq (\bigcup_{H\in C_p} H, \bigcup_{H\in C_p} E(G[H]))$.
 Then,
 $F$ is a minimum-cost solution for~$I$
 if and only if
 $F\cap E(G_p)$ is a minimum cost solution for~$I_p\ceq (G_p,c,\Fin\cap E(G_p),C_p,k)$ for every~$p\in\set{q}$.
\end{observation}
\begin{proof}
 By the definition of $\calG_{G,\Fin,\calH}$,
 we know that for every two distinct~$p,p'\in\set{q}$
 we have that $((F\cap E(G_p))\setminus \Fin)\cap ((F\cap E(G_{p'}))\setminus \Fin)=\emptyset$,
 since otherwise there is a habitat in~$C_p$ and one in~$C_{p'}$ that is adjacent.
 Thus,
 every minimum-cost solution $F$ can be represented as~$F=\Fin \uplus \biguplus_{p=1}^q (F\cap E(G_p))\setminus \Fin$.
 The claim now follows since,
 any subsolution not being of minimum cost can be replaced by a minimum-cost subsolution
 without losing feasibility.

 \RD{}
 Let~$F$ be minimum-cost solution.
 Suppose towards a contradiction that there is~$p'\in\set{q}$ such that~$F\cap E(G_{p'})$ is not a minimum cost solution for~$I_{p'}$.
 Let~$F_{p'}$ denote a minimum-cost solution for~$I_{p'}$.
 Then
 $F'=\Fin \uplus (F_{p'}\setminus \Fin) \uplus \biguplus_{p\in\set{q}\setminus\{p'\}} (F\cap E(G_p))\setminus \Fin$
 is a solution to~$I$ of smaller cost;
 a contradiction.

 \LD{}
 Let~$F_p$ be a minimum-cost solution for~$I_p$ for every~$p\in\set{q}$.
 Suppose towards a contradiction that
 $F=\Fin \uplus \biguplus_{p=1}^q (F\cap E(G_p))\setminus \Fin$ is not of minimum-cost,
 and let~$F'$ be a minimum-cost solution for~$I$.
 Since~$F'=\Fin \uplus \biguplus_{p=1}^q (F'\cap E(G_p))\setminus \Fin$,
 there is~$p'\in\set{q}$ such that~$F'\cap E(G_{p'})$ a solution to~$I_{p'}$ of cost smaller than~$F_p$;
 a contradiction.
\end{proof}
Next we show that if the basic habitat graph of an instance is a path or a cycle,
we can use dynamic programming to solve it.
The dynamic program essentially enumerates all feasible edge sets for each habitat and combines them with solutions of its at most two neighboring habitats.
\begin{lemma}
 \label{lem:HG:paths}
 Let~$I=(G,c,\Fin,\calH,k)$ be an instance of \rgbpAcr{} with~$\calG_{G,\Fin,\calH}$ being a (i) path or (ii) a cycle.
 Then,
 we can either report that~$I$ is a \no-instance
 or output a minimum-cost solution~$F$ for~$I$
 in
 (i)~$O(2^{\eta\cdot\Delta}\cdot \poly(|I|))$ time
 or (ii) $O(2^{3\eta\cdot\Delta/2}\cdot \poly(|I|))$ time.\footnote{We denote by $\poly(|I|)$ a polynomial in the input size $|I|$.}
\end{lemma}
\begin{proof}
	We first prove the case (i) of a path
	and then use it for the case (ii) of a cycle.

	Let~$\calH=\{H_1,\dots,H_h\}$ with~$h=|\calH|$ be enumerated such that
	$\{H_i,H_{i+1}\}\in E(\calG_{G,\Fin,\calH})$ for every~$i\in\set{h-1}$.
	For each habitat~$H_i$,
	let~$\calS(H_i)$ denote the set of all edge subsets~$F_i\subseteq E(G[H_i])$
	such that $(\Fin\cap E(G[H_i]))\subseteq F_i$ and~$G[H_i,F_i]$ is 2-connected.
	Note that $|\calS(H_i)|\leq 2^{\frac{\eta\cdot\Delta}{2}}$,
	since in a graph of maximum degree~$\delta$ every graph induced by a size-at-most~$\eta$ habitat has at most~$\frac{\eta\cdot\Delta}{2}$ many edges
	(by Handshaking lemma).
	Moreover,
	we can compute~$\calS(H_i)$ in $O(2^{\frac{\eta\cdot\Delta}{2}}\cdot \poly(|I|))$ time,
	since checking for 2-connectivity can be done in polynomial time.
	If any~$\calS(H_i)$ is empty,
	then report~\no.
	Let~$D[i,F_i]$ with~$F_i\in \calS(H_i)$
	denote the minimum cost of an edge set~$F_i^*$ such that $F_i\subseteq F_i^*$ and for every~$j\in\set{i}$,
	it holds that $G[H_j,F_i^*]$ is 2-connected and $(\Fin\cap E(G[H_j]))\subseteq F_i^*$.
	Let~$D[1,F_1]=c(F_1)$ for every~$F_1\in \calS(H_1)$.
	For each~$i\in\set[2]{n}$ and each~$F_i\in \calS(H_{i})$, let
		\[ D[i,F_i] = \min_{F_{i-1}\in \calS(H_{i-1})} \big( D[i-1, F_{i-1}] + c(F_i\setminus F_{i-1}) \big).
		\]
  We prove that
	$D$ is correct.
	For this,
	we need to prove that for every~$i\in\set{h}$,
	for every~$F_i\in \calS(H_i)$,
	it holds that
	$D[i,F_i]=\alpha$ if and only if
	a minimum-cost edge set~$F_i^*$ such that $F_i\subseteq F_i^*$ and for every~$j\in\set{i}$,
	$G[H_j,F_i^*]$ is 2-connected and $(\Fin\cap E(G[H_j]))\subseteq F_i^*$,
	has cost~$\alpha$.
	By construction,
	we know that this is correct for~$H_1$.
	Assume it is correct for~$i-1$ with~$i>1$.
	Let~$F_{i-1}\in \calS(H_{i-1})$ such that~$D[i,F_i] = D[i-1, F_{i-1}] + c(F_i\setminus F_{i-1})$.

	``$\leq$'': By induction,
	we know that~$D[i-1,F_{i-1}]$ is the minimum cost of an edge set~$F_{i-1}^*$ such that $F_{i-1}\subseteq F_{i-1}^*$ and for every~$j\in\set{i-1}$,
	$G[H_j,F_{i-1}^*]$ is 2-connected and $(\Fin\cap E(G[H_j]))\subseteq F_{i-1}^*$.
	Hence,
	since~$F_i\in \calS(H_{i})$,
	we know that~$D[i,F_i]$ is the cost of an edge set~$F_i^*$ such that $F_i\subseteq F_i^*$ and for every~$j\in\set{i}$,
	$G[H_j,F_i^*]$ is 2-connected and $(\Fin\cap E(G[H_j]))\subseteq F_i^*$.

	``$\geq$'': Let $F_i^*$ be an edge set of minimum cost~$\OPT$ such that $F_i\subseteq F_i^*$ and for every~$j\in\set{i}$,
	$G[H_j,F_i^*]$ is 2-connected and $(\Fin\cap E(G[H_j]))\subseteq F_i^*$.
	For every~$j\in\set{i}$,
	let~$F_{j}'\ceq F_i^*\cap E(G[H_j])$ ;
	note that~$F_{j}'\in \calS(H_j)$ since $G[H_j,F_i^*]$ is 2-connected and $(\Fin\cap E(G[H_j]))\subseteq F_i^*$.
	Then,
	$F'=\bigcup_{j=1}^{i-1} F_{j}'$ is an edge subset such that $F_{i-1}'\subseteq F'$ and for every~$j\in\set{i-1}$,
	$G[H_j,F']$ is 2-connected and $(\Fin\cap E(G[H_j]))\subseteq F'$,
	and let~$\beta$ denote its cost.
	Since~$F_{i-1}'\in\calS(H_{i-1})$,
	by induction,
	we have that~$\OPT \geq \beta+c(F_i\setminus F_{i-1}')\geq  \min_{F_{i-1}\in \calS(H_{i-1})} \big( D[i-1, F_{i-1}] + c(F_i\setminus F_{i-1}) \big) = D[i,F_i]$.
	Finally,
	we can fill each of the at most~$h\cdot 2^{\frac{\eta\cdot\Delta}{2}}$ table entries in~$O(2^{\frac{\eta\cdot\Delta}{2}}\cdot \poly(|I|))$ time.
	Via backtracking,
	we can compute the minimum-cost edge set~$F^*$ such that for every~$j\in\set{n}$,
	it holds that $G[H_j,F^*]$ is 2-connected and $(\Fin\cap E(G[H_j]))\subseteq F^*$.
	Let~$F\ceq\Fin\cup F^*$.
	If~$c(F)>k$,
	output \no,
	otherwise output~$F$:
	since every minimum-cost solution contains~$\Fin$,
	$F$ is a minimum-cost solution for~$I$.

	Almost as a consequence,
	we get (ii) as follows.
	Let~$\calH=\{H_0,H_1,\dots,H_h\}$ with~$h=|\calH|-1$ be enumerated such that $\{H_i,H_{i+1 \bmod h}\}\in E(\calG_{G,\Fin,\calH})$ for every~$i\in\set[0]{h}$.
	For each habitat~$H_i$,
	let~$\calS(H_i)$ denote the set of all edge subsets~$F_i\subseteq E(G[H_i])$
	such that $(\Fin\cap E(G[H_i]))\subseteq F_i$ and~$G[H_i,F_i]$ is 2-connected.
	Recall that~$|\calS(H_i)|\leq 2^{\eta\cdot \Delta/2}$.
	For every~$F^{(j)}_0\in \calS(H_0)$,
	let~$I^{(j)}=(G,c,\Fin^{(j)},\calH^{(j)},k)$ be the instance with~$\Fin^{(j)}=\Fin\cup F^{(j)}_0$ and~$\calH^{(j)}=\calH\setminus H_0$.
	Note that~$\calG_{G,\calH^{(j)}}$ is a path of the form~$(H_1,\dots,H_h)$.
	Hence,
	with~(i),
	we can compute a minimum-cost solution~$F^{(j)}$ for~$I^{(j)}$ in $O(2^{\eta\cdot\Delta}\cdot \poly(|I|))$ time.
	Let~$F \in \arg\,\min_{F^{(j)}_0\in \calS(H_0)} (\Fin^{(j)}\cup F^{(j)})$.
	If~$c(F)>k$,
	then output \no,
	otherwise output~$F$:
	note that since~$\Fin\subseteq \Fin^{(j)}$ is in every minimum-cost solution,
	$F$ as a minimum-cost solution for~$I$.
\end{proof}
\newcommand{\bnd}{\partial}
We also introduce the following notation that will simplify
our proofs. For a set~$\calH'\subseteq \calH$,
we denote by~$\bnd_{\calG_{G,\Fin,\calH}}(\calH')$ the set of all habitats~$H\in \calH\setminus \calH'$
with~$H\cap \bigcup_{H'\in\calH'} H'\neq \emptyset$,
$H\setminus \bigcup_{H'\in\calH'} H'\neq \emptyset$,
and~$\{H,H'\}\in\calG_{G,\Fin,\calH}$ for some~$H'\in\calH'$.
We often write $\bnd(\calH')$ when the subscript is clear from the context.

\section{Habitats of size four}
\label{sec:four}

We first settle that for the case of size-4 habitats,
we need not distinguish between 2-vertex- and 2-edge-connectivity.

\begin{observation}
 \label{obs:H4:same}
 For every graph~$G$ on four vertices,
 it holds that~$G$ is 2-connected if and only if it is 2-edge-connected.
\end{observation}
\begin{proof}
 We only need to prove the backward direction,
 which we prove via contraposition.
 Let~$G=(V,E)$ with~$V=\{u,v,w,x\}$ be connected but not 2-vertex-connected.
 We show that~$G$ is also not 2-edge-connected.
 \Wlog{}, let $v\in V$ such that~$G-v$ is disconnected.
 Let~$A,B$ be two connected components in~$G-v$.
 Since~$G-v$ has only three vertices,
 one of~$A$ or~$B$,
 say~$B$,
 contains a single vertex,
 say~$w$.
 Since~$G$ is connected,
 there is the edge~$e=\{v,w\}\in E$,
 but~$w$ is adjacent neither to~$u$ nor to~$x$.
 Hence,
 $G-e$ is disconnected,
 so
 $G$ is not 2-edge-connected.
\end{proof}
\subsection{\(\eta=4\) and \(\Delta\geq 7\)}

\begin{proposition}\label{prop:H4D7:NPhard}
 \rgbpAcr{} is \NP-hard even if~$\eta=4$ and~$\Delta\geq 7$.
\end{proposition}
We give a polynomial-time many-one reduction from \hcvcAcr{} (\probref{hcvc}).
\begin{intuition}
 \label{int:H4D7:NPhard}
 Each vertex of the vertex cover instance is represented by a $K_4$-habitat.
 All but three edges are forced into every solution (via auxiliary habitats),
 leaving a choice between selecting either one or two of the remaining edges
 to achieve 2-connectivity.
 Selecting two edges corresponds to choosing the vertex into the vertex cover;
 selecting only one corresponds to excluding it.
 If two edges are selected,
 then in the incident edge gadget only two further edges are needed
 to make each of its three $K_4$-habitats 2-connected.
\end{intuition}

\begin{construction}\label{constr:H4D7:NPhard}
 Let~$I=(G,p,C)$ be an instance of HCVC with~$G=(V,E)$ and let $V=\{v_1,\dots,v_n\}$ be enumerated such that~$C=(V,\bigcup_{i=1}^{n-1}\{\{v_i,v_{i+1}\}\}\cup \{\{v_n,v_1\}\})$ forms a Hamiltonian cycle in~$G$,
 and let~$m\ceq|E|$.
 Let~$X=\{(i,j)\mid \{v_i,v_j\}\in E \land i<j\}\setminus\{(1,n)\}\cup\{(n,1)\}$
 (note that~$|X|=|E|$).
 Construct an instance~$I'=(G',\calH,k)$ with~$G'=(V',E')$ of \brgbpAcr{} as follows
 (see~\cref{fig:H4D7:NPhard} for an illustration).
 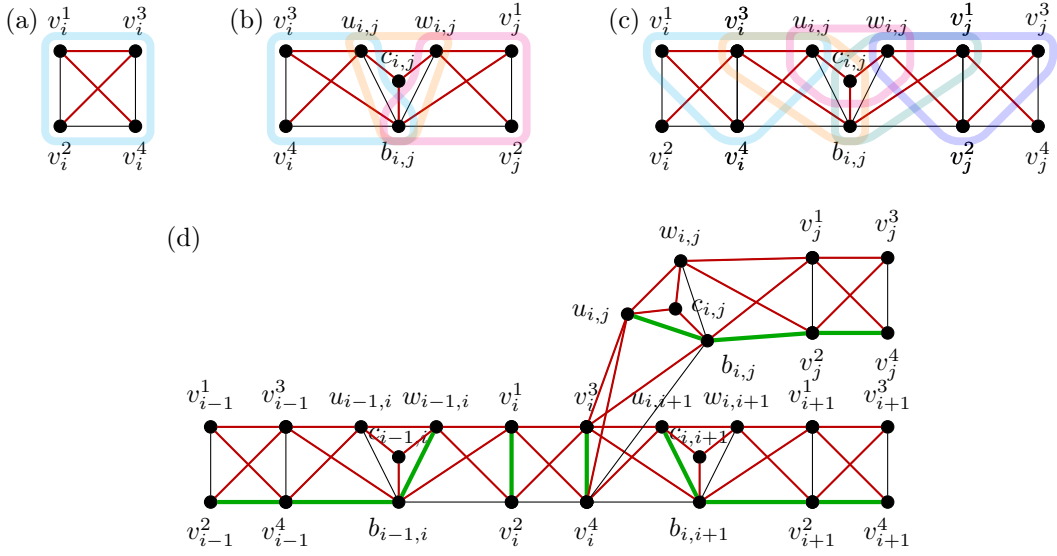
\begin{figure}[t]
 \centering
 \begin{tikzpicture}
  \def\xr{1}
  \def\yr{1}
  \tikzpreamble{}

  \newcommandx{\nodegadget}[7][6=1]{%
	\node (#1sw) at (#2*\xr-0.5*\xr,#3*\yr-0.5*\yr)[xnode, label=-#6*90:{$v_{#7}^2$}]{};
	\node (#1se) at (#2*\xr+0.5*\xr,#3*\yr-0.5*\yr)[xnode, label=-#6*90:{$v_{#7}^4$}]{};
	\node (#1nw) at (#2*\xr-0.5*\xr,#3*\yr+0.5*\yr)[xnode, label=#6*90:{$v_{#7}^1$}]{};
	\node (#1ne) at (#2*\xr+0.5*\xr,#3*\yr+0.5*\yr)[xnode, label=#6*90:{$v_{#7}^3$}]{};

	\foreach \x/\y in {%
		se/ne,
		sw/nw,
		se/sw,
		ne/nw,
		se/nw,
		sw/ne}{\draw[xedge] (#1\x) to (#1\y);}

	\foreach \x/\y in {#4}{\draw[xfedge] (#1\x) to (#1\y);}
	\foreach \x/\y in {#5}{\draw[xsedge] (#1\x) to (#1\y);}
  }

  \newcommand{\edgegadgetpur}[9]{%
    \node (#1c) at (#4*\xr,#5*\yr+0.1*\yr)[xnode,label={[label distance=-2pt]90:{$c_{#8,#9}$}}]{};
    \node (#1tll) at (#4*\xr-1.5*\yr,#5*\yr+0.5*\yr)[xnode,label=90:{$v_{#8}^3$}]{};
    \node (#1trr) at (#4*\xr+1.5*\yr,#5*\yr+0.5*\yr)[xnode,label=90:{$v_{#9}^1$}]{};
    \node (#1bl) at (#4*\xr-1.5*\yr,#5*\yr-0.5*\yr)[xnode,label=-90:{$v_{#8}^4$}]{};
    \node (#1br) at (#4*\xr+1.5*\yr,#5*\yr-0.5*\yr)[xnode,label=-90:{$v_{#9}^2$}]{};
    \node (#1b) at (#4*\xr,#5*\yr-0.5*\yr)[xnode,label=-90:{$b_{#8,#9}$}]{};
    \node (#1tl) at (#4*\xr-0.5*\yr,#5*\yr+0.5*\yr)[xnode,label=90:{$u_{#8,#9}$}]{};
    \node (#1tr) at (#4*\xr+0.5*\yr,#5*\yr+0.5*\yr)[xnode,label=90:{$w_{#8,#9}$}]{};

	\foreach \x/\y in {%
	    trr/br,
	    tll/bl,
		c/tl,
		c/b,
		c/tr,
		tr/tl,
		b/tl,
		b/tr,
		tr/trr,
		b/trr,
		tr/br,
		b/br,
		tl/tll,
		b/tll,
		tl/bl,
		b/bl}{\draw[xedge] (#1\x) to (#1\y);}
	\foreach \x/\y in {#6}{\draw[xfedge] (#1\x) to (#1\y);}
	\foreach \x/\y in {#7}{\draw[xsedge] (#1\x) to (#1\y);}
  }

  \newcommand{\edgegadget}[9]{%
    \node (#1c) at (#4*\xr,#5*\yr+0.1*\yr)[xnode,label={[label distance=-2pt]90:{$c_{#8,#9}$}}]{};
    \node (#1tll) at (#2ne)[xnode]{};
    \node (#1trr) at (#3nw)[xnode]{};
    \node (#1bl) at (#2se)[xnode]{};
    \node (#1br) at (#3sw)[xnode]{};
    \node (#1b) at (#4*\xr,#5*\yr-0.5*\yr)[xnode,label=-90:{$b_{#8,#9}$}]{};
    \node (#1tl) at (#4*\xr-0.5*\yr,#5*\yr+0.5*\yr)[xnode,label=90:{$u_{#8,#9}$}]{};
    \node (#1tr) at (#4*\xr+0.5*\yr,#5*\yr+0.5*\yr)[xnode,label=90:{$w_{#8,#9}$}]{};

		\foreach \x/\y in {%
			c/tl,c/b,c/tr,tr/tl,
			b/tl,b/tr,tr/trr,b/trr,
			tr/br,b/br,tl/tll,b/tll,tl/bl,
			b/bl}{\draw[xedge] (#1\x) to (#1\y);}
		\foreach \x/\y in {#6}{\draw[xfedge] (#1\x) to (#1\y);}
		\foreach \x/\y in {#7}{\draw[xsedge] (#1\x) to (#1\y);}
  }

  \newcommand{\edgegadgetx}[9]{%
    \node (#1c) at (#4*\xr,#5*\yr+0.1*\yr)[xnode,label=0:{$c_{#8,#9}$}]{};
    \node (#1tll) at (#2ne)[xnode]{};
    \node (#1trr) at (#3nw)[xnode]{};
    \node (#1bl) at (#2se)[xnode]{};
    \node (#1br) at (#3sw)[xnode]{};
    \node (#1b) at (#4*\xr,#5*\yr-0.5*\yr)[xnode,label=-45:{$b_{#8,#9}$}]{};
    \node (#1tl) at (#4*\xr-0.5*\yr,#5*\yr+0.5*\yr)[xnode,label=180:{$u_{#8,#9}$}]{};
    \node (#1tr) at (#4*\xr+0.5*\yr,#5*\yr+0.5*\yr)[xnode,label=90:{$w_{#8,#9}$}]{};

	\foreach \x/\y in {%
		c/tl,c/b,c/tr,tr/tl,
		b/tl,b/tr,tr/trr,b/trr,
		tr/br,b/br,tl/tll,b/tll,tl/bl,
		b/bl}{\draw[xedge] (#1\x) to (#1\y);}
	\foreach \x/\y in {#6}{\draw[xfedge] (#1\x) to (#1\y);}
	\foreach \x/\y in {#7}{\draw[xsedge] (#1\x) to (#1\y);}
  }
	\newcommandx{\xlabel}[3][1=-0.66,2=0.9]{\node at (#1*\xr,#2*\yr)[anchor=east]{(#3)};}

	\begin{scope}[xshift=-4*\xr cm]
		\xlabel{a}
		\nodegadget{A}{0}{0}{se/nw,sw/ne,ne/nw}{}{i}
		\draw[xhabA] \xnw{Anw} -- \xne{Ane} -- \xse{Ase} -- \xsw{Asw} -- cycle;
	\end{scope}

	\begin{scope}[xshift=0*\xr cm]
		\xlabel[-1.66]{b}
		\edgegadgetpur{AB}{A}{B}{0}{0}{b/c,c/tr,c/tl,tr/tl,tl/tll,tl/bl,tll/b,tr/trr,tr/br,trr/b}{}{i}{j}
		\draw[xhabA] \xnw{ABtll} -- \xne{ABtl} -- \xw{ABc} -- \xne{ABb} -- \xse{ABb} -- \xsw{ABbl} -- cycle;
		\draw[xhabB] \xnw{ABtl} -- \xne{ABtr} -- \xse{ABb} -- \xsw{ABb} -- cycle;
		\draw[xhabC] \xsw{ABb} -- \xnw{ABb} -- \xe{ABc} -- \xnw{ABtr} -- \xne{ABtrr} -- \xse{ABbr} -- cycle;
	\end{scope}

	\begin{scope}[xshift=4*\xr cm]
		\xlabel{c}
		\nodegadget{A}{0}{0}{se/nw,sw/ne,ne/nw}{}{i}
		\edgegadgetpur{AB}{A}{B}{2}{0}{b/c,c/tr,c/tl,tr/tl,tl/tll,tl/bl,tll/b,tr/trr,tr/br,trr/b}{}{i}{j}
		\nodegadget{B}{4}{0}{se/nw,sw/ne,ne/nw}{}{j}

		\draw[xhabA] \xsw{Anw} -- \xnw{Anw} -- \xne{Ane} -- \xne{ABtl} -- \xse{ABtl} -- \xse{Ase} -- \xsw{Ase} -- cycle;
		\draw[xhabB] \xsw{Ane} -- \xnw{Ane} -- \xne{ABtl} -- \xne{ABc} -- \xse{ABb} -- \xsw{ABb} -- cycle;
		\draw[xhabD] \xse{ABb} -- \xsw{ABb} -- \xnw{ABc} -- \xnw{ABtr} -- \xne{ABtrr} -- \xse{ABtrr} -- cycle;
		\draw[xhabE] \xnw{ABtr} -- \xsw{ABtr} -- \xsw{ABbr} -- \xse{ABbr} -- \xse{Bne} -- \xne{Bne} -- cycle;
		\def\teps{0.3*\xr*\yr}
		\draw[xhabC] \xsw{ABtl} -- \xnw{ABtl} -- \xne{ABtr} -- \xse{ABtr} -- \xse{ABc} -- \xsw{ABc} -- cycle;
	\end{scope}

	\begin{scope}[xshift=2*\xr cm, yshift=-5*\yr cm]
		\xlabel[-4.5][3]{d}
		\nodegadget{V}{0}{0}{se/nw,sw/ne,ne/nw}{se/ne,sw/nw}{i}
		\nodegadget{U}{-4}{0}{se/nw,sw/ne,ne/nw}{se/sw}{i-1}
		\nodegadget{W}{4}{0}{se/nw,sw/ne,ne/nw}{se/sw}{i+1}

		\nodegadget{X}{4}{2.25}{se/nw,sw/ne,ne/nw}{se/sw}[1]{j}

		\edgegadget{UV}{U}{V}{-2}{0}{b/c,c/tr,c/tl,tr/tl,tl/tll,tl/bl,tll/b,tr/trr,tr/br,trr/b}{b/tr,b/bl}{i-1}{i}
		\edgegadget{VW}{V}{W}{2}{0}{b/c,c/tr,c/tl,tr/tl,tl/tll,tl/bl,tll/b,tr/trr,tr/br,trr/b}{b/tl,b/br}{i}{i+1}

		\begin{scope}[rotate around={45:(1.75*\xr,2*\yr)}]
			\edgegadgetx{VX}{V}{X}{1.75}{2}{b/c,c/tr,c/tl,tr/tl,tl/tll,tl/bl,tll/b,tr/trr,tr/br,trr/b}{b/tl,b/br}{i}{j}
		\end{scope}
	\end{scope}
 \end{tikzpicture}
 \caption{Illustration to~\cref{constr:H4D7:NPhard}.
 (a) $V_i$ and habitat~$H_i$ (cyan).
 (b) $V_{i,j}$ with its neighbors in~$V_i$ and~$V_j$, as well as the habitats~$H_{i,j}$ (orange), $H_{i,j}'$ (cyan), and~$H_{i,j}''$ (magenta).
 (c) $V_i$, $V_j$, and $V_{i,j}$, as well as the habitats~$H_{i,j}^1$ (cyan), $H_{i,j}^2$ (orange), $H_{i,j}^3$ (orange), $H_{i,j}^4$ (teal), and $H_{i,j}^5$ (blue).
 (d) An excerpt from~$G'$ around~$V_i$,
 where green edges indicate additional edges in a solution.
 (Every red edge is contained in~$E_2'$.)}
 \label{fig:H4D7:NPhard}
\end{figure}
 Let~$V'=\bigcup_{i=1}^n V_i \cup \bigcup_{(i,j)\in X} V_{i,j}$ with~$V_i=\{v_i^1,v_i^2,v_i^3,v_i^4\}$ and~$V_{i,j} = \{u_{i,j},w_{i,j},c_{i,j},b_{i,j}\}$.
 Next,
 construct~$E'$.
 For each~$i\in\set{n}$,
 let~$V_i$ form a clique,
 i.e.,
 for each~$x,y\in V_i$,
 add the edge~$\{x,y\}$.
 Moreover,
 for each~$(i,j)\in X$,
 let each of $V_{i,j}$,
 $\{v_i^3,v_i^4,u_{i,j},b_{i,j}\}$,
 and~$\{v_j^1,v_j^2,w_{i,j},b_{i,j}\}$ form a clique.
 This finishes the construction of~$G'$.
 Denote by~$E_1'\ceq \bigcup_{i=1}^n E_i^*\cup \bigcup_{(i,j)\in X} E_{i,j}^*$,
 where $E_i^*\ceq \{\{v_i^1,v_i^2\},\{v_i^2,v_i^4\},\{v_i^3,v_i^4\}\}$
 for each~$i\in\set{n}$
 and $E_{i,j}^*\ceq \{\{v_i^4,b_{i,j}\},\{b_{i,j},v_j^2\},\{u_{i,j},b_{i,j}\},\{w_{i,j},b_{i,j}\}\}$
 for each~$(i,j)\in X$.
 Denote by~$E_2'\ceq E'\setminus E_1'$.
 We introduce this notation to distinguish between
 the edges~$E_1^*$ a solution can select from,
 and the edges~$E_2^*$ which are fully contained in every solution.
 Construct the habitat set~$\calH$ as follows.
 Add the habitat~$H_i=V_i$ for each~$i\in\set{n}$
 (intuitively,
 $H_i$ encodes whether~$v_i$ is part of the vertex cover).
 For each $(i,j)\in X$,
 add the habitats
 $H_{i,j}=V_{i,j}$,
 $H_{i,j}'=\{v_i^3,u_{i,j},b_{i,j},v_i^4\}$, and
 $H_{i,j}''=\{v_j^1,w_{i,j},b_{i,j},v_j^2\}$
 (intuitively,
 these three habitats encode that edge~$\{v_i,v_j\}$ must be covered).
 Moreover,
 for each $(i,j)\in X$,
 add the auxiliary habitats
 $H_{i,j}^1=\{v_i^1,v_i^3,u_{i,j},v_i^4\}$,
 $H_{i,j}^2=\{v_i^3,u_{i,j},c_{i,j},b_{i,j}\}$,
 $H_{i,j}^3=\{u_{i,j},c_{i,j},w_{i,j}\}$,
 $H_{i,j}^4 = \{v_j^1,w_{i,j},c_{i,j},b_{i,j}\}$,
 and $H_{i,j}^5=\{v_j^1,v_j^3,w_{i,j},v_j^2\}$.
 Set~$k=|E_2'|+2m+n+p$.
 \phantom{.}\cqed
\end{construction}
\begin{observation}
 \label{obs:H4D7:NPhard:seven}
 The maximum degree of~$G'$ is seven.
\end{observation}
\begin{proof}
 For each~$i\in\set{n}$,
 there are $(x,y),(x',y')\in X$ such that~$x=i$ and~$y'=i$,
 and~$i$ appears in exactly three tuples in~$X$.
 Thus,
 each vertex in~$V_i$ has exactly three neighbors inside~$V_i$ and at most four neighbors outside of~$V_i$.
 For each~$(i,j)\in X$,
 $u_{i,j}$ and~$w_{i,j}$ have exactly five (three inside and two outside~$V_{i,j}$)
 $c_{i,j}$ exactly three (all inside~$V_{i,j}$),
 and~$b_{i,j}$ exactly seven neighbors (three inside and four outside~$V_{i,j}$).
 Hence,
 all vertices have degree at most seven.
\end{proof}
\begin{lemma}
 \label{lem:H4D7:NPhard:forced}
 Let~$I'$ be a \yes-instance.
 Then every solution contains~$E_2'$.
\end{lemma}
\begin{proof}
 Let~$F$ be a solution and $(i,j)\in X$.
 In the graph induced by habitat
 $H_{i,j}^1=\{v_i^1,v_i^3,u_{i,j},v_i^4\}$,
 the vertices~$v_i^1$ and~$u_{i,j}$ have degree two,
 and hence edges~$\{v_i^1,v_i^3\}$, $\{v_i^3,u_{i,j}\}$. $\{u_{i,j},v_i^4\}$, and~$\{v_i^1,v_i^4\}$ are in~$F$.
 In the graph induced by the habitat
 $H_{i,j}^2=\{v_i^3,u_{i,j},c_{i,j},b_{i,j}\}$,
 the vertices~$v_i^3$ and~$c_{i,j}$ have degree two,
 and hence the edges $\{v_i^3,u_{i,j}\}$, $\{u_{i,j},c_{i,j}\}$, $\{c_{i,j},b_{i,j}\}$, and~$\{v_i^3,b_{i,j}\}$ are in~$F$.
 The graph induced by the habitat~$H_{i,j}^3=\{u_{i,j},c_{i,j},w_{i,j}\}$
 induces a triangle and hence all its edges are in~$F$.
 In the graph induced by habitat~$H_{i,j}^4=\{v_j^1,w_{i,j},c_{i,j},b_{i,j}\}$,
 the vertices~$v_j^1$ and~$c_{i,j}$ have degree two,
 and hence the edges~$\{v_j^1,w_{i,j}\}$, $\{w_{i,j},c_{i,j}\}$, $\{c_{i,j},b_{i,j}\}$, and~$\{v_j^1,b_{i,j}\}$ are in~$F$.
 In the graph induced by~$H_{i,j}^5=\{v_j^1,v_j^3,w_{i,j},v_j^2\}$,
 the vertices~$w_{i,j}$ and~$v_j^3$ have degree two,
 and hence the edges~$\{v_j^1,v_j^3\}$, $\{v_j^1,w_{i,j}\}$, $\{w_{i,j},v_j^2\}$, and~$\{v_j^3,v_j^2\}$ are in~$F$.
\end{proof}
Next,
we show that in case of a \yes-instance,
there is always a solution that matches our \cref{int:H4D7:NPhard};
the proof essentially uses edge switching arguments.

\begin{lemma}
 \label{lem:H4D7:NPhard}
 Let~$I'$ be a \yes-instance.
 Then there is a solution~$F$ where
 (i) for every~$(i,j)\in X$,
 $F$ contains exactly two edges in~$E_{i,j}^*$,
 and
 (ii) for every~$i\in\set{n}$,
 $F$ contains no two adjacent edges in~$E_i^*$.
\end{lemma}
\begin{proof}
 Suppose towards a contradiction that (i) is false,
 i.e.,
 for every solution~$F$ there is~$(i,j)\in X$ such that $F$ contains three or four edges from~$E_{i,j}^*$.
 Let~$F$ be a minimum solution such that the number of tuples~$(i,j)\in X$ with three or four edges from~$E_{i,j}^*$ being contained in~$F$ is minimal.
 Since~$G[H_{i,j},F]$ is 2-connected,
 we know that at least one the edges~$\{u_{i,j},b_{i,j}\}$ and~$\{w_{i,j},b_{i,j}\}$ is contained in~$F$.
 If both edges and~$\{v_i^4,b_{i,j}\}$ are contained,
 then $F\setminus \{\{u_{i,j},b_{i,j}\}\}$ is still a solution,
 contradicting the minimality of~$F$.
 If both edges and~$\{v_i^2,b_{i,j}\}$ are contained,
 then $F\setminus \{\{w_{i,j},b_{i,j}\}\}$ is still a solution,
 contradicting the minimality of~$F$.
 Finally,
 let $\{v_i^4,b_{i,j}\}$, $\{v_i^2,b_{i,j}\}$,
 and exactly one of~$\{u_{i,j},b_{i,j}\}$ and~$\{w_{i,j},b_{i,j}\}$ be contained in~$F$.
 If $\{u_{i,j},b_{i,j}\}\in F$,
 then~$(F\setminus \{\{v_i^4,b_{i,j}\}\}) \cup \{\{v_i^3,v_i^4\}\}$ is again a solution with less tuples~$(i,j)\in X$ with three or four edges from~$E_{i,j}^*$ being contained in~$F$,
 contradicting the choice of~$F$.
 If $\{w_{i,j},b_{i,j}\}\in F$,
 then~$(F\setminus \{\{v_i^2,b_{i,j}\}\}) \cup \{\{v_i^1,v_i^2\}\}$ is again a solution with less tuples~$(i,j)\in X$ with three or four edges from~$E_{i,j}^*$ being contained in~$F$,
 contradicting the choice of~$F$.

 Suppose towards a contradiction that given (i), (ii) is false,
 i.e.,
 for every solution~$F$ respecting (i) there is $i\in\set{n}$ such that $F$ contains two adjacent edges in~$E_i^*$.
 Let~$F$ be a minimal solution respecting (i) such that the number of~$i\in\set{n}$ with $F$ containing two adjacent edges in~$E_i^*$ is minimal.
 If~$F$ contains all three edges from~$E_i^*$,
 then~$F\setminus \{\{v_i^2,v_i^4\}\}$ is still a solution respecting (i),
 contradicting the minimality of~$F$.
 Let~$F$ contain exactly two adjacent edges from~$E_i^*$.
 Then $\{v_i^2,v_i^4\}\in F$ and one of $\{v_i^1,v_i^2\},\{v_i^3,v_i^4\}$ is in~$F$,
 and hence, $(F\setminus \{\{v_i^2,v_i^4\}\}) \cup \{\{v_i^1,v_i^2\},\{v_i^3,v_i^4\}\}$ is still a solution respecting (i);
 a contradiction to the choice of~$F$.
\end{proof}

\begin{proof}[Proof of \cref{prop:H4D7:NPhard}]
 Let~$I'$ be the instance obtained from input instance~$I=(G,p,C)$ of \prob{HCVC}
 via~\cref{constr:H4D7:NPhard} in polynomial-time.
 We prove that~$I$ is a \yes-instance if and only if~$I'$ is a \yes-instance.

 \RD{}
 Let~$W\subseteq V$ be a vertex cover of size~$p$.
 Let~$\calI_W\ceq \{i\in\set{n}\mid v_i\in W\}$ and~$\calJ_W\ceq \set{n}\setminus I_W$.
 Let~$F'\ceq \{\{v_i^2,v_i^4\}\mid i\in \calJ_W\} \cup \{\{v_i^1,v_i^2\},\{v_i^3,v_i^4\}\mid i\in \calI_W\}$
 and let~$F''\ceq \bigcup_{(i,j)\in X} F_{i,j}$,
 where
 \[ F_{i,j}\ceq \begin{cases}
									\{\{u_{i,j},b_{i,j}\},\{b_{i,j},v_j^2\}\}, &\cif{} i\in \calI_W,\\
									\{\{w_{i,j},b_{i,j}\},\{b_{i,j},v_j^4\}\}, &\cif{} i\in \calJ_W.
                \end{cases}
	\]
 Let~$F=E_2'\cup F'\cup F''$.
 Note that~$|F|=|E_2'|+2m+n+p$.
 We claim that~$F$ is a solution to~$I'$.
 For every~$(i,j)\in X$,
 we have that~$G[H_{i,j}^x,F]$ is 2-connected for every $x\in\set{5}$ since~$E_2'\subseteq F$.
 Also,
 it holds that $G[H_{i,j},F]$ is 2-connected since either~$\{u_{i,j},b_{i,j}\}$ or~$\{w_{i,j},b_{i,j}\}$ is contained in~$F_{i,j}$.
 By the definitions of~$F'$ and~$F_{i,j}$,
 we also have that $G[H_{i,j}',F]$ is 2-connected since either $\{b_{i,j},v_j^4\}$ or~$\{\{u_{i,j},b_{i,j}\},\{v_i^3,v_i^4\}\}$ is contained in~$F$.
 Moreover,
 for every~$i\in\set{n}$,
 it holds that~$G[H_i,F]$ is 2-connected.
 Suppose towards a contradiction that there is~$(i,j)\in X$ such that $G[H_{i,j}'',F]$ is not 2-connected.
 Then,
 neither~$\{b_{i,j},v_j^2\}$ nor~$\{\{w_{i,j},b_{i,j}\},\{v_j^1,v_j^2\}\}$ is contained in~$F$.
 This implies that~$i\in \calJ_W$ and that~$j\in \calJ_W$.
 Hence,
 edge~$\{v_i,v_j\}\in E$ is not covered by~$W$,
 a contradiction.

 \LD{}
 Let $F$ be a solution where for every~$(i,j)\in X$,
 $F$ contains exactly two edges in~$E_{i,j}^*$
 (let~$F''\subseteq F$ denote the set of all these edges),
 and for every~$i\in\set{n}$,
 $F$ contains no two adjacent edges in~$E_i^*$
 (which exists due to~\cref{lem:H4D7:NPhard}).
 We know that~$E_2'\subseteq F$ (see~\cref{lem:H4D7:NPhard:forced}).
 We claim that~$W=\{v_i\in V\mid |F\cap E_i^*|=2\}$ is a vertex cover of~$G$.
 Suppose towards a contradiction that
 there is an~$(i,j)\in X$ such that~$\{v_i,v_j\}\cap W=\emptyset$.
 Thus,
 $F\cap E_i^*=\{v_i^2,v_i^4\}$ and $F\cap E_j^*=\{v_j^2,v_j^4\}$.
 Since~$G[H_{i,j}',F]$ and~$G[H_{i,j}'',F]$ are 2-connected,
 $F$ contains both~$\{v_i^4,b_{i,j}\}$ and $\{v_j^2,b_{i,j}\}$.
 Since $F$ contains exactly two edges in~$E_{i,j}^*$,
 $G[H_{i,j},F]$ is not 2-connected;
 contradicting the fact that~$F$ is a solution.
 Finally,
 we have that~$k\geq |F|=|E_2'|+|F''|+|\{i\in\set{n} \mid |F\cap E_i^*|=1\}|+2\cdot |\{i\in\set{n} \mid |F\cap E_i^*|=2\}| = |E_2'| + 2m + n + |W|$.
 Hence,
 $p\geq |W|$.
\end{proof}

\subsection{$\eta=4$ and~$\Delta\leq 5$}

Only one type of habitats of size four remains in an reduced instance:
 If $I$ is reduced,
 then every habitat induces a~$K_4$.
In a~$K_4$,
every vertex has degree 3;
hence,
we get the following.

\begin{observation}
 \label{prop:H4D4:P}
 \rgbpAcr{} is polynomial-time solvable if~$\eta=4$ and~$\Delta\leq 4$.
\end{observation}
\begin{proof}
 If~$\Delta=3$,
 then for every vertex~$v\in H$,
 it holds that~$N_G(v)\subseteq H$.
 Thus,
 we can solve each habitat independently via brute force in constant time.

 Let~$\Delta=4$ and let the input instance~$I$ be reduced
 (and thus, $I$ contains only~$K_4$-habitats).
 Let~$H$ be a $K_4$-habitat such that there is another $K_4$-habitat~$H'\in\bnd(\{H\})$
 intersecting~$H$.
 If~$|H'\setminus H|=2$,
 then each of the vertices in~$H\cap H'$,
 since they have degree three in~$G[H]$,
 can only be adjacent to at most one of the vertices in $H'\setminus H$,
 yielding a contradiction to the fact that~$H'$ is a~$K_4$-habitat.
 Thus,
 let~$G[H\cap H']$ be a triangle.
 Let~$x=H\setminus H'$ and~$x'=H'\setminus H$.
 Assume towards a contradiction that there is
 a third habitat~$H''\in\bnd(\{H, H'\})$.
 Note that~$H''\setminus (H\cup H')$
 can only be adjacent with~$x,x'$ in~$G[H'']$
 since all vertices in~$H\cap H'$ have degree four.
 Moreover,
 $|H''\setminus (H\cup H')|<2$ since~$x,x'$ are of degree three in~$G[H\cup H']$.
 Thus,
 the vertex in $H''\setminus (H\cup H')$ neighbors only~$x,x'$,
 and hence~$H''$ cannot form a~$K_4$;
 a contradiction to the fact that~$I$ is reduced.
 It follows that the habitat graph consists of components of size at most~$\binom{5}{4}=5$,
 each of which can be solved via brute-force in constant time.
\end{proof}

\begin{proposition}\label{prop:H4D5:P}
 \rgbpAcr{} is polynomial-time solvable if~$\eta=4$ and~$\Delta=5$.
\end{proposition}
We first show the following.

\begin{lemma}
 \label{lem:H4D5:P:triangle}
 Let~$I=(G,c,\Fin,\calH,k)$ be reduced.
 If there is $\{H,H'\}\in E(\calG_{G,\Fin,\calH})$ such that~$G[H\cap H']$ is a triangle,
 then the union of all habitats in the component containing~$H$ and~$H'$ in~$\calG_{G,\Fin,\calH}$ is of cardinality six.
\end{lemma}
\begin{proof}
 Let~$H$ and~$H'$ intersect in~$u,v,w$ and let~$x\in H\setminus H'$ and~$x'\in H'\setminus H$.
 Note that~$\deg_{G[H\cup H']}(y)= 4$ for each~$y\in\{u,v,w\}$.
 We distinguish how a third habitat~$H''\in\bnd(\{H,H'\})$ intersects~$H\cup H'$
 (see~\cref{fig:H4D5:P}).
 \begin{figure}[t]
	\centering
	\begin{tikzpicture}
		\def\xr{1}
		\def\yr{0.8}
		\tikzpreamble{}

		\newcommandx{\nodegadget}[5][5=1]{%
		\node (sw) at (#1*\xr-0.5*\xr,#2*\yr-0.5*\yr)[xnode, label=-#5*180:{$v$}]{};
		\node (se) at (#1*\xr+0.5*\xr,#2*\yr-0.5*\yr)[xnode, label=-#5*0:{$w$}]{};
		\node (nw) at (#1*\xr-0.5*\xr,#2*\yr+0.5*\yr)[xnode, label=#5*90:{$x$}]{};
		\node (ne) at (#1*\xr+0.5*\xr,#2*\yr+0.5*\yr)[xnode, label={[label distance=-3pt]#5*45:{$u$}}]{};
		\node (nee) at (#1*\xr+1.5*\xr,#2*\yr+0.5*\yr)[xnode, label=#5*90:{$x'$}]{};

		\foreach \x/\y in {%
			se/ne,sw/nw,se/sw,ne/nw,se/nw,sw/ne,
			ne/nee,se/nee,sw/nee%
			}{\draw[xedge] (\x) to (\y);}

		\foreach \x/\y in {#3}{\draw[xfedge] (\x) to (\y);}
		\foreach \x/\y in {#4}{\draw[xsedge] (\x) to (\y);}
		}
		\newcommand{\xlabel}[1]{\node at (-0.66*\xr,1.*\yr)[anchor=east]{(#1)};}

	  \begin{scope}
		 \xlabel{a}
	   \nodegadget{0}{0}{}{}
	   \node (z) at ($(nw)!0.5!(ne)+(0,0.5*\yr)$)[xnodeA,label={[label distance=-1pt,yshift=1pt]0:{$x''$}}]{};
	   \draw[xedgeA] (sw) -- (z) -- (se);
	   \draw[xedgeA] (ne) -- (z);

		 \draw[xedgeB,dashed] (nw) -- (z) -- (nee);
		 \draw[xedgeB,dashed] (nw) to [out=125,in=55,looseness=1.2](nee);
	  \end{scope}

	  \begin{scope}[xshift=4*\xr cm]
	   \xlabel{b}
	   \nodegadget{0}{0}{}{}
	   \node (z) at ($(nw)!0.5!(ne)+(0,0.5*\yr)$)[xnodeA,label={[label distance=-1pt,yshift=1pt]0:{$x''$}}]{};
	   \draw[xedgeA] (sw) -- (z) -- (ne);
	   \draw[xedgeA] (nw) -- (z);
	  \end{scope}

	  \begin{scope}[xshift=8*\xr cm]
	   \xlabel{c}
	   \nodegadget{0}{0}{}{}
	   \node (z) at ($(nw)!0.5!(ne)+(0,0.5*\yr)$)[xnodeA,label={[label distance=-1pt,yshift=1pt]0:{$x''$}}]{};
	   \draw[xedgeA] (nee) -- (z) -- (ne);
	   \draw[xedgeA] (nw) -- (z);

	   \draw[gray] (nw) to [out=125,in=55,looseness=1.2](nee);
	  \end{scope}

	\end{tikzpicture}
	\caption{Illustration to the three cases (a), (b), and (c) in the proof of~\cref{lem:H4D5:P:triangle}.}
	\label{fig:H4D5:P}
 \end{figure}
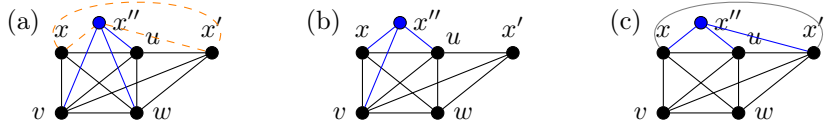
 Note that~$H''$ cannot be intersecting and have two vertices~$x'',y''\in H''\setminus (H\cup H')$ since~$H$ and~$H'$ have no adjacent vertices of degree at most three.
 Let~$x''= H''\setminus(H\cup H')$.

 \xcase{(a)}{$H\cap H''=\{u,v,w\}$}
 Then, $\deg_{G[H\cup H'\cup H'']}(y)= 5$ for each~$y\in\{u,v,w\}$.
 We claim that~$\bnd(\{H,H',H''\})=\emptyset$.
 Suppose that
 there is~$H^*\in \bnd(\{H,H',H''\})$ with~$z\in H^*\setminus (H\cup H'\cup H'')$.
 Vertex $z$ cannot be adjacent with all of~$x,x',x''$ in~$G[H^*]$,
 since either two of them are not adjacent,
 or they have degree six.
 Hence,
 $H^*$ must contain a vertex from~$\{u,v,w\}$,
 none of which is adjacent with~$z$;
 contradicting the fact that~$I$ is reduced.

 \xcase{(b)}{$H''$ contains~$x$ and two vertices from~$\{u,v,w\}$}
 \Wlog{} let~$H\cap H''=\{x,u,v\}$.
 Then,
 $\deg_{G[H\cup H'\cup H'']}(x)= 4$ and $\deg_{G[H\cup H'\cup H'']}(y)= 5$ for each~$y\in\{v,u\}$.
 Overall,
 in~$G[H\cup H'\cup H'']$ only the vertices~$x,x',w,x''$ can be of degree smaller than five,
 where~$x',x''$ are of degree at least 3,
 and~$x,w$ are of degree at least 4.
 We claim that~$\bnd(\{H,H',H''\})=\emptyset$.
 Suppose towards a contradiction that
 there is~$H^*\in \bnd(\{H,H',H''\})$.
 Note that~$|H^*\setminus (H\cup H'\cup H'')|<2$ since~$G[H^*]$ must contain an edge from~$G[H]$, $G[H']$, or $G[H'']$.
 Let~$z= H^*\setminus (H\cup H'\cup H'')$.
 Then,
 $H^*$ must contain a vertex~$a\in\{x',x''\}$ and one vertex~$b\in\{x,w\}$ where~$\{a,b\}\not \in E(G[H\cup H'\cup H''])$;
 contradicting the fact that~$I$ is reduced.

 \xcase{(c)}{$H''$ contains~$x,x'$ and one vertex from~$\{u,v,w\}$}
 In this case,
 we have~$\{x,x'\}\in E$ since otherwise~$H''$ does not form a~$K_4$.
 \Wlog{} let~$u= H''\cap \{u,v,w\}$.
 Then,
 only the vertices in~$\{v,w\}$ have degree at least four in~$G[H\cup H'\cup H'']$.
 We claim that~$\bnd(\{H,H',H''\})=\emptyset$.
 Suppose towards a contradiction that
 there is~$H^*\in \bnd(\{H,H',H''\})$.
 Note that~$|H^*\setminus (H\cup H'\cup H'')|<2$ since~$x''$ is the only vertex in $G[H\cup H'\cup H'']$ that can have degree three.
 Let~$z\in H^*\setminus (H\cup H'\cup H'')$.
 Then,
 $H^*$ must contain~$\{v,w\}$,
 each of which is then of degree two in~$G[H^*]$;
 contradicting the fact that~$I$ is reduced.
\end{proof}
\cref{lem:H4D5:P:triangle} implies that the following is correct and executable in polynomial-time.

\begin{rrule}\label{rrule:H4D5:smallcomps}
 Let $\{H,H'\}\in E(\calG_{G,\Fin,\calH})$ such that~$G[H\cap H']$ is a triangle.
 Let~$\calH'$ be the set of all habitats in the connected component in~$\calG_{G,\Fin,\calH}$ containing~$H$ and~$H'$,
 and let~$V'=\bigcup_{H''\in\calH'} H''$.
 Decide instance~$I'=(G[V'],c,\Fin\cap E(G[V']),\calH',k)$.
 If~$I'$ is a \no-instance,
 then output a trivial \no-instance.
 If $I'$ is a \yes-instance,
 then compute a minimum-cost solution~$F$ for~$I'$,
 remove~$\calH'$ from~$\calH$,
 and add~$F$ to~$\Fin$.
\end{rrule}
\begin{lemma}
 \label{lem:H4D5:cyclespaths}
 If~$I=(G,c,\Fin,\calH,k)$ is reduced and \cref{rrule:H4D5:smallcomps} is inapplicable,
 then every connected component in~$\calG_{G,\Fin,\calH}$ is a path or a cycle.
\end{lemma}
\begin{proof}
 Since $I=(G,c,\Fin,\calH,k)$ is reduced and \cref{rrule:H4D5:smallcomps} is inapplicable,
 we know that for every~$\{H,H'\}\in E(\calG_{G,\Fin,\calH})$ that~$H$ and~$H'$ intersect only in an edge.
 In this case,
 every vertex in~$H\cap H'$ has one neighbor in~$H\cap H'$,
 two neighbors in~$H\setminus H'$ and two neighbors in~$H'\setminus H$.
 Hence,
 each habitat can intersect at most two other habitats,
 and only on non-adjacent edges.
 It follows that~$\calG_{G,\Fin,\calH}$ has degree at most two,
 and hence,
 it contains only disjoint paths and cycles.
\end{proof}
\begin{proof}[Proof of \cref{prop:H4D5:P}]
	Let~$I=(G,c,\Fin,\calH,k)$ be a reduced instance such that
	\cref{rrule:H4D5:smallcomps} is inapplicable.
	Compute the \bhg{}~$\calG_{G,\Fin,\calH}$ in polynomial time.
	Let~$\calC$ denote the set of all connected components in~$\calG_{G,\Fin,\calH}$.
	Due to~\cref{lem:H4D5:cyclespaths},
	we know that each component is a path or a cycle.
	For each connected component~$C\in \calC$,
	let~$\calH_C$ denote the set of all habitats in~$C$
	and~$G_C=(\bigcup_{H\in\calH_C} H,\bigcup_{H\in\calH_C} E(G[H]))$.
	Let~$I_C=(G_C,c,\Fin\cap E(G_C),\calH_C,k)$ denote the subinstance of~$I$ restricted to~$C$.
	Due to~\cref{lem:HG:paths},
	we can compute a minimum-cost solution~$F_C$ for~$I_C$ in polynomial time
	(or report that~$I$ is a \no-instance).
	Then,
	let~$F=\Fin\cup \bigcup_{C\in\calC} F_C$.
	Output~\yes{} if~$c(F)\leq k$ (which is correct due to~\cref{obs:HG:components}),
	and \no{} otherwise.
\end{proof}

\section{Habitats of size five}
\label{sec:five}

\subsection{\(\eta=5\) and \(\Delta\geq 6\)}

\begin{proposition}
 \label{prop:H5D6:NPhard}
 \brgbpAcr{} is \NP-hard even if~$\eta=5$ and~$\Delta\geq 6$.
\end{proposition}
We give a polynomial-time many-one reduction from \tttsatAcr{} (\probref{ttts}).

\begin{intuition}
 \label{int:H5D6:NPhard}
 Each variable is represented by a vertex gadget inducing a wheel~$W_5$
 (a size-4 cycle~$C_4$ with an additional \emph{center} vertex adjacent to all cycle vertices),
 and each clause by three overlapping $W_5$-habitats sharing the same center.
 A variable gadget is connected via an additional $W_5$-habitat
 to each clause gadget in which it appears.
 In every $W_5$-habitat,
 all but two non-adjacent edges can be forced,
 so that selecting either of them suffices for 2-connectivity.
 For a variable gadget,
 this choice encodes a truth assignment.
 For a clause gadget,
 a non-selected edge corresponds to a satisfied literal,
 and ensures that the connecting $W_5$-habitat is already 2-connected
 via the choice made in the variable gadget.
\end{intuition}

\begin{construction}\label{constr:H5D6:NPhard}
 Let~$(X,\phi=\bigland_{i=1}^M C_i)$ be an instance of~\prob{(2,2)-3SAT} with~$X=\{x_1,\dots,x_N\}$.
 Construct an instance~$I'=(G,\calH,k)$ with~$G=(V,E)$ of \brgbpAcr{} as follows
 (see~\cref{fig:H5D6:NPhard} for an illustration).
 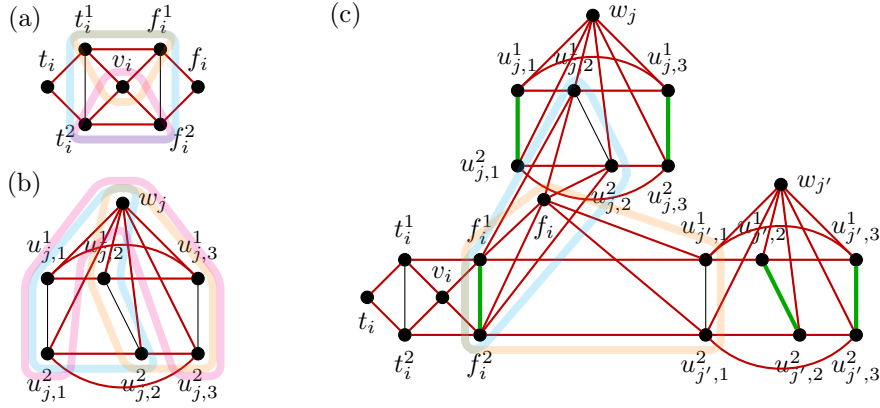
\begin{figure}[t]
  \centering
  \begin{tikzpicture}
   \def\xr{1}
   \def\yr{1}
   \tikzpreamble{}
   \def\teps{0.2*\xr*\yr}

   \newcommand{\vargadget}[6]{%
		\node (#1nw) at (#2*\xr-0.5*\xr,#3*\yr+0.5*\yr)[xnode,label=90:{$t_{#6}^1$}]{};
		\node (#1ne) at (#2*\xr+0.5*\xr,#3*\yr+0.5*\yr)[xnode,label=90:{$f_{#6}^1$}]{};
		\node (#1c) at (#2*\xr,#3*\yr)[xnode,label=90:{$v_{#6}$}]{};

		\ifstrequal{#1}{P}{%
			\node (#1sw) at (#2*\xr-0.5*\xr,#3*\yr-0.5*\yr)[xnode,label={[label distance=-3pt,yshift=3pt]-135:{$t_{#6}^2$}}]{};
			\node (#1se) at (#2*\xr+0.5*\xr,#3*\yr-0.5*\yr)[xnode,label={[label distance=-3pt,yshift=3pt]-45:{$f_{#6}^2$}}]{};
			\node (#1e) at (#2*\xr+1*\xr,#3*\yr)[xnode,label=90:{$f_{#6}$}]{};
			\node (#1w) at (#2*\xr-1*\xr,#3*\yr)[xnode,label=90:{$t_{#6}$}]{};
			\draw[xfedge] (#1sw) -- (#1w) -- (#1nw);
			\draw[xfedge] (#1se) -- (#1e) -- (#1ne);
		}{}
		\ifstrequal{#1}{A}{%
			\node (#1sw) at (#2*\xr-0.5*\xr,#3*\yr-0.5*\yr)[xnode,label=-90:{$t_{#6}^2$}]{};
		\node (#1se) at (#2*\xr+0.5*\xr,#3*\yr-0.5*\yr)[xnode,label=-90:{$f_{#6}^2$}]{};
			\node (#1w) at (#2*\xr-1*\xr,#3*\yr)[xnode,label=-90:{$t_{#6}$}]{};
			\draw[xfedge] (#1sw) -- (#1w) -- (#1nw);
		}{}

		\foreach \x/\y in {%
			se/ne,
			sw/nw,
			se/sw,
			ne/nw,
			se/c,c/nw,
			sw/c,c/ne}{\draw[xedge] (#1\x) to (#1\y);}

		\foreach \x/\y in {se/sw,nw/ne,ne/c,nw/c,se/c,sw/c}{\draw[xfedge] (#1\x) to (#1\y);}
		\foreach \x/\y in {#5}{\draw[xsedge] (#1\x) to (#1\y);}
		}

		\newcommandx{\clausegadget}[7][7=-90]{%
		\node (#1sw) at (#2*\xr-1*\xr,#3*\yr-0.5*\yr)[xnode,label=#7:{$u_{#6,1}^2$}]{};
		\node (#1se) at (#2*\xr+1*\xr,#3*\yr-0.5*\yr)[xnode,label=-90:{$u_{#6,3}^2$}]{};
		\node (#1nw) at (#2*\xr-1*\xr,#3*\yr+0.5*\yr)[xnode,label=90:{$u_{#6,1}^1$}]{};
		\node (#1ne) at (#2*\xr+1*\xr,#3*\yr+0.5*\yr)[xnode,label=90:{$u_{#6,3}^1$}]{};
		\node (#1nc) at (#2*\xr-0.25*\xr,#3*\yr+0.5*\yr)[xnode,label=90:{$u_{#6,2}^1$}]{};
		\node (#1sc) at (#2*\xr+0.25*\xr,#3*\yr-0.5*\yr)[xnode,label=-90:{$u_{#6,2}^2$}]{};
		\node (#1c) at (#2*\xr,#3*\yr+1.5*\yr)[xnode,label=0:{$w_{#6}$}]{};

		\foreach \x/\y in {%
			se/ne,
			sw/nw,
			sc/nc,
			se/sc,sc/sw,
			ne/nc,nc/nw,
			se/c,c/nw,c/nc,c/sc,
			sw/c,c/ne}{\draw[xedge] (#1\x) to (#1\y);}

		\foreach \x/\y in {%
			se/sc,sc/sw,
			ne/nc,nc/nw,
			se/c,c/nw,c/nc,c/sc,
			sw/c,c/ne}{\draw[xfedge] (#1\x) to (#1\y);}
		\draw[xfedge] (#1sw) to [out=-45,in=-135](#1se);
		\draw[xfedge] (#1nw) to [out=45,in=135](#1ne);

		\foreach \x/\y in {#4}{\draw[xfedge] (#1\x) to (#1\y);}
		\foreach \x/\y in {#5}{\draw[xsedge] (#1\x) to (#1\y);}
		}

		\newcommandx{\xlabel}[2][1=1]{\node at (-0.95*\xr,#1*1.0*\yr)[anchor=east]{(#2)};}

		\begin{scope}[yshift=2.8*\yr cm]
		 \xlabel[0.9]{a}
		 \vargadget{P}{0}{0}{}{}{i}

		 \draw[xhabA] \xnw{Pnw} -- \xne{Pne} -- \xse{Pse} -- \xsw{Psw} -- cycle;
		 \draw[xhabB] \xnw{Pnw} -- \xne{Pne} -- \xse{Pc} -- \xsw{Pc} -- cycle;
		 \draw[xhabC] \xsw{Psw} -- \xse{Pse} -- \xne{Pc} -- \xnw{Pc} -- cycle;
		\end{scope}
		\begin{scope}[yshift=-0.25*\yr cm]
		 \xlabel[1.75]{b}
		 \clausegadget{C}{0}{0}{}{}{j}

		 \draw[xhabA] \xnw{Cc} -- \xne{Cc} -- \xe{Cnc} -- \xse{Csc} -- \xsw{Csw} -- \xw{Cnw} --cycle;
		 \draw[xhabB] \xnw{Cc} -- \xne{Cc} -- \xe{Cne} -- \xse{Cse} -- \xsw{Csc} -- \xw{Cnc} --cycle;
		 \def\teps{0.3*\xr*\yr}
		 \draw[xhabC] \xnw{Cc} -- \xne{Cc} -- \xe{Cne} -- \xse{Cse} -- \xsw{Cse} -- \xs{Cc} -- \xnw{Cnc} -- \xse{Csw} -- \xsw{Csw} -- \xw{Cnw} --cycle;
		\end{scope}

		\begin{scope}[xshift=4.25*\xr cm]
		 \xlabel[3.75]{c}
		 \vargadget{A}{0}{0}{}{se/ne}{i}
		 \clausegadget{C}{2.0}{2.25}{}{sw/nw,se/ne}{j}[180]
		 \clausegadget{D}{4.5}{0}{}{sc/nc,se/ne}{j'}

		 \node (c) at (1.35*\xr,1.3*\yr)[xnode, label=-90:{$f_i$}]{};

		 \foreach \x/\y in {%
			Ane/c,Ane/Cnc,Ase/c,Ase/Csc,
			Ane/Dnw,Ase/Dsw,
			c/Cnc,c/Csc,c/Dnw,c/Dsw}{ \draw[xfedge] (\x) to (\y);}

		 \draw[xhabA] \xn{Ane} -- \xn{Cnc} -- \xe{Cnc} -- \xe{Csc} --  \xse{Csc} -- \xs{Ase} -- \xsw{Ase} -- \xw{Ane} -- cycle;
		 \draw[xhabB] \xnw{Ane} -- \xn{c} -- \xne{Dnw} -- \xse{Dsw} --  \xsw{Ase} -- cycle;
		\end{scope}
  \end{tikzpicture}
	\caption{Illustration to \cref{constr:H5D6:NPhard} (red edges are in every solution).
	(a) $V_{x_i}$ and the habitats $H_{x_i}^1$ (orange),
		$H_{x_i}^2$ (magenta),
		and
		$H_{x_i}$ (cyan).
	(b) $W_j$ and the habitats~$H_{C_j}^{-1}$ (orange),
		$H_{C_j}^{-2}$ (magenta),
		and $H_{C_j}^{-3}$ (cyan)
	(c) An excerpt around~$V_i$ according to~$x_i$'s negated appearances
		with habitats~$H_{x_i,C_j}$ (cyan) and~$H_{x_i,C_{j'}}$ (orange).
		Green edges indicate a part of a solution that corresponds to setting~$x_i$ to false.
	}
	\label{fig:H5D6:NPhard}
 \end{figure}
 First, construct the graph~$G$.
 For each~$i\in \set{N}$,
 add the vertex set~$V_i=\{v_i\}\cup \{t_i,t_i^1,t_i^2\}\cup\{f_i,f_i^1,f_i^2\}$,
 and the edge set~$E_{x_i}=E_{x_i}'\cup E_{x_i}''\cup E_{x_i}^*$ with
 $E_{x_i}'=\{\{v_i,z\}\mid z\in\{t_i^1,t_i^2,f_i^1,f_i^2\}\}\cup \{\{t_i^1,f_i^1\},\{t_i^2,f_i^2\}\}$,
 $E_{x_i}''=\{\{t_i,t_i^1\},\{t_i,t_i^2\},\{f_i,f_i^1\},\{f_i,f_i^2\}\}$,
 and edge set
 $E_{x_i}^*=\{\{t_i^1,t_i^2\},\{f_i^1,f_i^2\}\}$
 (intuitively,
 a solution selects exactly one of~$E_{x_i}^*$,
 encoding whether $x_i$ is set to true, selecting~$\{t_i^1,t_i^2\}$,
 or to false, selecting~$\{f_i^1,f_i^2\}$).
 For each~$j\in \set{M}$,
 add the vertex set~$W_j=\{w_j\}\cup U_j$ with~$U_j=\{u_{j,p}^q\mid p\in\{1,2,3\},q\in\{1,2\}\}$
 and the edge set~$E_{C_j}=E_{C_j}'\cup E_{C_j}''\cup E_{C_j}^*$ with
 $E_{C_j}'=\{\{w_j,u\}\mid u\in U_j\}$,
 $E_{C_j}''=\bigcup_{q=1}^2 \{\{u_{j,1}^q,u_{j,2}^q\},\{u_{j,2}^q,u_{j,3}^q\},\{u_{j,1}^q,u_{j,3}^q\}\}$,
 and
 $E_{C_j}^*=\{\{u_{j,p}^1, u_{j,p}^2\}\mid p\in\{1,2,3\}\}$
 (intuitively,
 a solution selects exactly two from~$E_{C_j}^*$,
 leaving out the one for which the literal is evaluated to true).
 If variable~$x_i$ appears unnegated in~$C_j$ as the~$p$-th literal,
 $p\in\{1,2,3\}$,
 add the edge set $E_{x_i,C_j}=\{\{u_{j,p}^1,t_i\},\{u_{j,p}^2,t_i\},\{u_{j,p}^1,t_i^1\},\{u_{j,p}^2,t_i^2\}\}$.
 If variable~$x_i$ appears negated in~$C_j$ as the~$p$-th literal,
 $p\in\{1,2,3\}$,
 add the edge set $E_{x_i,C_j}=\{\{u_{j,p}^1,f_i\},\{u_{j,p}^2,f_i\},\{u_{j,p}^1,f_i^1\},\{u_{j,p}^2,f_i^2\}\}$.
 Next,
 we construct the habitat set~$\calH$.
 For each~$i\in\set{N}$,
 add the habitats~$H_{x_i}^1=\{v_i,t_i^1,f_i^1\}$,
 $H_{x_i}^2=\{v_i,t_i^2,f_i^2\}$,
 and
 $H_{x_i}=H_{x_i}^1\cup H_{x_i}^2$
 (intuitively,
 $H_{x_i}$ encodes the truth assignment of~$x_i$).
 For each~$j\in\set{M}$,
 and each~$p\in\{1,2,3\}$,
 add the habitat~$H_{C_j}^{-p}=W_j\setminus \{u_{j,p}^1,u_{j,p}^2\}$
 (intuitively,
 $H_{C_j}^{-1}$, $H_{C_j}^{-2}$, and $H_{C_j}^{-3}$ together encode whether and which literal of~$C_j$ is evaluated to true),
 and for each~$p'\in\{1,2,3\}$ with~$p'>p$ and each~$\ell\in\{1,2\}$,
 add $H_{C_j,\ell}^{p,p'}=\{w_j,u_{j,p}^\ell,u_{j,p'}^\ell\}$.
 If variable~$x_i$ appears unnegated in~$C_j$ as the~$p$-th literal,
 $p\in\{1,2,3\}$,
 then add the habitats~$H_{x_i,C_j}=\{u_{j,p}^1,u_{j,p}^2,t_i,t_i^1,t_i^2\}$,
 $H_{x_i,C_j}'=\{u_{j,p}^1,t_i,t_i^1\}$,
 and $H_{x_i,C_j}''=\{u_{j,p}^2,t_i,t_i^2\}$.
 If variable~$x_i$ appears negated in~$C_j$ as the~$p$-th literal,
 $p\in\{1,2,3\}$,
 then add the habitats~$H_{x_i,C_j}=\{u_{j,p}^1,u_{j,p}^2,f_i,f_i^1,f_i^2\}$,
 $H_{x_i,C_j}'=\{u_{j,p}^1,f_i,f_i^1\}$,
 and $H_{x_i,C_j}''=\{u_{j,p}^2,f_i,f_i^2\}$.
 (Intuitively,
 $H_{x_i,C_j}$ makes sure that the truth assignment
 is correctly propagated to the clause gadget.)
 Set~$k=27N + 14M$.
 \cqed
\end{construction}
\begin{observation}
 \label{lemma:H5D6:NPhard:degree}
 The maximum degree of~$G$ is six.
\end{observation}

\begin{proof}
 For each~$i\in\set{N}$,
 $v_i$ has four neighbors in~$V_i$,
 each of~$t_i^1,t_i^2,f_i^1,f_i^2$ has four neighbors in~$V_i$ and two neighbors outside~$V_i$,
 and each of~$t_i,f_i$ has two neighbors in~$V_i$ and four neighbors outside~$V_i$.
 For each~$J\in\set{M}$,
 $w_j$ has six neighbors inside~$W_j$,
 and each~$u_{j,p}^\ell$, $p\in\{1,2,3\}$, $\ell\in\{1,2\}$,
 has four neighbors in~$W_j$ and two neighbors outside~$W_j$.
\end{proof}

\begin{lemma}\label{lemma:H5D6:NPhard}
 Let~$I'$ be \yes-instance.
 Then,
 for every solution~$F$,
 for every~$i\in\set{N}$ and for every~$j\in\set{M}$ with~$x_i$ appearing as literal in~$C_j$,
 it holds that~$E_{x_i,C_j}\subseteq F$.
 Moreover,
 for every~$i\in\set{N}$,
 it holds that $E_{x_i}'\cup E_{x_i}''\subseteq F$ and $|E_{x_i}^*\cap F|=|E_{x_i}^*|-1$,
 and for every $j\in\set{M}$,
 it holds that $E_{C_j}'\cup E_{C_j}''\subseteq F$ and $|E_{C_j}^*\cap F|=|E_{C_j}^*|-1$.
\end{lemma}

\begin{proof}
 For every~$i\in\set{N}$ and for every~$j\in\set{M}$ with~$x_i$ appearing as literal in~$C_j$,
 due to the habitats~$H_{x_i,C_j}'$
 and $H_{x_i,C_j}''$
 (each of which induces a triangle),
 it holds that $E_{x_i,C_j}\subseteq F$.
 These are $k'\ceq 16N$ edges.
 In particular,
 it also holds that~$E_{x_i}''\subseteq F$.
 For every~$i\in\set{N}$,
 due to the habitats~$H_{x_i}^1$ and~$H_{x_i}^2$
 (each of which induces a triangle),
 it holds that all edges from~$E_{x_i}'$ are in~$F$.
 Since~$G[H_{x_i},F]$ is 2-connected,
 at least one of the edges from~$E_{x_i}^*$ is in~$F$.
 For every $j\in\set{M}$,
 due to the habitats~$H_{C_j,\ell}^{p,p'}$ with $p,p'\in\{1,2,3\}$, $p'>p$, and $\ell\in\{1,2\}$,
 we know that all edges from~$E_{C_j}'\cup E_{C_j}''$ are in~$F$.
 Since $G[H_{C_j}^{-p},F]$ is 2-connected for every $p\in\{1,2,3\}$,
 at least two edges from~$E_{C_j}^*$ are in~$F$.

 Now suppose towards a contradiction that there is an~$i\in\set{N}$ such that
 $|E_{x_i}\cap F|=|E_{x_i}|$ or a $j\in\set{M}$ such that
 $|E_{C_j}\cap F|=|E_{C_j}|$.
 Then,
 we have $\sum_{i\in \set{N}} |E_{x_i}\cap F|>11N$ or $\sum_{j\in \set{M}} |E_{C_j}\cap F|>14M$,
 and hence,
 $|F|=\sum_{i\in \set{N}} |E_{x_i}\cap F| + k' + \sum_{j\in \set{M}} |E_{C_j}\cap F| > 11N + 16N + 14M = k$,
 a contradiction to~$F$ being a solution.
\end{proof}

\begin{proof}[Proof of \cref{prop:H5D6:NPhard}]
 Let~$I'$ be the instance obtained from input instance~$I=(X,\phi=\bigland_{i=1}^M C_i)$ of~\prob{(2,2)-3SAT} with~$X=\{x_1,\dots,x_N\}$
 via~\cref{constr:H5D6:NPhard}.
 We prove that~$I$ is a \yes-instance if and only if~$I'$ is a \yes-instance.

 \RD{}
 Let~$\alpha$ be a satisfying truth assignment.
 We construct a solution~$F$ as follows.
 Let~$F'\subseteq F$ be all edges forced by triangles.
 These are all but~$2N+3M$ edges.
 It follows that
 $G[H,F]$ is 2-connected for each habitat~$H$ inducing a triangle.
 Then,
 if~$\alpha(x_i)=\true$,
 then add~$\{t_i^1,t_i^2\}$ to~$F$,
 otherwise add $\{f_i^1,f_i^2\}$ to~$F$.
 By this,
 $G[H_{x_i},F]$ is 2-connected.
 These are another~$N$ edges.
 For each clause~$C_j$,
 take all edges from~$E_{C_J}^*$ into~$F$ except for the first literal that is evaluated to true in~$C_j$ by~$\alpha$,
 which exists since~$\alpha$ is satisfying.
 These are another $2M$ edges.
 By this,
 $G[H_{C_j}^{-p},F]$ for each~$j\in\set{M}$ and~$p\in\{1,2,3\}$ is 2-connected.
 Moreover,
 let~$x_i$ be the variable that corresponds to the first literal that is evaluated to true in~$C_j$.
 Then~$G[H_{x_i,C_j},F]$ is 2-connected,
 since if~$x_i$ appears unnegated, then we have~$\{t_i^1,t_i^2\}\in F$,
 and if~$x_i$ appears negated, then we have~$\{f_i^1,f_i^2\}\in F$.
 In total,
 $G[H,F]$ is 2-connected for each habitat~$H\in\calH$
 and
 $N+M=|E|-k$ edges from~$E$ are not contained in~$F$;
 thus,
 $|F|\leq k$.

 \LD{}
 Let~$F$ be a solution to~$I'$.
 Due to~\cref{lemma:H5D6:NPhard},
 we know that that there is exactly one edge from each~$E_{x_i}^*$ not in~$F$.
 If~$\{t_i^1,t_i^2\}\in F$,
 then set~$\alpha(x_i)=\true$ ,
 and if~$\{f_i^1,f_i^2\}\in F$,
 then set $\alpha(x_i)=\false$.
 Note that~$\alpha$ is well-defined.
 We claim that~$\alpha$ is satisfying.
 Suppose towards a contradiction that~$\alpha$ is not satisfying,
 and let~$C_j$ be an unsatisfied clause.
 Let~$p\in\{1,2,3\}$ such that $\{u_{j,p}^1,u_{j,p}^2\}\not\in F$.
 (which exists due to \cref{lemma:H5D6:NPhard}).
 Let~$x_i$ be the variable that appears as $p$-th literal in~$C_j$.
 Let~$x_i$ appear unnegated in~$C_j$.
 Then,
 $G[H_{x_i,C_j},F]$ is not 2-connected,
 since the edge~$\{t_i^1,t_i^2\}$ is not in~$F$ either,
 and hence~$t_i$ is a separator in~$G[H_{x_i,C_j},F]$;
 a contradiction to the fact that~$F$ is a solution.
 Let~$x_i$ appear negated in~$C_j$.
 Then,
 $G[H_{x_i,C_j},F]$ is not 2-connected,
 since the edge~$\{f_i^1,f_i^2\}$ is not in~$F$ either,
 and hence~$f_i$ is a separator in~$G[H_{x_i,C_j},F]$;
 a contradiction to the fact that~$F$ is a solution.
\end{proof}
	\subsection{\(\eta=5\) and \(\Delta\leq 4\)}
\cref{fig:H5D3:P} depicts
all pairwise non-isomorphic 2-connected size-5 habitats with degree at most 3.
In a reduced instance,
only size-5 habitats isomorphic to~(b) can appear.
Since all but~$x_1$ have degree 3,
no other habitat can intersect in an edge.
Thus,
each component of the basic habitat graph is small (cf.~\Cref{prop:H4D4:P}).

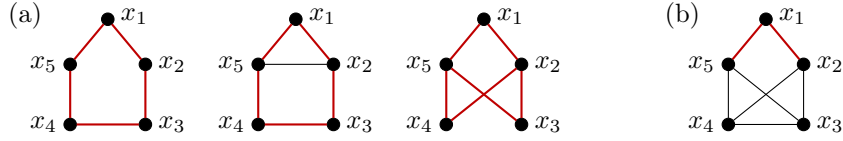
\begin{figure}[t]
\centering
\begin{tikzpicture}
	\def\xr{1}
	\def\yr{1}
	\def\xsh{2.5}
	\def\ysh{3.25}
	\tikzpreamble{}

	\newcommand{\nodegadget}[3]{%
		\node (sw) at (-0.5*\xr,-0.4*\yr)[xnode,label={[label distance=-1pt](-180):{$x_4$}}]{};
		\node (se) at (0.5*\xr,-0.4*\yr)[xnode,label={[label distance=-1pt](-0):{$x_3$}}]{};
		\node (nw) at (-0.5*\xr,0.4*\yr)[xnode,label={[label distance=-1pt](180):{$x_5$}}]{};
		\node (ne) at (0.5*\xr,0.4*\yr)[xnode,label={[label distance=-1pt](0):{$x_2$}}]{};
		\node (n) at (0,1*\yr)[xnode,label={[label distance=-1pt,yshift=1pt](0):{$x_1$}}]{};

		\foreach \x/\y in {#1}{\draw[xedge] (\x) to (\y);}

		\foreach \x/\y in {#2}{\draw[xfedge] (\x) to (\y);}
		\foreach \x/\y in {#3}{\draw[xsedge] (\x) to (\y);}
	}
	\newcommand{\xlabel}[1]{\node at (-0.75*\xr,1.05*\yr)[anchor=east]{(#1)};}

	\begin{scope}[xshift=0*\xsh*\xr cm]
		\xlabel{a}
		\nodegadget{}{n/ne,ne/se,se/sw,sw/nw,nw/n}{}
	\end{scope}
	\begin{scope}[xshift=1*\xsh*\xr cm]
		\nodegadget{nw/ne}{n/ne,n/nw,nw/sw,ne/se,se/sw}{}
	\end{scope}
	\begin{scope}[xshift=2*\xsh*\xr cm]
		\nodegadget{}{n/ne,n/nw,nw/sw,ne/se,se/nw,sw/ne}{}
	\end{scope}
	\begin{scope}[xshift=3.5*\xsh*\xr cm]
		\xlabel{b}
		\nodegadget{nw/sw,nw/se,sw/se,se/ne,ne/sw}{n/nw,n/ne}{}
	\end{scope}
\end{tikzpicture}
\caption{All pairwise non-isomorphic 2-vertex- and 2-edge-connected habitats of size five with degree at most three. Red edges are forced in every reduced instance.}
\label{fig:H5D3:P}
\end{figure}

\begin{observation}
 \label{prop:H5D3:P}
 \rgbpAcr{} is polynomial-time solvable if~$\eta=5$ and~$\Delta=3$.
\end{observation}
\begin{proof}
 Let~$I=(G,c,\Fin,\calH,k)$ be a reduced instance of~\rgbpAcr{}.
 All
 habitats from (a) do not appear,
 and hence,
 only habitats isomorphic to (b) remain.
 Let~$H\in\calH$ be isomorphic to (b).
 Suppose towards a contradiction that
 there is~$H'\in \bnd(\{H\})$ of size five.
 Let~$y\in H'\setminus H$.
 Then,
 $x_1\in H'$ and
 there is a vertex~$z\in H'\setminus H$ that is adjacent with~$x_1$.
 But then,
 $H'$ is neither 2-vertex- nor 2-edge-connected ($z$ is a separator, and~$\{x_1,z\}$ is a cut-edge);
 a contradiction to the fact that~$I$ is reduced.
 Further,
 note that not habitat isomorphic to~(b) can intersect with a habitat isomorphic to a~$K_4$.
 So,
 for each size-5 habitat~$H$,
 compute the edge set~$F_H\subseteq E(G[H])$ with~$\Fin\cap E(G[H])\subseteq F_H$ of minimum cost such that~$G[H,F_H]$ is 2-connected.
 Set~$k\ceq k-c(F_H)$,
 and remove~$\Fin\cap E(G[H])$ from~$\Fin$ and~$H$ from both~$\calH$ and~$G$.
 Then, an instance remains with only habitats of size four and degree three,
 which is polynomial-time solvable due to~\cref{prop:H4D4:P}.
\end{proof}
For~$\Delta=4$,
the proof is more involved,
consisting of an exhaustive case distinction over all possible ways a size-5 habitat can intersect with another habitat.

\begin{proposition}
	\label{prop:H5D4:P}
 \rgbpAcr{} is polynomial-time solvable if~$\eta=5$ and~$\Delta=4$.
\end{proposition}
\cref{fig:H5D4:P} shows all pairwise non-isomorphic 2-vertex- resp.\ 2-edge-connected graphs of size five and degree at most four.
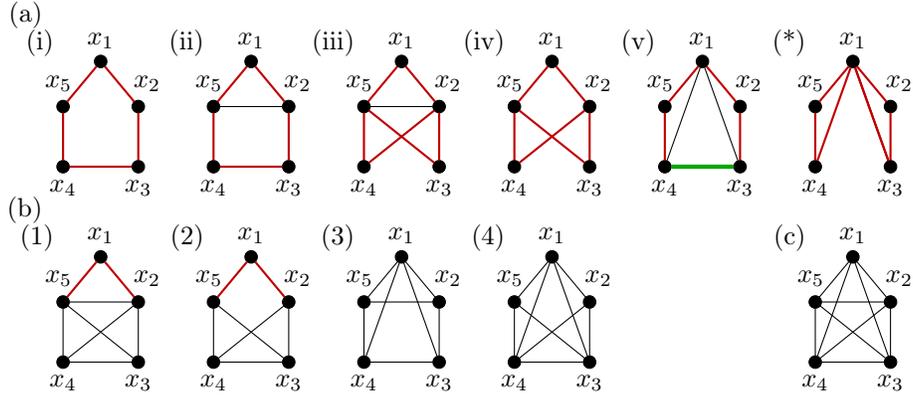
\begin{figure}[t]
 \centering
 \begin{tikzpicture}
  \def\xr{1}
  \def\yr{1}
  \def\xsh{2}
  \def\ysh{2.6}
  \tikzpreamble{}

  \newcommand{\nodegadget}[3]{%
		\node (sw) at (-0.5*\xr,-0.4*\yr)[xnode,label={[label distance=-1pt](-90):{$x_4$}}]{};
		\node (se) at (0.5*\xr,-0.4*\yr)[xnode,label={[label distance=-1pt](-90):{$x_3$}}]{};
		\node (nw) at (-0.5*\xr,0.4*\yr)[xnode,label={[label distance=0pt](90):{$x_5$~~}}]{};
		\node (ne) at (0.5*\xr,0.4*\yr)[xnode,label={[label distance=0pt](90):{~~$x_2$}}]{};
		\node (n) at (0,1*\yr)[xnode,label={[label distance=-1pt](90):{$x_1$}}]{};

		\foreach \x/\y in {#1}{\draw[xedge] (\x) to (\y);}

		\foreach \x/\y in {#2}{\draw[xfedge] (\x) to (\y);}
		\foreach \x/\y in {#3}{\draw[xsedge] (\x) to (\y);}
	 }
	 \newcommandx{\xlabel}[2][1=1]{\node at (-#1*0.5*\xr,#1*1.25*\yr)[anchor=east]{(#2)};}

	\begin{scope}[xshift=0*\xsh*\xr cm]
	 \xlabel[1.3]{a}
	 \xlabel{i}
	 \nodegadget{}{n/ne,ne/se,se/sw,sw/nw,nw/n}{}
	\end{scope}
	\begin{scope}[xshift=1*\xsh*\xr cm]
	 \xlabel{ii}
	 \nodegadget{nw/ne}{n/ne,n/nw,nw/sw,ne/se,se/sw}{}
	\end{scope}
	\begin{scope}[xshift=2*\xsh*\xr cm]
	 \xlabel{iii}
	 \nodegadget{nw/ne}{n/ne,n/nw,nw/sw,ne/se,se/nw,sw/ne}{}
	\end{scope}
	\begin{scope}[xshift=3*\xsh*\xr cm]
	 \xlabel{iv}
	 \nodegadget{}{n/ne,n/nw,nw/sw,ne/se,se/nw,sw/ne}{}
	\end{scope}
	\begin{scope}[xshift=4*\xsh*\xr cm]
	 \xlabel{v}
	 \nodegadget{n/sw,n/se}{n/nw,n/ne,ne/se,sw/nw}{se/sw}
	\end{scope}
	\begin{scope}[xshift=5*\xsh*\xr cm]
	 \xlabel{*}
	 \nodegadget{}{n/ne,n/nw,n/se,n/sw,n/se,se/ne,sw/nw}{}
	\end{scope}
	\begin{scope}[xshift=0*\xsh*\xr cm,yshift=-1*\ysh*\yr cm]
	 \xlabel[1.3]{b}
	 \xlabel{1}
	 \nodegadget{nw/ne,nw/sw,nw/se,sw/se,se/ne,ne/sw}{n/nw,n/ne}{}
	\end{scope}
	\begin{scope}[xshift=1*\xsh*\xr cm,yshift=-1*\ysh*\yr cm]
	\xlabel{2}
	 \nodegadget{nw/sw,nw/se,sw/se,se/ne,ne/sw}{n/nw,n/ne}{}
	\end{scope}
	\begin{scope}[xshift=2*\xsh*\xr cm,yshift=-1*\ysh*\yr cm]
	 \xlabel{3}
	 \nodegadget{n/ne,n/nw,n/sw,n/se,nw/ne,ne/se,se/sw,sw/nw}{}{}
	\end{scope}
	\begin{scope}[xshift=3*\xsh*\xr cm,yshift=-1*\ysh*\yr cm]
	 \xlabel{4}
	 \nodegadget{n/ne,n/nw,n/sw,n/se,ne/se,se/sw,sw/nw,se/nw,ne/sw}{}{}
	\end{scope}
	\begin{scope}[xshift=5*\xsh*\xr cm,yshift=-1*\ysh*\yr cm]
	 \xlabel{c}
	 \nodegadget{n/ne,n/nw,n/sw,n/se,nw/ne,ne/se,se/sw,sw/nw,se/nw,ne/sw}{}{}
	\end{scope}
 \end{tikzpicture}
 \caption{All pairwise non-isomorphic 2-connected habitats of size five with degree at most four. Red edges are forced. The green edge is also forced for 2-vertex-connectivity but not for 2-edge-connectivity. Habitat~$(*)$ is 2-edge-connected but not 2-vertex-connected.}
 \label{fig:H5D4:P}
\end{figure}
Assume that our instance is reduced.
Then,
none of the habitats in (a) appears.
From \cref{fig:H5D4:P}(i), (ii), and~(iv), we derive the following.
\begin{corollary}\label{obs:H5D4:P:P3free}
 There is no habitat~$H$ such that~$G[H]$ contains an induced~$P_3$ and outside the~$P_3$,
 the two endpoints have
 \begin{inparaenum}[(i)]
  \item no neighbor in common,
	\item exactly one neighbor in common, or
	\item all common neighbors being pairwise non-adjacent.
 \end{inparaenum}
\end{corollary}
From \cref{fig:H5D4:P}(ii), (iii), and~(v), we derive the following.
\begin{corollary}\label{obs:H5D4:P:K3free}
 There is no habitat~$H$ such that~$G[H]$ contains an induced~$K_3$ with exactly one vertex of degree two and outside the~$K_3$,
 the two non-degree-2 vertices in the~$K_3$ have
 \begin{inparaenum}[(i)]
  \item no neighbor in common,
	\item exactly one neighbor in common, or
	\item all common neighbors being pairwise non-adjacent.
 \end{inparaenum}
\end{corollary}
Thus,
only habitats from (b) remain.
We have the following:

\begin{observation}\label{obs:H5D4:P:notwotwo}
 If~$I$ is reduced,
 there is no habitat~$H$ such that~$G[H]$ contains at least two degree-2 vertices.
\end{observation}
\newcommand{\Gcap}{\hat{G}}
\newcommand{\Vcap}{\hat{V}}
\newcommand{\Ecap}{\hat{E}}
We will argue over the way a size-5 habitat~$H$ intersects with another habitat~$H'\in\bnd_{\calG_{G,F^*,\calH}}(\{H\})$.
Note that since~$|H|=\eta=5$,
it also holds that $H\setminus H'\neq \emptyset$.
For the habitat in (c),
we have the following.

\begin{lemma}
 \label{lem:H5D4:P:K5}
 If for a habitat~$H\in \calH$ we have that~$G[H]$ is isomorphic to a~$K_5$,
 then the union of habitats in the connected component in~$\calG_{G,F^*\calH}$ containing~$H$ is of order five.
\end{lemma}
\begin{proof}
 Each vertex in~$G[H]$ has degree four,
 and hence~$\bnd(\{H\})=\emptyset$.
\end{proof}
With the next observations,
we basically rule out possible intersection
for further investigations.

\begin{observation}\label{obs:H5D4:P:docking33}
 If in~$G[H\cup H']$,
 all but two vertices have degree four and these two have degree three,
 then the union of habitats in the connected component in~$\calG_{G,F^*\calH}$ containing~$H$ and~$H'$ is of order at most~$|H\cup H'|+1$.
\end{observation}

\begin{proof}
 Let~$v,w\in H\cup H'$ denote the two vertices of degree three in~$G[H\cup H']$.
 Let~$H''\in\bnd(\{H,H'\})$.
 Every vertex in~$H''\setminus (H\cup H')$ can only be adjacent with~$v,w$ in~$H\cup H'$,
 and since~$v,w$ are of degree three,
 we have~$|H''\setminus (H\cup H')|=1$.
 Thus,
 in~$G[H\cup H'\cup H'']$,
 all vertices are of degree four except for the vertex in~$H''\setminus (H\cup H')$.
 Hence,
 $\bnd(\{H,H',H''\})=\emptyset$.
\end{proof}

\begin{observation}\label{obs:H5D4:P:docking32}
 If in~$G[H\cup H']$,
 all but two vertices have degree four and these two have degree three and degree two and no common neighbor,
 then the union of habitats in the connected component in~$\calG_{G,F^*\calH}$ containing~$H$ and~$H'$ is of order at most $|H\cup H'|+1$.
\end{observation}

\begin{proof}
 Let~$v\in H\cup H'$ denote the vertex of degree three and~$w\in H\cup H'$ denote the vertex of degree two in~$G[H\cup H']$.
 Assume that there is~$H''\in\bnd(\{H,H'\})$.
 It holds that~$|H''\setminus (H\cup H')|<2$.
 since~$v$ is of degree three and~$v,w$ have no common neighbor in~$G[H\cup H']$.
 If $|H''\setminus (H\cup H')|=1$,
 then either~\cref{obs:H5D4:P:P3free} or~\cref{obs:H5D4:P:K3free}
 applies;
 a contradiction.
\end{proof}

In the following,
we write~$\Gcap\ceq G[H\cap H']$, $\Vcap\ceq V(\Gcap)$ and~$\Ecap\ceq E(\Gcap)$.
We first filter out several ways how~$H$ and~$H'$ cannot intersect.

\begin{observation}
 \label{obs:H5D4:P:deg1}
 If a vertex~$v\in\Vcap$ has exactly one neighbor~$w$ in~$\Gcap$,
 then edge~$\{v,w\}$ is forced.
\end{observation}

\begin{proof}
 Let~$v\in \Vcap$ such that~$w\in N_{\Gcap}(v)$ is the unique neighbor of~$v$.
 Since both~$G[H]$ and~$G[H']$ are connected,
 $v$ has a neighbor in~$H\setminus H'$ and in~$H'\setminus H$.
 It follows that~$v$ can only have one additional neighbor in~$H$ or~$H'$ and not in~$\Vcap$.
 Thus,
 $v$ has degree two in one of~$G[H]$ or~$G[H']$.
\end{proof}

For a graph~$G$,
we write~$G+xv$ for the graph~$G$ when $x$ many isolated vertices are added.
\cref{obs:H5D4:P:deg1} rules out the following graphs for~$\Gcap$.

\begin{corollary}
 $\Gcap$ is not isomorphic to~$P_2$, $P_2+v$, $P_3$, $P_3+v$, $P_2+2v$,
 or~$2 P_2$.
\end{corollary}

The next observation rules out several graphs for $\Gcap$ when~$\Vcap=4$.

\begin{observation}
 \label{obs:H5D4:P:deg11}
 If~$\Vcap=4$,
 then~$\Gcap$ contains at most one degree-1 vertices.
\end{observation}

\begin{proof}
 Suppose towards a contradiction that
 $\Gcap$ contains at least two degree-1 vertices~$x,y$.
 Each must have another neighbor in~$H\setminus H'$ and in~$H'\setminus H$.
 Since~$\Vcap=4$,
 we know that~$|H\setminus H'|=|H'\setminus H|=1$.
 Thus,
 each of~$x,y$ has degree two in each of~$G[H]$ and~$G[H']$,
 a contradiction to~\cref{obs:H5D4:P:notwotwo}.
\end{proof}

\cref{obs:H5D4:P:deg11} now excludes the following for~$\Gcap$.

\begin{corollary}
 $\Gcap$ is not isomorphic to~$P_4$, or $Claw$.
\end{corollary}

In total,
what remains for the intersections are~$K_3=C_3$,
$C_4$, $Claw+e$ (the graph with degree sequence $(3,2,2,1)$),
$K_4-e$, and~$K_4$.
We next carefully analyse each case separately.

\begin{lemma}
 If~$\Gcap$ is isomorphic to the~$K_3$,
 then the union of habitats in the connected component in~$\calG_{G,F^*\calH}$ containing~$H$ and~$H'$ is of order at most 8.
\end{lemma}

\begin{proof}
 Let~$\{u,v,w\}$ form the~$K_3$.
 Let~$\{y,z\}= H\setminus H'$ and either~$y'=H'\setminus H$ or~$\{y',z\}=H'\setminus H$.
 Since~$H$ and~$H'$ cannot have at least two degree-2 vertices,
 each vertex from~$\{u,v,w\}$ is incident to a vertex in~$H\setminus H'$ and to a vertex in~$H'\setminus H$.
 Let~$y\in H$ be the vertex incident to two vertices from~$\{u,v,w\}$ in~$G[H]$.
 We distinguish the cases for~$G[H']$.

 \xcase{1}{$y'=H'\setminus H$}
 In this case,
 $H'$ is isomorphic to a~$K_4$.
 Let~$H''\in\bnd(\{H, H'\})$.
 Beforehand,
 we show that there are no additional edges between~$y,z$ and~$y'$
 (see~\cref{fig:H5D4:P:Intersect:K3:K4}).
 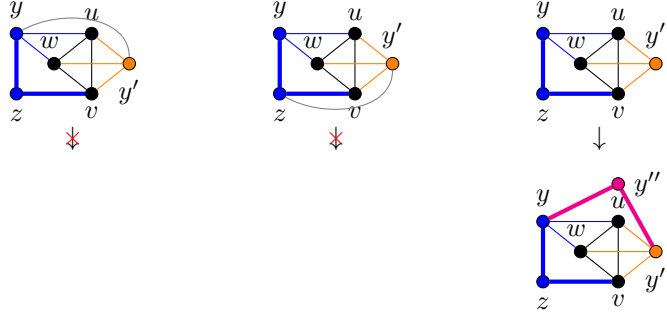
\begin{figure}[t]
	\centering
	\begin{tikzpicture}
		\def\xr{1}
		\def\yr{1}
		\def\xsh{3.5}
		\def\ysh{2.5}
		\tikzpreamble{}

		\newcommand{\nthreegadget}[3]{%
			\node (n) at (0.25*\xr,0.4*\yr)[xnode,label={[label distance=-1pt](90):{$u$}}]{};
			\node (s) at (0.25*\xr,-0.4*\yr)[xnode,label={[label distance=-1pt](-90):{$v$}}]{};
			\node (w) at (-0.25*\xr,0.0*\yr)[xnode,label={[label distance=0pt](90):{$w$~~}}]{};

			\foreach \x/\y in {#1}{\draw[xedge] (\x) to (\y);}
			\foreach \x/\y in {#2}{\draw[xfedge] (\x) to (\y);}
			\foreach \x/\y in {#3}{\draw[xsedge] (\x) to (\y);}
		}
		\newcommand{\xlabel}[1]{\node at (-0.5*\xr,1.25*\yr)[anchor=east]{(#1)};}

		\newcommand{\currH}{\nthreegadget{n/s,s/w,w/n}{}{}}

		\begin{scope}[xshift=0*\xsh*\xr cm,yshift=-1*\ysh*\yr cm]
			\currH{}
			\node (yA) at ($(w)+(-0.5,0.4)$)[xnodeA, label=90:{$y$}]{};
			\node (zA) at ($(w)+(-0.5,-0.4)$)[xnodeA, label=-90:{$z$}]{};
			\node (yB) at ($(w)+(+1,0.0)$)[xnodeB, label=-90:{$y'$}]{};

			\draw[xedgeA] (n) -- (yA);
			\draw[xedgeA,ultra thick] (yA) -- (zA) -- (s);
			\draw[xedgeB] (s) -- (yB) -- (n);
			\draw[xedgeA] (yA) -- (w);
			\draw[xedgeB] (yB) -- (w);

			\draw[gray] (yA) to [out=30,in=90](yB);
			\node at (0,-1*\yr)[]{$\downarrow$};
			\node at (0,-1*\yr)[color=red,thick]{$\times$};
		\end{scope}

		\begin{scope}[xshift=1*\xsh*\xr cm,yshift=-1*\ysh*\yr cm]
			\currH{}
			\node (yA) at ($(w)+(-0.5,0.4)$)[xnodeA, label=90:{$y$}]{};
			\node (zA) at ($(w)+(-0.5,-0.4)$)[xnodeA, label=-90:{$z$}]{};
			\node (yB) at ($(w)+(+1,0.0)$)[xnodeB, label=90:{$y'$}]{};

			\draw[xedgeA] (n) -- (yA);
			\draw[xedgeA,ultra thick] (yA) -- (zA) -- (s);
			\draw[xedgeB] (s) -- (yB) -- (n);
			\draw[xedgeA] (yA) -- (w);
			\draw[xedgeB] (yB) -- (w);
			\draw[gray] (zA) to [out=-30,in=-90](yB);
			\node at (0,-1*\yr)[]{$\downarrow$};
			\node at (0,-1*\yr)[color=red,thick]{$\times$};
		\end{scope}

		\begin{scope}[xshift=2*\xsh*\xr cm,yshift=-1*\ysh*\yr cm]
		\currH{}
		\node (yA) at ($(w)+(-0.5,0.4)$)[xnodeA, label=90:{$y$}]{};
		\node (zA) at ($(w)+(-0.5,-0.4)$)[xnodeA, label=-90:{$z$}]{};
		\node (yB) at ($(w)+(+1,0.0)$)[xnodeB, label=90:{$y'$}]{};

		\draw[xedgeA] (n) -- (yA);
		\draw[xedgeA,ultra thick] (yA) -- (zA) -- (s);
		\draw[xedgeB] (s) -- (yB) -- (n);
		\draw[xedgeA] (yA) -- (w);
		\draw[xedgeB] (yB) -- (w);
		\node at (-0.0*\xr,-1*\yr)[]{$\downarrow$};
		\end{scope}

		\begin{scope}[xshift=2*\xsh*\xr cm,yshift=-2*\ysh*\yr cm]
		\currH{}
		\node (yA) at ($(w)+(-0.5,0.4)$)[xnodeA, label=90:{$y$}]{};
		\node (zA) at ($(w)+(-0.5,-0.4)$)[xnodeA, label=-90:{$z$}]{};
		\node (yB) at ($(w)+(+1,0.0)$)[xnodeB, label=-90:{$y'$}]{};

		\draw[xedgeA] (n) -- (yA);
		\draw[xedgeA,ultra thick] (yA) -- (zA) -- (s);
		\draw[xedgeB] (s) -- (yB) -- (n);
		\draw[xedgeA] (yA) -- (w);
		\draw[xedgeB] (yB) -- (w);

		\node (z) at ($(n)+(+0.0,0.5)$)[xnodeC,label=0:{$y''$}]{};
		\draw[xedgeC,ultra thick] (yA) -- (z) -- (yB);
		\end{scope}

	\end{tikzpicture}
	\caption{$K_3$-Intersections for case 1. Blue-colored vertices and edges are exclusive for~$H$,
	and
	orange-colored vertices and edges are exclusive for~$H'$.
	Magenta-colored vertices and edges are exclusive for a third habitat~$H''$.
	Thick edges are forced.
	}
	\label{fig:H5D4:P:Intersect:K3:K4}
	\end{figure}
 If~$\{y,y'\}\in E$,
 then $H''$ can only be adjacent with~$z$.
 Let~$\{z,y'\}\in E$.
 Note that any vertex in $H''\setminus (H\cup H')$ can only be adjacent with~$y,z$,
 and that~$|H''\setminus (H\cup H')|<2$,
 since~$y,z$ are of degree three in~$G[H\cup H']$ and have no common neighbor.
 Thus,
 $|H''\setminus (H\cup H')|=1$ and~\cref{obs:H5D4:P:K3free} applies,
 a contradiction.

 Let there be no additional edges between~$y,z$ and~$y'$.
 Note that~$|H''\setminus (H\cup H')|<3$,
 since every two from~$y,z,y'$ are either non-adjacent or connected by a forced edge.
 Let~$H''\setminus (H\cup H')=\{y'',z''\}$.
 Since~$y'$ has degree three and is neither adjacent with~$y$ nor with~$z$,
 we have that~$y'',z''$ cannot be adjacent with all three of~$y,z,y'$.
 They can also not be adjacent only with~$y,y'$ since in this case,
 since the sum of their degrees would be four.
 This implies that~$y'',z''$ can only be adjacent with~$z$ and~$y'$,
 and~$z$ is adjacent with both~$y'',z''$.
 In this case,
 $y'$ and one of~$y'',z''$ are of degree two.

 Finally,
 let~$y''=H''\setminus (H\cup H')$.
 We already know that~$N_{G[H'']}(y'')=\{y,z\}$ is excluded.
 If~$N_{G[H'']}(y'')=\{y,z,y'\}$,
 then any of~$v,u,w$ must be in $H''$.
 Hence,
 $y'$ and the vertex in $\{v,u,w\}\cap H''$ are of degree two.
 Let~$N_{G[H'']}(y'')=\{y,y'\}$.
 Note that then~$H''$ must contain~$u,w$.
 Let~$H^*\in\bnd(\{H,H',H''\})$.
 Then~$H^*\setminus (H\cup H'\cup H'') $ can only be adjacent with~$z,y''$.
 If~$|H^*\setminus (H\cup H'\cup H'')|\geq 2$,
 then either~$H^*$ intersects no unforced edge if~$\{z,y''\}\not\in E$,
 or, if~$\{z,y''\}\in E$,
 $H^*$ contains two degree-two vertices among~$H^*\setminus (H\cup H'\cup H'')$.
 If~$|H^*\setminus (H\cup H'\cup H'')|=1$,
 then~\cref{obs:H5D4:P:P3free} applies (if $\{z,y''\}\not\in E$) or \cref{obs:H5D4:P:K3free} (if~$\{z,y''\}\in E$).
 Let~$N_{G[H'']}(y'')=\{z,y'\}$.
 Then,
 $y''$ is of degree two and to have~$z$ not of degree two,
 both~$y,v$ must be in~$H''$.
 But then~$y'$ is of degree two;
 a contradiction.

 \xcase{2}{$\{y',z'\}=H'\setminus H$ and~$|N_{G[H]}(y)\cap N_{G[H']}(y')|=2$}
 It follows that~$y'$ is such that~$N_{G[H]}(z)$ and~$N_{G[H']}(y')$ are disjoint.
 Let~$H''\in\bnd(\{H,H'\})$.
 Note that~$H''$ must contain a vertex from~$u,w$ since every unforced edge is incident with~$u$ or~$w$.
 First,
 we filter out how the vertices in~$\{y,z,y',z'\}$ can be adjacent
 (see~\cref{fig:H5D4:P:Intersect:K3:adj}(a)).
 Beforehand,
 we point out
 that independent of the adjacencies,
 $H''\setminus(H\cup H')$ cannot be adjacent only with $y,z$ (and symmetrically with~$y',z'$),
 since~$y$ has degree three in~$G[H\cup H']$ and~$y,z$ have no common neighbor in~$G[H\cup H']$.
 \begin{figure}[t]
	\centering
	\begin{tikzpicture}
		\def\xr{1}
		\def\yr{1}
		\def\xsh{3.25}
		\def\ysh{2}
		\tikzpreamble{}

		\newcommand{\nthreegadget}[3]{%
			\node (n) at (0.25*\xr,0.4*\yr)[xnode,label={[label distance=-1pt](90):{$u$}}]{};
			\node (s) at (0.25*\xr,-0.4*\yr)[xnode,label={[label distance=-1pt](-90):{$v$}}]{};
			\node (w) at (-0.25*\xr,0.0*\yr)[xnode,label={[label distance=0pt](90):{$w$~~}}]{};

			\foreach \x/\y in {#1}{\draw[xedge] (\x) to (\y);}
			\foreach \x/\y in {#2}{\draw[xfedge] (\x) to (\y);}
			\foreach \x/\y in {#3}{\draw[xsedge] (\x) to (\y);}
		}
		\newcommand{\xlabel}[1]{\node at (-0.5*\xr,1.125*\yr)[anchor=east]{(#1)};}

		\newcommand{\currH}{\nthreegadget{n/s,s/w,w/n}{}{}}

		\begin{scope}[xshift=0*\xsh*\xr cm,yshift=-1*\ysh*\yr cm]
			\xlabel{a}
			\currH{}
			\node (yA) at ($(w)+(-0.5,0.4)$)[xnodeA,label=90:{$y$}]{};
			\node (zA) at ($(w)+(-0.5,-0.4)$)[xnodeA,label=-90:{$z$}]{};
			\node (yB) at ($(n)+(+0.5,0.0)$)[xnodeB,label=90:{$y'$}]{};
			\node (zB) at ($(s)+(+0.5,-0.0)$)[xnodeB,label=-90:{$z'$}]{};

			\draw[xedgeA] (n) -- (yA);
			\draw[xedgeA,ultra thick] (yA) -- (zA) -- (s);
			\draw[xedgeB,ultra thick] (s) -- (zB) -- (yB);
			\draw[xedgeB] (yB) -- (n);
			\draw[xedgeA] (yA) -- (w);
			\draw[xedgeB] (yB) -- (w);

			\draw[gray] (yA) to [out=30,in=150](yB);
			\draw[gray,dashed] (zA) to [out=-30,in=-150](zB);
		\end{scope}

		\begin{scope}[xshift=1*\xsh*\xr cm,yshift=-1*\ysh*\yr cm]
			\currH{}
			\node (yA) at ($(w)+(-0.5,0.4)$)[xnodeA,label=90:{$y$}]{};
			\node (zA) at ($(w)+(-0.5,-0.4)$)[xnodeA,label=-90:{$z$}]{};
			\node (yB) at ($(n)+(+0.5,0.0)$)[xnodeB,label=90:{$y'$}]{};
			\node (zB) at ($(s)+(+0.5,-0.0)$)[xnodeB,label=-90:{$z'$}]{};

			\draw[xedgeA] (n) -- (yA);
			\draw[xedgeA,ultra thick] (yA) -- (zA) -- (s);
			\draw[xedgeB,ultra thick] (s) -- (zB) -- (yB);
			\draw[xedgeB] (yB) -- (n);
			\draw[xedgeA] (yA) -- (w);
			\draw[xedgeB] (yB) -- (w);

			\draw[gray,dashed] (zA) to [out=-30,in=-150](zB);
			\draw[gray] (zA) to [out=135,in=135,looseness=1.8](yB);
			\draw[gray] (zB) to [out=45,in=45,looseness=1.8](yA);
		\end{scope}

		\begin{scope}[xshift=0.5*\xsh*\xr cm,yshift=-2*\ysh*\yr cm]
			\currH{}
			\node (yA) at ($(w)+(-0.5,0.4)$)[xnodeA,label=90:{$y$}]{};
			\node (zA) at ($(w)+(-0.5,-0.4)$)[xnodeA,label=-90:{$z$}]{};
			\node (yB) at ($(n)+(+0.5,0.0)$)[xnodeB,label=90:{$y'$}]{};
			\node (zB) at ($(s)+(+0.5,-0.0)$)[xnodeB,label=-90:{$z'$}]{};

			\draw[xedgeA] (n) -- (yA);
			\draw[xedgeA,ultra thick] (yA) -- (zA) -- (s);
			\draw[xedgeB,ultra thick] (s) -- (zB) -- (yB);
			\draw[xedgeB] (yB) -- (n);
			\draw[xedgeA] (yA) -- (w);
			\draw[xedgeB] (yB) -- (w);

			\draw[gray] (zA) to [out=-30,in=-150](zB);
			\draw[gray] (zA) to [out=135,in=135,looseness=1.8](yB);
		\end{scope}
		\begin{scope}[xshift=2*\xsh*\xr cm,yshift=-1*\ysh*\yr cm]
		  \xlabel{b}
			\currH{}
			\node (yA) at ($(w)+(-0.5,0.4)$)[xnodeA,label=90:{$y$}]{};
			\node (zA) at ($(w)+(-0.5,-0.4)$)[xnodeA,label=-90:{$z$}]{};
			\node (yB) at ($(n)+(+0.5,0.0)$)[xnodeB,label=90:{$y'$}]{};
			\node (zB) at ($(s)+(+0.5,-0.0)$)[xnodeB,label=-90:{$z'$}]{};

			\draw[xedgeA] (n) -- (yA);
			\draw[xedgeA,ultra thick] (yA) -- (zA) -- (s);
			\draw[xedgeB] (s) -- (zB);
			\draw[xedgeB,ultra thick] (zB) -- (yB) -- (n);
			\draw[xedgeA] (yA) -- (w);
			\draw[xedgeB] (zB) -- (w);

			\draw[gray, dashed] (zA) to [out=135,in=135,looseness=1.8](yB);
			\draw[gray] (zB) to [out=45,in=45,looseness=1.8](yA);
		\end{scope}

		\begin{scope}[xshift=3*\xsh*\xr cm,yshift=-1*\ysh*\yr cm]
			\currH{}
			\node (yA) at ($(w)+(-0.5,0.4)$)[xnodeA,label=90:{$y$}]{};
			\node (zA) at ($(w)+(-0.5,-0.4)$)[xnodeA,label=-90:{$z$}]{};
			\node (yB) at ($(n)+(+0.5,0.0)$)[xnodeB,label=90:{$y'$}]{};
			\node (zB) at ($(s)+(+0.5,-0.0)$)[xnodeB,label=-90:{$z'$}]{};

			\draw[xedgeA] (n) -- (yA);
			\draw[xedgeA,ultra thick] (yA) -- (zA) -- (s);
			\draw[xedgeB] (s) -- (zB);
			\draw[xedgeB,ultra thick] (zB) -- (yB) -- (n);
			\draw[xedgeA] (yA) -- (w);
			\draw[xedgeB] (zB) -- (w);

			\draw[gray, dashed] (zA) to [out=135,in=135,looseness=1.8](yB);
			\draw[gray] (zA) to [out=-30,in=-150](zB);
			\draw[gray] (yA) to [out=30,in=150](yB);
		\end{scope}

		\begin{scope}[xshift=2.5*\xsh*\xr cm,yshift=-2*\ysh*\yr cm]
			\currH{}
			\node (yA) at ($(w)+(-0.5,0.4)$)[xnodeA,label=90:{$y$}]{};
			\node (zA) at ($(w)+(-0.5,-0.4)$)[xnodeA,label=-90:{$z$}]{};
			\node (yB) at ($(n)+(+0.5,0.0)$)[xnodeB,label=90:{$y'$}]{};
			\node (zB) at ($(s)+(+0.5,-0.0)$)[xnodeB,label=-90:{$z'$}]{};

			\draw[xedgeA] (n) -- (yA);
			\draw[xedgeA,ultra thick] (yA) -- (zA) -- (s);
			\draw[xedgeB] (s) -- (zB);
			\draw[xedgeB,ultra thick] (zB) -- (yB) -- (n);
			\draw[xedgeA] (yA) -- (w);
			\draw[xedgeB] (zB) -- (w);

			\draw[gray] (zA) to [out=135,in=135,looseness=1.8](yB);
			\draw[gray] (zA) to [out=-30,in=-150](zB);
		\end{scope}

	\end{tikzpicture}
	\caption{Impossible further adjacencies for $K_3$-Intersections for cases~2 and~3.}
	\label{fig:H5D4:P:Intersect:K3:adj}
	\end{figure}
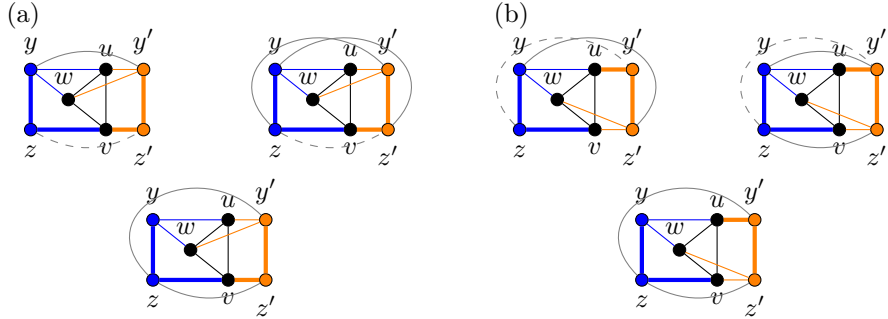
	Let~$\{y,y'\}\in E$.
	Then every vertex in $H''\setminus (H\cup H')$ can only be adjacent with~$z,z'$.
	If additionally~$\{z,z'\}\in E$,
	then~\cref{obs:H5D4:P:docking33} applies.
	If additionally~$\{z,z'\}\not\in E$,
	then,
	since~$u,w$ are not neighboring~$z,z'$,
	any in~$H''\cap\{u,v\}$ must be of degree at most one in~$H''$;
	a contradiction.
	Let~$\{z,y'\},\{y,z'\}\in E$.
	If~$\{z,z'\}\in E$,
	then all vertices have degree four in~$G[H\cup H']$,
	and if $\{z,z'\}\not\in E$,
	then~\cref{obs:H5D4:P:docking33} applies.
	Let~$\{z,y'\},\{z,z'\}\in E$.
	In this case~\cref{obs:H5D4:P:docking33} applies.

	It follows that we only consider the cases where no two of $\{y,z,y',z'\}$ are adjacent,
	$\{z,y'\}\in E$, and~$\{z,z'\}\in E$
	(see~\cref{fig:H5D4:P:Intersect:3}).
	We distinguish into these cases.
	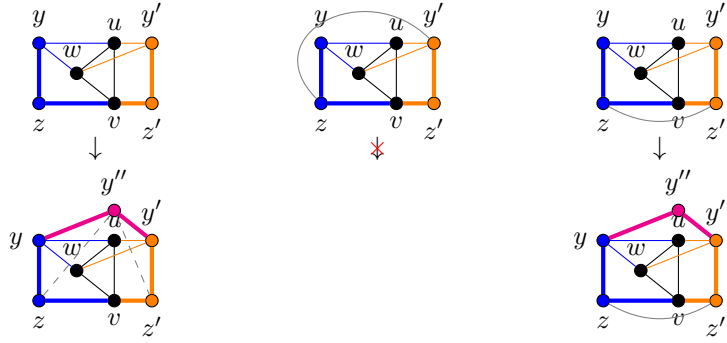
\begin{figure}[t]
	\centering
	\begin{tikzpicture}
		\def\xr{1}
		\def\yr{1}
		\def\xsh{3}
		\def\ysh{2.625}
		\tikzpreamble{}

		\newcommand{\nthreegadget}[3]{%
			\node (n) at (0.25*\xr,0.4*\yr)[xnode,label={[label distance=-1pt](90):{$u$}}]{};
			\node (s) at (0.25*\xr,-0.4*\yr)[xnode,label={[label distance=-1pt](-90):{$v$}}]{};
			\node (w) at (-0.25*\xr,0.0*\yr)[xnode,label={[label distance=0pt](90):{$w$~~}}]{};

			\foreach \x/\y in {#1}{\draw[xedge] (\x) to (\y);}
			\foreach \x/\y in {#2}{\draw[xfedge] (\x) to (\y);}
			\foreach \x/\y in {#3}{\draw[xsedge] (\x) to (\y);}
		}
		\newcommand{\xlabel}[1]{\node at (-0.5*\xr,1.25*\yr)[anchor=east]{(#1)};}

		\newcommand{\currH}{\nthreegadget{n/s,s/w,w/n}{}{}}

		\begin{scope}[xshift=0*\xsh*\xr cm,yshift=-1*\ysh*\yr cm]
			\currH{}
			\node (yA) at ($(w)+(-0.5,0.4)$)[xnodeA,label=90:{$y$}]{};
			\node (zA) at ($(w)+(-0.5,-0.4)$)[xnodeA,label=-90:{$z$}]{};
			\node (yB) at ($(n)+(+0.5,0.0)$)[xnodeB,label=90:{$y'$}]{};
			\node (zB) at ($(s)+(+0.5,-0.0)$)[xnodeB,label=-90:{$z'$}]{};

			\draw[xedgeA] (n) -- (yA);
			\draw[xedgeA,ultra thick] (yA) -- (zA) -- (s);
			\draw[xedgeB,ultra thick] (s) -- (zB) -- (yB);
			\draw[xedgeB] (yB) -- (n);
			\draw[xedgeA] (yA) -- (w);
			\draw[xedgeB] (yB) -- (w);
			\node at (0,-1*\yr)[]{$\downarrow$};
		\end{scope}

		\begin{scope}[xshift=0.0*\xsh*\xr cm,yshift=-2*\ysh*\yr cm]
			\currH{}
			\node (yA) at ($(w)+(-0.5,0.4)$)[xnodeA,label=180:{$y$}]{};
			\node (zA) at ($(w)+(-0.5,-0.4)$)[xnodeA,label=-90:{$z$}]{};
			\node (yB) at ($(n)+(+0.5,0.0)$)[xnodeB,label=90:{$y'$}]{};
			\node (zB) at ($(s)+(+0.5,-0.0)$)[xnodeB,label=-90:{$z'$}]{};

			\draw[xedgeA] (n) -- (yA);
			\draw[xedgeA,ultra thick] (yA) -- (zA) -- (s);
			\draw[xedgeB,ultra thick] (s) -- (zB) -- (yB);
			\draw[xedgeB] (yB) -- (n);
			\draw[xedgeA] (yA) -- (w);
			\draw[xedgeB] (yB) -- (w);

			\node (z) at ($(n)+(+0.0,0.4)$)[xnodeC,label=90:{$y''$}]{};
			\draw[xedgeC,ultra thick] (yA) -- (z) -- (yB);
			\draw[gray,dashed] (zA) -- (z) -- (zB);
		\end{scope}

		\begin{scope}[xshift=1.25*\xsh*\xr cm,yshift=-1*\ysh*\yr cm]
			\currH{}
			\node (yA) at ($(w)+(-0.5,0.4)$)[xnodeA,label=90:{$y$}]{};
			\node (zA) at ($(w)+(-0.5,-0.4)$)[xnodeA,label=-90:{$z$}]{};
			\node (yB) at ($(n)+(+0.5,0.0)$)[xnodeB,label=90:{$y'$}]{};
			\node (zB) at ($(s)+(+0.5,-0.0)$)[xnodeB,label=-90:{$z'$}]{};

			\draw[xedgeA] (n) -- (yA);
			\draw[xedgeA,ultra thick] (yA) -- (zA) -- (s);
			\draw[xedgeB,ultra thick] (s) -- (zB) -- (yB);
			\draw[xedgeB] (yB) -- (n);
			\draw[xedgeA] (yA) -- (w);
			\draw[xedgeB] (yB) -- (w);

			\draw[gray] (zA) to [out=135,in=135,looseness=1.8](yB);
			\node at (0,-1*\yr)[]{$\downarrow$};
			\node at (0,-1*\yr)[color=red,thick]{$\times$};
		\end{scope}

		\begin{scope}[xshift=2.5*\xsh*\xr cm,yshift=-1*\ysh*\yr cm]
			\currH{}
			\node (yA) at ($(w)+(-0.5,0.4)$)[xnodeA,label=90:{$y$}]{};
			\node (zA) at ($(w)+(-0.5,-0.4)$)[xnodeA,label=-90:{$z$}]{};
			\node (yB) at ($(n)+(+0.5,0.0)$)[xnodeB,label=90:{$y'$}]{};
			\node (zB) at ($(s)+(+0.5,-0.0)$)[xnodeB,label=-90:{$z'$}]{};

			\draw[xedgeA] (n) -- (yA);
			\draw[xedgeA,ultra thick] (yA) -- (zA) -- (s);
			\draw[xedgeB,ultra thick] (s) -- (zB) -- (yB);
			\draw[xedgeB] (yB) -- (n);
			\draw[xedgeA] (yA) -- (w);
			\draw[xedgeB] (yB) -- (w);

			\draw[gray] (zA) to [out=-30,in=-150](zB);

			\node at (0,-1*\yr)[]{$\downarrow$};
		\end{scope}

		\begin{scope}[xshift=2.5*\xsh*\xr cm,yshift=-2*\ysh*\yr cm]
			\currH{}
			\node (yA) at ($(w)+(-0.5,0.4)$)[xnodeA,label=180:{$y$}]{};
			\node (zA) at ($(w)+(-0.5,-0.4)$)[xnodeA,label=-90:{$z$}]{};
			\node (yB) at ($(n)+(+0.5,0.0)$)[xnodeB,label=90:{$y'$}]{};
			\node (zB) at ($(s)+(+0.5,-0.0)$)[xnodeB,label=-90:{$z'$}]{};

			\draw[xedgeA] (n) -- (yA);
			\draw[xedgeA,ultra thick] (yA) -- (zA) -- (s);
			\draw[xedgeB,ultra thick] (s) -- (zB) -- (yB);
			\draw[xedgeB] (yB) -- (n);
			\draw[xedgeA] (yA) -- (w);
			\draw[xedgeB] (yB) -- (w);

			\draw[gray] (zA) to [out=-30,in=-150](zB);

			\node (z) at ($(n)+(+0.0,0.4)$)[xnodeC,label=90:{$y''$}]{};
			\draw[xedgeC,ultra thick] (yA) -- (z) -- (yB);
		\end{scope}

	\end{tikzpicture}
	\caption{$K_3$-Intersections in case~2.}
	\label{fig:H5D4:P:Intersect:3}
	\end{figure}

	\xsubcase{a}{$\binom{\{y,z,y',z'\}}{2}\cap E=\emptyset$}
	Note that~$H''\setminus (H\cup H')$ cannot be adjacent with all four~$y,z,y',z'$ in~$G[H'']$,
	since each of them would have degree two.
	Assume that~$H''\setminus (H\cup H')$ is adjacent exactly with~$y,z,y'$ in~$G[H'']$
	(the case of~$y,y,z'$ is symmetric).
	Then,
	since $H''$ contains exactly one of~$u,w$,
	$z$ and~$y'$ are of degree at most two;
	a contradiction.

	Assume that~$H''\setminus (H\cup H')$ is adjacent exactly with~$y,z,z'$ in~$G[H'']$
	(the case of~$y',z,z'$ is symmetric).
	Then,
	the one vertex of~$u,v$ in~$H''$ is of degree one;
	a contradiction.

	Assume that~$H''\setminus (H\cup H')$ is adjacent exactly with~$z,z'$ in~$G[H'']$.
	We have that $|H''\setminus (H\cup H')|<2$ since none of~$u,v$ is adjacent with~$z$ or~$z'$.
	If~$|H''\setminus (H\cup H')|=1$,
	then~\cref{obs:H5D4:P:P3free} applies;
	a contradiction.

	Assume that $H''\setminus (H\cup H')$ is adjacent exactly with~$y,z'$ in~$G[H'']$
	(the case of~$y',z$ is symmetric).
	Then,
	$|H''\setminus (H\cup H')|=1$
	since~$y,z'$ have no common neighbor in~$G[H\cup H']$.
	Moreover,
	\cref{obs:H5D4:P:P3free} applies;
	a contradiction.

	Finally,
	assume that $H''\setminus (H\cup H')$ is adjacent exactly with~$y,y'$ in~$G[H'']$
	and let~$y''=H''\setminus (H\cup H')$
	(see~\cref{fig:H5D4:P:Intersect:3}).

	First,
	we show that if any of the edges~$\{z,y''\}$ or~$\{z',y''\}$
	exists,
	it is not part of a habitat~$H^\dagger$ with~$H^\dagger\subseteq H\cup H'\cup H''$
	and~$\{H^\dagger,H^\ddagger\}\in E(\calG_{G,\Fin,\calH})$ with~$H^\ddagger\in\{H,H',H''\}$.
	To see this,
	observe that each of~$z,y''$ (or $z',y''$)
	has only one common neighbor.
	Hence,
	$H^\dagger$ cannot be isomorphic to a~$K_4$,
	and has to contain five vertices.
	Since~$y''$ is of degree two,
	$H^\dagger$ must contain a vertex from~$\{u,w\}$.
	Consider the edge~$\{z,y''\}$ (the case of edge~$\{z',y''\}$ is symmetric).
	If~$y\not\in H^\dagger$,
	then both~$v\in H^\dagger$ (otherwise~$z$ has degree one)
	and~$y',z'\in H^\dagger$ (otherwise~$y''$ is of degree at most two);
	contradicting that~$H^\dagger$ has five vertices.
	Let~$y\in H^\dagger$.
	So,
	we have that~$\{y,z,y''\}$ form a triangle in which no two have a common neighbor outside the triangle;
	this can only result in~\cref{fig:H5D4:P}(ii);
	a contradiction.

	Now,
	assume that there is~$H^*\in\bnd(\{H,H',H''\})$.
	Again,
	$H^*$ must contain a vertex from~$u,v$ since every unforced edge is incident with~$u$ or~$v$.
	Note that none of $z,z',y''$ is adjacent with $u$ or~$v$.
	Hence,
	$|H^*\setminus (H\cup H'\cup H'')|<2$.
	Also,
	note that~$y^*=H^*\setminus (H\cup H'\cup H'')$
	cannot be adjacent with all three in~$G[H^*]$,
	since then~$\{u,v\}\cap H^*$ is isolated.
	If~$y^*$ is only adjacent with two,
	then \cref{obs:H5D4:P:P3free} or \cref{obs:H5D4:P:K3free} applies;
	a contradiction.

	\xsubcase{b}{$\binom{\{y,z,y',z'\}}{2}\cap E=\{z,y'\}$}
	First,
	observe that the edge $\{z,y'\}$
	is not an unforced edge for a habitat~$H^\dagger$ with~$H^\dagger\subseteq H\cup H'$
	and~$\{H^\dagger,H^\ddagger\}\in E(\calG_{G,\Fin,\calH})$ with~$H^\ddagger\in\{H,H'\}$.
	To see this,
	observe that to not be forced in~$G[H^\dagger]$,
	both~$y,v$ must be neighboring~$z$ in~$G[H^\dagger]$,
	and at least two from~$\{u,w,z'\}$ must be neighboring~$y'$.
	This is not possible since the maximum habitat size is five.
	Thus,
	$H''$ still must contain a vertex from~$u,w$.

	Assume that~$H''\setminus (H\cup H')$ is adjacent exactly with~$y,z,z'$ in~$G[H'']$.
	Then,
	$|H''\setminus (H\cup H')|=1$ and $H''$ contains exactly one of~$u,w$,
	which is then of degree one;
	a contradiction.

	Assume that~$H''\setminus (H\cup H')$ is adjacent exactly with~$z,z'$ in~$G[H'']$.
	Note that $|H''\setminus (H\cup H')|<2$,
	since none of~$z,z'$ is adjacent with any of~$u,w$.
	If~$|H''\setminus (H\cup H')|=1$,
	then \cref{obs:H5D4:P:P3free} applies;
	a contradiction.

	\xsubcase{c}{$\binom{\{y,z,y',z'\}}{2}\cap E=\{z,z'\}$}
	First,
	observe that the edge $\{z,z'\}$
	is not an unforced edge of a habitat~$H^\dagger$ with~$H^\dagger\subseteq H\cup H'$
	and~$\{H^\dagger,H^\ddagger\}\in E(\calG_{G,\Fin,\calH})$ with~$H^\ddagger\in\{H,H'\}$.
	To see this,
	observe that to not be forced in~$H^\dagger$,
	both~$z$ must be neighboring~$v$ and~$y$ in~$G[H^\dagger]$,
	and~$z'$ must be neighboring~$v$ and~$y'$ in~$G[H^\dagger]$.
	Thus, $H^\dagger=\{z,z',v,y,y'\}$,
	where both~$y,y'$ are of degree one in~$G[H^\dagger]$.
	Thus,
	$H''$ still must contain a vertex from~$u,w$.

	Now,
	note that~$|H''\setminus (H\cup H')|=1$ since all vertices have degree at least three in~$G[H\cup H']$.
	Thus,
	let $y''=H''\setminus (H\cup H')$.
	Note that~$y''$ cannot be adjacent with~$y,z,y',z'$ in~$G[H'']$,
	since~$H''$ contains one vertex from~$u,w$.

	Assume that $y''$ is adjacent exactly with~$y,z,y'$ in~$G[H'']$
	(the case of~$y,z',y'$ is symmetric).
	Then,
	$z$ and the vertex in~$\{u,w\}\cap H''$ are of degree two;
	a contradiction.

	Assume that $y''$ is adjacent exactly with~$y,z,z'$ in~$G[H'']$
	(the case of~$y',z,z'$ is symmetric).
	the vertex in~$\{u,w\}\cap H''$ is of degree one;
	a contradiction.

	Assume that~$y''$ is adjacent exactly with~$z,z'$ or with~$y,z'$ in~$G[H'']$ .
	Then~\cref{obs:H5D4:P:K3free} or \cref{obs:H5D4:P:P3free} applies.

	Assume that~$y''$ is adjacent exactly with~$y,y'$ in~$G[H'']$
	(see~\cref{fig:H5D4:P:Intersect:3}).
	Let~$H^*\in\bnd(\{H, H',H''\})$.
	Since~$y''$ is of degree two,
	$H^*$ must contain a vertex from~$u,w$.
	Note that~$|H^*\setminus (H\cup H'\cup H'')|<2$,
	since none of~$z,z',y''$ is adjacent with any of~$u,w$.
	Let~$y^*=H^*\setminus (H\cup H'\cup H'')$.
	If $y^*$ is adjacent with exactly two vertices in~$G[H^*]$,
	then \cref{obs:H5D4:P:P3free} or \cref{obs:H5D4:P:K3free} applies.
	If $y^*$ is adjacent with exactly three vertices in~$G[H^*]$,
	then the vertex in~$\{u,w\}\cap H^*$ is isolated;
	a contradiction.

	\xcase{3}{$\{y',z'\}=H'\setminus H$ and~$|N_{G[H]}(y)\cap N_{G[H']}(y')|=|N_{G[H]}(y)\cap N_{G[H']}(z')|=1$}
	Let~$y'$ be such that~$N_{G[H]}(z)$ and~$N_{G[H']}(y')$ are disjoint.
	Let~$H''\in\bnd(\{H,H'\})$.
	First,
	we filter out how the vertices in~$\{y,z,y',z'\}$ can be adjacent
	(see~\cref{fig:H5D4:P:Intersect:K3:adj}(b)).
	Let~$\{y,z'\}\in E$.
	If~$\{z,y'\}\in E$,
	then~\cref{obs:H5D4:P:docking33} applies.
	If~$\{z,y'\}\not\in E$,
	since~$z,y'$ are not adjacent and have no common neighbor in~$G[H\cup H']$,
	both have degree at most two in~$G[H'']$;
	a contradiction.
	Let~$\{y,y'\},\{z,z'\}\in E$.
	If~$\{z,y'\}\in E$,
	then all vertices have degree four.
	If~$\{z,y'\}\not\in E$,
	then~\cref{obs:H5D4:P:docking33} applies.
	Let~$\{z,y'\},\{z,z'\}\in E$
	(the case of~$\{y,y'\},\{y',z'\}\in E$ is symmetric).
	If~$\{y,y'\}\in E$,
	then all vertices have degree four.
	If~$\{y,y'\}\not\in E$,
	then~\cref{obs:H5D4:P:docking33} applies.

	It follows that we only consider the cases where no two of $\{y,z,y',z'\}$ is adjacent,
	$\{z,z'\}\in E$, or~$\{z,y'\}\in E$.
	(see~\cref{fig:H5D4:P:Intersect:2}).
	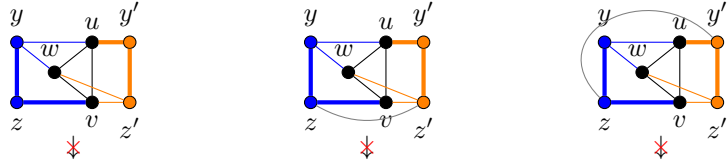
\begin{figure}[t]
	\centering
	\begin{tikzpicture}
		\def\xr{1}
		\def\yr{1}
		\def\xsh{3.125}
		\def\ysh{2.5}
		\tikzpreamble{}

		\newcommand{\nthreegadget}[3]{%
			\node (n) at (0.25*\xr,0.4*\yr)[xnode,label={[label distance=-1pt](90):{$u$}}]{};
			\node (s) at (0.25*\xr,-0.4*\yr)[xnode,label={[label distance=-1pt](-90):{$v$}}]{};
			\node (w) at (-0.25*\xr,0.0*\yr)[xnode,label={[label distance=0pt](90):{$w$~~}}]{};

			\foreach \x/\y in {#1}{\draw[xedge] (\x) to (\y);}
			\foreach \x/\y in {#2}{\draw[xfedge] (\x) to (\y);}
			\foreach \x/\y in {#3}{\draw[xsedge] (\x) to (\y);}
		}
		\newcommand{\xlabel}[1]{\node at (-0.5*\xr,1.25*\yr)[anchor=east]{(#1)};}

		\newcommand{\currH}{\nthreegadget{n/s,s/w,w/n}{}{}}

		\begin{scope}[xshift=0*\xsh*\xr cm,yshift=-1*\ysh*\yr cm]
			\currH{}
			\node (yA) at ($(w)+(-0.5,0.4)$)[xnodeA,label=90:{$y$}]{};
			\node (zA) at ($(w)+(-0.5,-0.4)$)[xnodeA,label=-90:{$z$}]{};
			\node (yB) at ($(n)+(+0.5,0.0)$)[xnodeB,label=90:{$y'$}]{};
			\node (zB) at ($(s)+(+0.5,-0.0)$)[xnodeB,label=-90:{$z'$}]{};

			\draw[xedgeA] (n) -- (yA);
			\draw[xedgeA,ultra thick] (yA) -- (zA) -- (s);
			\draw[xedgeB] (s) -- (zB);
			\draw[xedgeB,ultra thick] (zB) -- (yB) -- (n);
			\draw[xedgeA] (yA) -- (w);
			\draw[xedgeB] (zB) -- (w);
			\node at (0,-1*\yr)[]{$\downarrow$};
			\node at (0,-1*\yr)[color=red,thick]{$\times$};
		\end{scope}

		\begin{scope}[xshift=1.25*\xsh*\xr cm,yshift=-1*\ysh*\yr cm]
			\currH{}
			\node (yA) at ($(w)+(-0.5,0.4)$)[xnodeA,label=90:{$y$}]{};
			\node (zA) at ($(w)+(-0.5,-0.4)$)[xnodeA,label=-90:{$z$}]{};
			\node (yB) at ($(n)+(+0.5,0.0)$)[xnodeB,label=90:{$y'$}]{};
			\node (zB) at ($(s)+(+0.5,-0.0)$)[xnodeB,label=-90:{$z'$}]{};

			\draw[xedgeA] (n) -- (yA);
			\draw[xedgeA,ultra thick] (yA) -- (zA) -- (s);
			\draw[xedgeB] (s) -- (zB);
			\draw[xedgeB,ultra thick] (zB) -- (yB) -- (n);
			\draw[xedgeA] (yA) -- (w);
			\draw[xedgeB] (zB) -- (w);

			\draw[gray] (zA) to [out=-30,in=-150](zB);
			\node at (0,-1*\yr)[]{$\downarrow$};
			\node at (0,-1*\yr)[color=red,thick]{$\times$};
		\end{scope}

		\begin{scope}[xshift=2.5*\xsh*\xr cm,yshift=-1*\ysh*\yr cm]
			\currH{}
			\node (yA) at ($(w)+(-0.5,0.4)$)[xnodeA,label=90:{$y$}]{};
			\node (zA) at ($(w)+(-0.5,-0.4)$)[xnodeA,label=-90:{$z$}]{};
			\node (yB) at ($(n)+(+0.5,0.0)$)[xnodeB,label=90:{$y'$}]{};
			\node (zB) at ($(s)+(+0.5,-0.0)$)[xnodeB,label=-90:{$z'$}]{};

			\draw[xedgeA] (n) -- (yA);
			\draw[xedgeA,ultra thick] (yA) -- (zA) -- (s);
			\draw[xedgeB] (s) -- (zB);
			\draw[xedgeB,ultra thick] (zB) -- (yB) -- (n);
			\draw[xedgeA] (yA) -- (w);
			\draw[xedgeB] (zB) -- (w);

			\draw[gray] (zA) to [out=135,in=135,looseness=1.8](yB);
			\node at (0,-1*\yr)[]{$\downarrow$};
			\node at (0,-1*\yr)[color=red,thick]{$\times$};
		\end{scope}

	\end{tikzpicture}
	\caption{$K_3$-Intersections for case~3.}
	\label{fig:H5D4:P:Intersect:2}
	\end{figure}
	We distinguish into these cases.

	\xsubcase{a}{$\binom{\{y,z,y',z'\}}{2}\cap E=\emptyset$}
	Note that~$H''$ contains one of~$u,v,w$ since every unforced edge is incident with~$u$, $v$, or~$w$.
	Thus,
	$H''\setminus (H\cup H')$ cannot be adjacent with all four of~$\{y,z,y',z'\}$ in~$G[H'']$
	since $\{u,v,w\}\cap H''\neq \emptyset$.
    Let $H''\setminus (H\cup H')$ be ajdacent with exactly three of~$\{y,z,y',z'\}$ in~$G[H'']$.
    Again,
    since $\{u,v,w\}\cap H''\neq \emptyset$,
    we have that $|H''\setminus (H\cup H')|=1$.
	Assume that~$H''\setminus (H\cup H')$ is adjacent with~$y,z,y'$
	(the case of~$y,z',y'$ is symmetric).
	It follows that~$u\in H''$,
	since otherwise there is a degree-1 vertex in~$G[H'']$.
	Thus,
	each of~$z,u,y'$ is of degree two.

	Assume that~$H''\setminus (H\cup H')$ is adjacent only with~$y,z,z'$ in~$G[H'']$
	(the case of~$y',z,z'$ is symmetric).
	Then~$v$ or~$w$ is in~$H''$,
	and either way, this vertex and~$z'$ are of degree two.

	Assume that~$H''\setminus (H\cup H')$ is adjacent only with~$y,y'$ in~$G[H'']$
	(the case of~$z,z'$ is symmetric).
	He have~$|H''\setminus (H\cup H')|=1$,
	since $y$ and~$y'$ are not adjacent and~$y'$ of degree three in~$G[H\cup H']$.
	Hence,
	\cref{obs:H5D4:P:P3free} applies.

	Assume that~$H''\setminus (H\cup H')$ is adjacent only with~$y,z'$ in~$G[H'']$
	(the case of~$z,y'$ is symmetric).
	He have~$|H''\setminus (H\cup H')|=1$,
	since both $y$ and~$z'$ are of degree three in~$G[H\cup H']$.
	Then,
	\cref{obs:H5D4:P:P3free} applies.

	\xsubcase{b}{$\binom{\{y,z,y',z'\}}{2}\cap E=\{z,z'\}$}
	Assume that~$H''\setminus (H\cup H')$ is adjacent with~$y,z,y'$ in~$G[H'']$.
	If~$|H''\setminus (H\cup H')|=2$,
	then all of~$y,z,y'$ are of degree at most two in~$G[H'']$.
	If~$|H''\setminus (H\cup H')|=1$
	then~$y'$ and one of~$y,z$ is of degree two in~$G[H'']$.

	Assume that~$H''\setminus (H\cup H')$ is adjacent with~$z,y'$ in~$G[H'']$.
	Note that $|H''\setminus (H\cup H')|<2$
	since~$z,y'$ are not adjacent and~$y$ is of degree three in~$G[H\cup H']$.
	If $|H''\setminus (H\cup H')|=1$,
	then \cref{obs:H5D4:P:P3free} applies.

	Assume that~$H''\setminus (H\cup H')$ is adjacent with~$y,y'$ in~$G[H'']$.
	Note that $|H''\setminus (H\cup H')|<2$
	since~$y,y'$ are not adjacent and~$y$ is of degree three.
	If $|H''\setminus (H\cup H')|=1$,
	then \cref{obs:H5D4:P:P3free} applies.

	\xsubcase{c}{$\binom{\{y,z,y',z'\}}{2}\cap E=\{z,y'\}$}
	First,
	observe that the edge $\{z,y'\}$
	is not unforced in a habitat~$H^\dagger$ with~$H^\dagger\subseteq H\cup H'$
	and~$\{H^\dagger,H^\ddagger\}\in E(\calG_{G,\Fin,\calH})$ with~$H^\ddagger\in\{H,H'\}$.
	To see this,
	observe that to not be forced in~$H^\dagger$,
	$z$ must be neighboring~$v$ and~$y$ in~$G[H^\dagger]$,
	and~$y'$ must be neighboring~$u$ and~$y'$ in~$G[H^\dagger]$;
	a contradiction to~$H^\dagger$ being of size at most five.
	Thus,
	$H''$ contains one of~$u,v,w$ since every unforced edge is incident with~$u$, $v$, or~$w$.

	Since all vertices are of degree at least three
	and either are not adjacent or have no common neighbor in~$G[H\cup H']$,
	it holds that~$|H''\setminus (H\cup H')|=1$.
	Let~$y''=H''\setminus (H\cup H')$.
	Note that~$y''$ cannot be adjacent with all four in~$G[H'']$
	since~$\{u,v,w\}\cap H''\neq\emptyset$.
	Assume that~$y''$ is adjacent with~$y,z,y'$ in~$G[H'']$
	(the case of~$z',z,y'$ is symmetric).
	It follows that~$u\in H''$,
	since otherwise $y,y'$ are of degree two or there is a degree-1 vertex in~$G[H'']$.
	Thus,
	$u$ is of degree two in~$G[H'']$,
	and hence,
	no intersecting edge is unforced.

	Assume that~$y''$ is adjacent with~$y,z,z'$ in~$G[H'']$
	(the case of~$y,y',z'$ is symmetric).
	He have that either~$v$ or $w$ is in~$H''$,
	either of which and~$z'$ are of degree two.

	Assume that~$y''$ is adjacent with~$y,y'$
	(the case of~$z,z'$ is symmteric)
	or with~$y,z'$ in~$G[H'']$.
	Then,
	\cref{obs:H5D4:P:P3free} applies.

	Assume that~$y''$ is adjacent with~$z,y'$ or with~$y,z'$ in~$G[H'']$.
	Then,
	\cref{obs:H5D4:P:K3free} applies.

	This exhausts all possible adjacencies for $K_3$-intersections,
	hence no union of habitats in a connected component can exceed size 8.
\end{proof}

\begin{lemma}
 \label{lem:H5D4:P:C4}
 If~$\Gcap$ is isomorphic to the~$C_4$,
 then the union of habitats in the connected component in~$\calG_{G,F^*\calH}$ containing~$H$ and~$H'$ is of order at most 7.
\end{lemma}

\begin{proof}
 Let~$(u,v,w,x)$ form the~$C_4$.
 Let~$y= H\setminus H'$ and~$y'=H'\setminus H$.
 Since~$H$ and~$H'$ cannot have at least two vertices of degree-2,
 we know that each of~$y,y'$ has at least three neighbors in~$\Vcap$.
 We distinguish the following cases
 (see~\cref{fig:H5D4:P:Intersect:C4})
	\begin{figure}[t]
	\centering
	\begin{tikzpicture}
		\def\xr{1}
		\def\yr{1}
		\def\xsh{2.5}
		\def\ysh{2.5}
		\tikzpreamble{}

		\newcommand{\nfourgadget}[3]{%
			\node (ne) at (0.25*\xr,0.4*\yr)[xnode,label={[label distance=-1pt](90):{$u$}}]{};
			\node (se) at (0.25*\xr,-0.4*\yr)[xnode,label={[label distance=-1pt](-90):{$v$}}]{};
			\node (sw) at (-0.25*\xr,-0.4*\yr)[xnode,label={[label distance=0pt](-90):{$w$~~}}]{};
			\node (nw) at (-0.25*\xr,0.4*\yr)[xnode,label={[label distance=0pt](90):{$x$~~}}]{};

			\foreach \x/\y in {#1}{\draw[xedge] (\x) to (\y);}
			\foreach \x/\y in {#2}{\draw[xedge,ultra thick] (\x) to (\y);}
			\foreach \x/\y in {#3}{\draw[xsedge] (\x) to (\y);}
		}
		\newcommand{\xlabel}[1]{\node at (-0.5*\xr,1.25*\yr)[anchor=east]{(#1)};}

		\newcommand{\currH}{\nfourgadget{ne/se,se/sw,sw/nw,nw/ne}{}{}}

		\begin{scope}[xshift=-1.0*\xsh*\xr cm,yshift=-1*\ysh*\yr cm]
			\currH{}
			\node (yA) at ($(nw)!0.5!(sw)+(-0.5,-0.0)$)[xnodeA,label=90:{$y$}]{};
			\node (yB) at ($(ne)!0.5!(se)+(+0.5,-0.0)$)[xnodeB,label=90:{$y'$}]{};

			\foreach \x in {sw,se,nw,ne}{
				\draw[xedgeA] (yA) -- (\x);
				\draw[xedgeB] (yB) -- (\x);
			}
			\node at (0,-1*\yr)[]{$\downarrow$};
			\node at (0,-1*\yr)[color=red,thick]{$\times$};
		\end{scope}

		\begin{scope}[xshift=-0*\xsh*\xr cm,yshift=-1*\ysh*\yr cm]
			\nfourgadget{se/sw,sw/nw}{ne/se,nw/ne}{}
			\node (yA) at ($(nw)!0.5!(sw)+(-0.5,-0.0)$)[xnodeA,label=90:{$y$}]{};
			\node (yB) at ($(ne)!0.5!(se)+(+0.5,-0.0)$)[xnodeB,label=90:{$y'$}]{};

			\foreach \x in {sw,se,nw}{
				\draw[xedgeA] (yA) -- (\x);
				\draw[xedgeB] (yB) -- (\x);
			}
			\draw[gray,dotted] (yA) -- (yB);
			\node at (0*\xr,-1*\yr)[]{$\downarrow$};
		\end{scope}

		\begin{scope}[xshift=-0*\xsh*\xr cm,yshift=-2*\ysh*\yr cm]
			\nfourgadget{se/sw,sw/nw}{ne/se,nw/ne}{}
			\node (yA) at ($(nw)!0.5!(sw)+(-0.5,-0.0)$)[xnodeA,label=90:{$y$}]{};
			\node (yB) at ($(ne)!0.5!(se)+(+0.5,-0.0)$)[xnodeB,label=90:{$y'$}]{};

			\foreach \x in {sw,se,nw}{
				\draw[xedgeA] (yA) -- (\x);
				\draw[xedgeB] (yB) -- (\x);
			}
			\node (z) at ($(se)!0.5!(sw)+(+0.0,-0.6)$)[xnodeC]{};
			\draw[xedgeC,ultra thick] (yA) -- (z) -- (yB);
		\end{scope}

		\begin{scope}[xshift=1.5*\xsh*\xr cm,yshift=-1*\ysh*\yr cm]
			\nfourgadget{se/sw}{ne/se,nw/ne,sw/nw}{}
			\node (yA) at ($(nw)!0.5!(sw)+(-0.5,-0.0)$)[xnodeA,label=90:{$y$}]{};
			\node (yB) at ($(ne)!0.5!(se)+(+0.5,-0.0)$)[xnodeB,label=90:{$y'$}]{};

			\foreach \x in {sw,se,nw}{
				\draw[xedgeA] (yA) -- (\x);
			}
			\foreach \x in {sw,se,ne}{
				\draw[xedgeB] (yB) -- (\x);
			}
			\draw[gray,dotted] (yA) -- (yB);
			\node at (-0.5*\xr,-1*\yr)[]{$\swarrow$};
			\node at (0.5*\xr,-1*\yr)[]{$\searrow$};
		\end{scope}

		\begin{scope}[xshift=1*\xsh*\xr cm,yshift=-2*\ysh*\yr cm]
			\nfourgadget{se/sw}{ne/se,nw/ne,sw/nw}{}
			\node (yA) at ($(nw)!0.5!(sw)+(-0.5,-0.0)$)[xnodeA,label=90:{$y$}]{};
			\node (yB) at ($(ne)!0.5!(se)+(+0.5,-0.0)$)[xnodeB,label=90:{$y'$}]{};

			\foreach \x in {sw,se,nw}{
				\draw[xedgeA] (yA) -- (\x);
			}
			\foreach \x in {sw,se,ne}{
				\draw[xedgeB] (yB) -- (\x);
			}
			\node (z) at ($(ne)!0.5!(nw)+(+0.0,+0.6)$)[xnodeC,label=0:{$y''$}]{};
			\draw[xedgeC] (yA) -- (z) -- (nw);
			\draw[xedgeC] (yB) -- (z);
		\end{scope}

		\begin{scope}[xshift=2*\xsh*\xr cm,yshift=-2*\ysh*\yr cm]
			\nfourgadget{se/sw}{ne/se,nw/ne,sw/nw}{}
			\node (yA) at ($(nw)!0.5!(sw)+(-0.5,-0.0)$)[xnodeA,label=90:{$y$}]{};
			\node (yB) at ($(ne)!0.5!(se)+(+0.5,-0.0)$)[xnodeB,label=90:{$y'$}]{};

			\foreach \x in {sw,se,nw}{
				\draw[xedgeA] (yA) -- (\x);
			}
			\foreach \x in {sw,se,ne}{
				\draw[xedgeB] (yB) -- (\x);
			}
			\node (z) at ($(se)!0.5!(sw)+(+0.0,-0.6)$)[xnodeC,label=0:{$y''$}]{};
			\draw[xedgeC,ultra thick] (yA) -- (z);
			\draw[xedgeC,ultra thick] (yB) -- (z);
		\end{scope}

		\begin{scope}[xshift=3.5*\xsh*\xr cm,yshift=-1*\ysh*\yr cm]
			\nfourgadget{}{ne/se,nw/ne,se/sw,sw/nw}{}
			\node (yA) at ($(nw)!0.5!(sw)+(-0.5,-0.0)$)[xnodeA,label=90:{$y$}]{};
			\node (yB) at ($(ne)!0.5!(se)+(+0.5,-0.0)$)[xnodeB,label=90:{$y'$}]{};
			\foreach \x in {sw,se,nw}{
				\draw[xedgeA] (yA) -- (\x);
			}
			\foreach \x in {ne,se,nw}{
				\draw[xedgeB] (yB) -- (\x);
			}
			\draw[gray,dotted] (yA) -- (yB);
			\node at (-0.5*\xr,-1*\yr)[]{$\swarrow$};
			\node at (0.5*\xr,-1*\yr)[]{$\searrow$};
		\end{scope}

		\begin{scope}[xshift=3*\xsh*\xr cm,yshift=-2*\ysh*\yr cm]
			\nfourgadget{}{ne/se,nw/ne,se/sw,sw/nw}{}
			\node (yA) at ($(nw)!0.5!(sw)+(-0.5,-0.0)$)[xnodeA,label=90:{$y$}]{};
			\node (yB) at ($(ne)!0.5!(se)+(+0.5,-0.0)$)[xnodeB,label=90:{$y'$}]{};

			\foreach \x in {sw,se,nw}{
				\draw[xedgeA] (yA) -- (\x);
			}
			\foreach \x in {ne,se,nw}{
				\draw[xedgeB] (yB) -- (\x);
			}

			\node (z) at ($(ne)+(-0.2,+0.6)$)[xnodeC,label=0:{$y''$}]{};
			\draw[xedgeC] (yA) -- (z);
			\draw[xedgeC] (yB) -- (z);
			\draw[xedgeC] (ne) -- (z);
		\end{scope}

		\begin{scope}[xshift=4*\xsh*\xr cm,yshift=-2*\ysh*\yr cm]
			\nfourgadget{}{ne/se,nw/ne,se/sw,sw/nw}{}
			\node (yA) at ($(nw)!0.5!(sw)+(-0.5,-0.0)$)[xnodeA,label=90:{$y$}]{};
			\node (yB) at ($(ne)!0.5!(se)+(+0.5,-0.0)$)[xnodeB,label=90:{$y'$}]{};

			\foreach \x in {sw,se,nw}{
				\draw[xedgeA] (yA) -- (\x);
			}
			\foreach \x in {ne,se,nw}{
				\draw[xedgeB] (yB) -- (\x);
			}

			\node (z) at ($(ne)+(-0.2,+0.6)$)[xnodeC,label=0:{$y''$}]{};
			\draw[xedgeC] (yA) -- (z);
			\draw[xedgeC] (sw) -- (z);
			\draw[xedgeC] (ne) -- (z);
		\end{scope}

	\end{tikzpicture}
	\caption{$C_4$-Intersections.}
	\label{fig:H5D4:P:Intersect:C4}
	\end{figure}
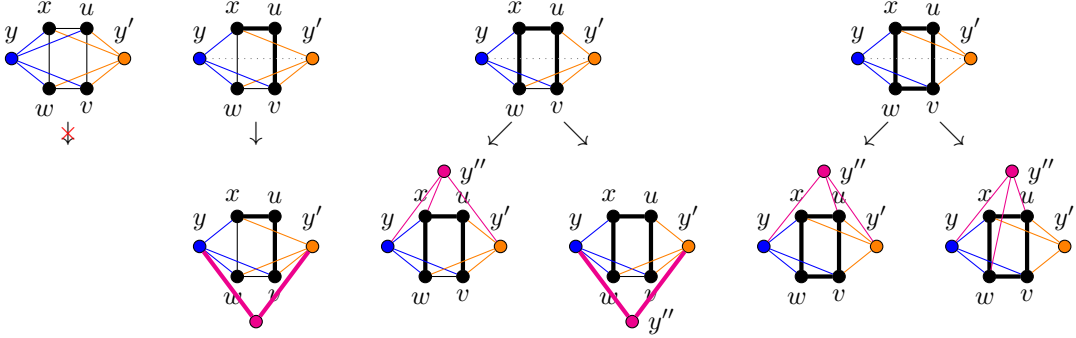

	\xcase{1}{$N_{G[H]}(y)=N_{G[H']}(y')=\Vcap$}
	In this case,
	all vertices have degree four and there cannot be another vertex

	\xcase{2}{$N_{G[H]}(y)=N_{G[H']}(y')=\{x,v,w\}$}
		  Let~$H''\in\bnd(\{H,H'\})$.
			Note that~$H''\setminus (H\cup H')$ can only be adjacent with~$y,u,y'$.
			We have that~$|H''\setminus (H\cup H')|<3$
			since no pair of~$y,u,y''$ is adjacent
			(and hence
			the habitat would contain a degree-1 vertex).
			Assume~$|H''\setminus (H\cup H')|=2$.
			Then,
			$H''$ contains~$u$ as otherwise both in~$H''\setminus (H\cup H')$ are of degree 2.
			Moreover,
			$H''$ contains exactly one of~$y,y'$.
			Thus,
			$H''$ contains two degree two vertices,
			one of~$H''\setminus (H\cup H')$ and one of~$y,y'$;
			a contradiction.

			Let~$y''= H''\setminus (H\cup H')$.
			If~$y''$ is adjacent with all three~$y,u,y'$,
			then,
			independent of which habitat~$H''$ additionally contains from~$\{x,w,v\}$,
			$u$, $y$, and~$y'$ have degree two.
			If~$y''$ is adjacent with~$y,u$,
			then \cref{obs:H5D4:P:P3free} applies.
			If~$y''$ is adjacent with~$y',u$,
			then \cref{obs:H5D4:P:P3free} applies.
			Let~$N_{G[H'']}(y'')\neq \{y,y'\}$.
			Then,
			any fourth habitat~$H^*\in\bnd(\{H,H',H''\})$ must contain~$u$ and~$y''$,
			and since~$u$ and~$y''$ have no common neighbor,
			we have~$|H^*\setminus (H\cup H'\cup H'')|=1$.
			Thus, \cref{obs:H5D4:P:P3free} applies.

		\xcase{3}{$N_{G[H]}(y)=\{x,v,w\}$ and $N_{G[H']}(y')=\{u,v,w\}$}
		Let~$H''\in\bnd(\{H,H'\})$.
		First observe that~$ \{y,y'\}\not\in E$
		since otherwise~$H''\setminus (H\cup H')$ is only adjacent with~$x,u$,
		each of which is has degree three in`$G[H\cup H']$,
		and thus,
		\cref{obs:H5D4:P:K3free} applies.
		If~$|H''\setminus (H\cup H')|\geq 2$,
		then either there are two vertices from $y,x,u,y'$ of degree two in~$G[H'']$,
		or one of~$H''\setminus (H\cup H')$ and one of $y,x,u,y'$ are of degree two.
		Thus,
		let~$y''= H''\setminus (H\cup H')$.
		We have that~$y'$ is adjacent to some of~$y,x,u,y'$.
		If~$y'$ is adjacent to all four,
		then $H''$ is isomorphic to \cref{fig:H5D4:P}(v).
		If~$y'$ is adjacent to~$\{y,x,u\}$
		(the case~$\{x,u,y'\}$ is symmetry),
		then~$u$ and either~$y$ (if~$H''$ contains no further vertex)
		or one of~$w,v,y'')$ is of degree at most two in~$G[H'']$.
		Let~$y'$ be adjacent to~$\{y,x,y'\}$
		(the case~$\{y,u,y'\}$ is symmetric).
		Assume that there is~$H^*\in\bnd(\{H,H',H''\})$.
		Then,
		$H^*\setminus (H\cup H'\cup H'')$ must be adjacent with~$y'',u$,
		each of which is of degree three.
		Thus,
		\cref{obs:H5D4:P:P3free} applies.

		If $y''$ is adjacent only with~$\{y,u\}$ or $\{x,y'\}$,
		then \cref{obs:H5D4:P:P3free} applies.
		If $y''$ is adjacent only with~$\{y,x\}$, $\{u,x\}$, or $\{u,y'\}$,
		then \cref{obs:H5D4:P:K3free} applies.
		Let~$y''$ be adjacent to~$\{y,y'\}$.
		Let~$H^*\in\bnd(\{H,H',H''\})$.
		$H^*$ must contain~$y''$
		and
		at least one of~$x,u$.
		Hence,
		that~$H^*$ cannot contain both~$y$ and~$y'$,
		since otherwise any of these two is of degree one.
		Moreover,
		$|H^*\setminus (H\cup H'\cup H'')|<3$.

		Let~$y^*=H^*\setminus (H\cup H'\cup H'')$.
		Note that~$y^*$ cannot be adjacent to only one of~$x$,
		then then~\cref{obs:H5D4:P:P3free} applies.
		Let~$x,u\in H^*$,
		then one of~$y,y'$ must be in~$H^*$
		(since any~$w,v$ would be of degree one).
		If~$y\in H^*$,
		then~$y$ and~$u$ are of degree two.
		If~$y'\in H^*$,
		then~$y'$ and~$x$ are of degree two.
		Either way yields a contradiction.

		Let~$\{y^*,z^*\}=H^*\setminus (H\cup H'\cup H'')$.
		If~$x,u\in H^*$,
		then~$y''$ must be adjacent with both~$y^*,z^*$,
		and hence,
		both~$x$ and~$u$ are of degree two.
		Let~$x\in H^*$ and~$u\not\in H^*$
		(the case $x\not\in H^*$ and~$u\in H^*$ is symmetric).
		Then,
		$H^*$ must contain~$y$.
		Then,
		$y$ and $x$ are of degree two in~$G[H^*]$;
		a contradiction.

		\xcase{4}{$N_{G[H]}(y)=\{x,v,w\}$ and $N_{G[H']}(y')=\{x,u,v\}$}
	  Let~$H''\in\bnd(\{H,H'\}$.
	  Note that~$H''\setminus (H\cup H')$ can only be adjacent with vertices from~$\{y,u,w,y''\}$.
	  If~$\{y,y'\}\in E$,
	  then~$H''\setminus (H\cup H')$ is only adjacent with~$u$ and~$w$
	  each of which has degree three,
	  and hence,
	  \cref{obs:H5D4:P:P3free} applies.
	  Let~$\{y,y'\}\not\in E$.
	  $H''\setminus (H\cup H')$ cannot be adjacent with all four~$y,u,w,y''$
	  since there is no edge~$\{a,b\}\in E$ with~$a\in\{y,w\}$ and~$b\in\{u,y'\}$.

		Let~$H''\setminus (H\cup H')$ be adjacent with~$y,u,y'$
		(the case of $y,w,y'$ is symmetric).
		Since~$y$ is neither adjacent with $u$ nor with~$y'$,
		we have that $y''= H''\setminus (H\cup H')$.
		Assume that there is~$H^*\in \bnd(\{H,H'.H''\})$.
		Since~$w$ and~$y''$ are of degree three,
		$|H^*\setminus (H\cup H'\cup H'')|=1$,
		and thus,
		\cref{obs:H5D4:P:P3free} applies.

		Let~$H''\setminus (H\cup H')$ be adjacent with~$y,u,w$
		(the case of $y',u,w$ is symmetric).
		Since~$u$ is neither adjacent with $y$ nor with~$w$,
		we have that $y''= H''\setminus (H\cup H')$.
		Assume that there is~$H^*\in \bnd(\{H,H'.H''\})$.
		Since~$y'$ and~$y''$ are of degree three,
		$|H^*\setminus (H\cup H'\cup H'')|=1$,
		and thus,
		\cref{obs:H5D4:P:P3free} applies.

		Next we consider the case that $H''\setminus (H\cup H')$ is adjacent with exactly two in~$H\cup H'$.
		Note that all vertices are of degree at least three,
		we have that~$|H''\setminus (H\cup H')|=1$ and let~$y''=H''\setminus (H\cup H')$.
		If~$y''$ is adjacent with~$u,w$, $y,u$, $w,y''$, or~$y,y'$,
		then~\cref{obs:H5D4:P:P3free} applies.
		If~$y''$ is adjacent with~$y,w$ or $u,y'$,
		then~\cref{obs:H5D4:P:K3free} applies.
\end{proof}

\begin{lemma}\label{lem:H5D4:P:Claw}
 If~$\Gcap$ is isomorphic to the~$Claw+e$,
 then the union of habitats in the connected component in~$\calG_{G,F^*\calH}$ containing~$H$ and~$H'$ is of order at most 7.
\end{lemma}

\begin{proof}
 Let~$\{u,v,w,x\}$ form the~$Claw+e$ with center~$w$ and additional edge~$\{x,u\}$.
 Let~$y= H\setminus H'$ and~$y'=H'\setminus H$.
 Since~$H$ and~$H'$ cannot have at least two degree-2,
 we know that each of~$y,y'$ is adjacent to at least~$u,v,x$.
 Moreover,
 the edges~$\{y,v\},\{y',v\},\{v,w\}$ are forced.
 We distinguish the following cases
 (see~\cref{fig:H5D4:P:Intersect:Claw}).
	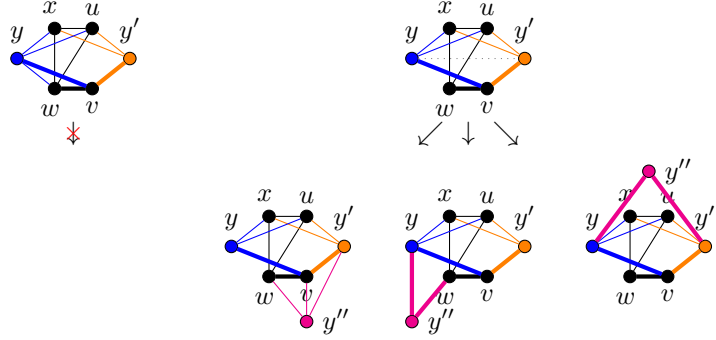
\begin{figure}[t]
	\centering
	\begin{tikzpicture}
		\def\xr{1}
		\def\yr{1}
		\def\xsh{3}
		\def\ysh{2.5}
		\tikzpreamble{}

		\newcommand{\nfourgadget}[3]{%
			\node (ne) at (0.25*\xr,0.4*\yr)[xnode,label={[label distance=-1pt](90):{$u$}}]{};
			\node (se) at (0.25*\xr,-0.4*\yr)[xnode,label={[label distance=-1pt](-90):{$v$}}]{};
			\node (sw) at (-0.25*\xr,-0.4*\yr)[xnode,label={[label distance=0pt](-90):{$w$~~}}]{};
			\node (nw) at (-0.25*\xr,0.4*\yr)[xnode,label={[label distance=0pt](90):{$x$~~}}]{};

			\foreach \x/\y in {#1}{\draw[xedge] (\x) to (\y);}
			\foreach \x/\y in {#2}{\draw[xedge,ultra thick] (\x) to (\y);}
			\foreach \x/\y in {#3}{\draw[xsedge] (\x) to (\y);}
		}
		\newcommand{\xlabel}[1]{\node at (-0.5*\xr,1.25*\yr)[anchor=east]{(#1)};}

		\newcommand{\currH}{\nfourgadget{nw/ne,ne/sw,sw/nw}{sw/se}{}}

		\begin{scope}[xshift=0*\xsh*\xr cm,yshift=-1*\ysh*\yr cm]
			\currH{}
			\node (yA) at ($(nw)!0.5!(sw)+(-0.5,-0.0)$)[xnodeA,label=90:{$y$}]{};
			\node (yB) at ($(ne)!0.5!(se)+(+0.5,-0.0)$)[xnodeB,label=90:{$y'$}]{};

			\foreach \x in {se,nw,ne}{
				\draw[xedgeA] (yA) -- (\x);
				\draw[xedgeB] (yB) -- (\x);
			}
			\draw[gray,dotted] (yA) -- (yB);
			\draw[xedgeA,ultra thick] (yA) -- (se);
			\draw[xedgeB,ultra thick] (yB) -- (se);
			\node at (0*\xr,-1*\yr)[]{$\downarrow$};
			\node at (-0.5*\xr,-1*\yr)[]{$\swarrow$};
			\node at (0.5*\xr,-1*\yr)[]{$\searrow$};
		\end{scope}

		\begin{scope}[xshift=-0.8*\xsh*\xr cm,yshift=-2*\ysh*\yr cm]
			\currH{}
			\node (yA) at ($(nw)!0.5!(sw)+(-0.5,-0.0)$)[xnodeA,label=90:{$y$}]{};
			\node (yB) at ($(ne)!0.5!(se)+(+0.5,-0.0)$)[xnodeB,label=90:{$y'$}]{};

			\foreach \x in {se,nw,ne}{
				\draw[xedgeA] (yA) -- (\x);
				\draw[xedgeB] (yB) -- (\x);
			}
			\draw[xedgeA,ultra thick] (yA) -- (se);
			\draw[xedgeB,ultra thick] (yB) -- (se);
			\node (z) at ($(se)+(+0.0,-0.6)$)[xnodeC,label=0:{$y''$}]{};
			\draw[xedgeC] (se) -- (z) -- (sw);
			\draw[xedgeC] (yB) -- (z);

		\end{scope}

		\begin{scope}[xshift=0.8*\xsh*\xr cm,yshift=-2*\ysh*\yr cm]
			\currH{}
			\node (yA) at ($(nw)!0.5!(sw)+(-0.5,-0.0)$)[xnodeA,label=90:{$y$}]{};
			\node (yB) at ($(ne)!0.5!(se)+(+0.5,-0.0)$)[xnodeB,label=90:{$y'$}]{};

			\foreach \x in {se,nw,ne}{
				\draw[xedgeA] (yA) -- (\x);
				\draw[xedgeB] (yB) -- (\x);
			}
			\draw[xedgeA,ultra thick] (yA) -- (se);
			\draw[xedgeB,ultra thick] (yB) -- (se);
			\node (z) at ($(ne)!0.5!(nw)+(+0.0,+0.6)$)[xnodeC,label=0:{$y''$}]{};
			\draw[xedgeC,ultra thick] (yA) -- (z);
			\draw[xedgeC,ultra thick] (yB) -- (z);

		\end{scope}

		\begin{scope}[xshift=0.0*\xsh*\xr cm,yshift=-2*\ysh*\yr cm]
			\currH{}
			\node (yA) at ($(nw)!0.5!(sw)+(-0.5,-0.0)$)[xnodeA,label=90:{$y$}]{};
			\node (yB) at ($(ne)!0.5!(se)+(+0.5,-0.0)$)[xnodeB,label=90:{$y'$}]{};

			\foreach \x in {se,nw,ne}{
				\draw[xedgeA] (yA) -- (\x);
				\draw[xedgeB] (yB) -- (\x);
			}
			\draw[xedgeA,ultra thick] (yA) -- (se);
			\draw[xedgeB,ultra thick] (yB) -- (se);

			\node (z) at ($(se)!0.5!(sw)+(-0.75,-0.6)$)[xnodeC,label=0:{$y''$}]{};
			\draw[xedgeC,ultra thick] (sw) -- (z);
			\draw[xedgeC,ultra thick] (yA) -- (z);
		\end{scope}

		\begin{scope}[xshift=-1.75*\xsh*\xr cm,yshift=-1*\ysh*\yr cm]
		\currH{}
		\node (yA) at ($(nw)!0.5!(sw)+(-0.5,-0.0)$)[xnodeA,label=90:{$y$}]{};
		\node (yB) at ($(ne)!0.5!(se)+(+0.5,-0.0)$)[xnodeB,label=90:{$y'$}]{};

		\foreach \x in {se,nw,ne}{
			\draw[xedgeA] (yA) -- (\x);
			\draw[xedgeB] (yB) -- (\x);
		}
		\draw[xedgeA,ultra thick] (yA) -- (se);
		\draw[xedgeB,ultra thick] (yB) -- (se);
		\draw[xedgeA] (yA) -- (sw);
		\node at (0*\xr,-1*\yr)[]{$\downarrow$};
		\node at (0,-1*\yr)[color=red,thick]{$\times$};
		\end{scope}

	\end{tikzpicture}
	\caption{$(Claw+e)$-Intersections.}
	\label{fig:H5D4:P:Intersect:Claw}
	\end{figure}

	\xcase{1}{$N_{G[H]}(y)=\{u,v,w,x\}$}
	Then~$y'$ is only adjacent to~$x,u,v$ in~$G[H']$.
	Assume that there is~$H''\in\bnd(\{H, H'\})$.
	Any vertex in $H''\setminus(H\cup H')$ can only be adjacent with~$v$ and~$y'$.
	Note that~$v$ and~$y'$ have no common neighbor in~$G[H\cup H']$ and the edge~$\{v,y'\}$ is forced.
	Thus, $|H''\setminus(H\cup H')|=1$,
	and \cref{obs:H5D4:P:K3free} applies;
	a contradiction.

	\xcase{2}{$N_{G[H]}(y)=N_{G[H]}(y')=\{u,v,x\}$}
	Let~$H''\in\bnd(\{H, H'\})$.
	If~$\{y,y'\}\in E$,
	then~\cref{obs:H5D4:P:docking33} applies.
	Let~$\{y,y'\}\not\in E$.
	Note that~$|H''\setminus(H\cup H')|<2$,
	since all vertices in $y,y',v,w$ are of degree three and
	any of pair of them is either not adjacent or has no common neighbor in~$G[H\cup H']$.
	Let~$y''=H''\setminus(H\cup H')$.

	\xsubcase{a}{$y''$ is adjacent with all four of $y,y',v,w$ in~$G[H'']$}
	All possible intersecting edges are forced;
	a contradiction.

	\xsubcase{b}{$y''$ is adjacent with exactly three from~$y,y',v,w$}
	Then,
	$H''$ contains exactly one of~$x,u$,
	say~$x$ (by symmetry),
	since all other edges in~$G[H\cup H']$ are forced.
	If~$H''$ contains not both~$y,y'$,
	then~$x$ is of degree two, and hence there is no unforced edge in the intersection.
	If~$H''$ contains both~$y,y'$,
	then these two are of degree two if~$H''$ contains~$w$,
	or~$x$ is of degree two if~$H''$ contains~$v$
	(again, in this case there is no unforced edge in the intersection).
	Either way yields a contradiction.

	\xsubcase{c}{$y''$ is adjacent with exactly two from~$y,y',v,w$}
	If~$y''$ is adjacent with~$v$,
	then \cref{obs:H5D4:P:P3free} or \cref{obs:H5D4:P:K3free} applies;
	a contradiction.
	If~$y''$ is adjacent with~$w,y$
	(the case of~$w,y'$ is symmetric),
	then~$x,u$ must be in~$H''$.
	Assume that there is~$H^*\in\bnd(\{H,H',H''\})$.
    Let~$\{y'',v\}\in E$.
    Then~$H^*\setminus (H\cup H'\cup H'')$ is adjacent only with~$y'',z$ in~$G[H^*]$.
    Since both~$y'',z$ are of degree three and are non-adjacent,
    $|H^*\setminus (H\cup H'\cup H'')|<2$.
    If $|H^*\setminus (H\cup H'\cup H'')|=1$,
    then \cref{obs:H5D4:P:P3free} applies.
    Let~$\{y'',y'\}\in E$,
    Then~$H^*\setminus (H\cup H'\cup H'')$ is adjacent only with~$y'',v$ in~$G[H^*]$.
    Since both~$y'',v$ are of degree three and are non-adjacent in~$G[H\cup H'\cup H'']$,
    $|H^*\setminus (H\cup H'\cup H'')|<2$.
    If $|H^*\setminus (H\cup H'\cup H'')|=1$,
    then \cref{obs:H5D4:P:P3free} applies.
    Let $\{y'',v\}\{y'',y'\}\not\in E$.
    If $H^*\setminus (H\cup H'\cup H'')$ is adjacent to all three $y'',y',v$ in~$G[H^*]$,
    then both~$v,y'$ are of degree two.
    Let~$H^*\setminus (H\cup H'\cup H'')$ be adjacent with exactly two of $y'',y',v$ in~$G[H^*]$.
    Since~$v$ and~$y'$ are of degree three and none of them is adjacent with~$y''$,
    we have that~$|H^*\setminus (H\cup H'\cup H'')|=1$.
    Then \cref{obs:H5D4:P:P3free} or \cref{obs:H5D4:P:K3free} applies;
	a contradiction.
	Finally,
	let $y''$ be adjacent exactly with~$y,y'$ in~$G[H'']$.
	It follows that~$H''=\{y'',y,y',x,u\}$.
	Assume that there is~$H^*\in\bnd(\{H,H',H''\})$.
	$H^*$ must contain~$y''$,
	any of~$v,w$ (since they are, next to~$y''$, the only vertices of degree smaller than four),
	and any of~$x,u$ (since every unforced edge contains~$x$ or~$u$),
	say~$x$ (by symmetry).
	Since~$x$ cannot be of degree one in~$G[H^*]$,
	$H^*$ must contain two from~$N_{G[H\cup H'\cup H'']}(x)=\{y,u,w,y'\}$.
	It follows that~$|H^*\setminus(H\cup H'\cup H'')|=1$,
	let $y^*=H^*\setminus(H\cup H'\cup H'')$,
	and that~$H^*$ contains exactly one of~$v,w$.
	Thus,
	$y^*$ is of degree two.
	Since $H^*$ cannot contain both~$y,y'$ (since it is of size at most 5),
	$y''$ is also of degree two;
	a contradiction.
\end{proof}

\begin{lemma}\label{lem:H5D4:P:K4e}
 If~$\Gcap$ is isomorphic to the~$K_4-e$,
 then the union of habitats in the connected component in~$\calG_{G,F^*\calH}$ containing~$H$ and~$H'$ is of order at most 7.
\end{lemma}

\begin{proof}
 Let~$\{u,v,w,x\}$ form the~$K_4-e$ with edge~$\{x,v\}$ being missing.
 Let~$y= H\setminus H'$ and~$y'=H'\setminus H$.
 Since~$H$ and~$H'$ cannot have at least two degree-2,
 we know that each of~$y,y'$ is adjacent to at least one of~$x$ and~$v$,
 and if only to one,
 say~$x$,
 then it must contain~$u,w$ as otherwise~$y$ or~$y'$ are also of degree two.
 We distinguish the following cases
 (see~\cref{fig:H5D4:P:Intersect:K4e} and \cref{fig:H5D4:P:Intersect:K4e:22}).
	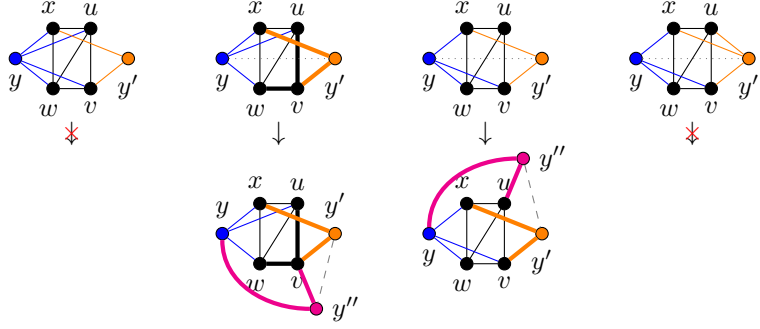
\begin{figure}[t]
	\centering
	\begin{tikzpicture}
		\def\xr{1}
		\def\yr{1}
		\def\xsh{2.75}
		\def\ysh{2.33}
		\tikzpreamble{}

		\newcommand{\nthreegadget}[3]{%
			\node (n) at (0.25*\xr,0.4*\yr)[xnode,label={[label distance=-1pt](90):{$u$}}]{};
			\node (s) at (0.25*\xr,-0.4*\yr)[xnode,label={[label distance=-1pt](-90):{$v$}}]{};
			\node (w) at (-0.25*\xr,0.0*\yr)[xnode,label={[label distance=0pt](90):{$w$~~}}]{};

			\foreach \x/\y in {#1}{\draw[xedge] (\x) to (\y);}
			\foreach \x/\y in {#2}{\draw[xedge,ultra thick] (\x) to (\y);}
			\foreach \x/\y in {#3}{\draw[xsedge] (\x) to (\y);}
		}
		\newcommand{\nfourgadget}[3]{%
			\node (ne) at (0.25*\xr,0.4*\yr)[xnode,label={[label distance=-1pt](90):{$u$}}]{};
			\node (se) at (0.25*\xr,-0.4*\yr)[xnode,label={[label distance=-1pt](-90):{$v$}}]{};
			\node (sw) at (-0.25*\xr,-0.4*\yr)[xnode,label={[label distance=0pt](-90):{$w$~~}}]{};
			\node (nw) at (-0.25*\xr,0.4*\yr)[xnode,label={[label distance=0pt](90):{$x$~~}}]{};

			\foreach \x/\y in {#1}{\draw[xedge] (\x) to (\y);}
			\foreach \x/\y in {#2}{\draw[xedge,ultra thick] (\x) to (\y);}
			\foreach \x/\y in {#3}{\draw[xsedge] (\x) to (\y);}
		}
		\newcommand{\xlabel}[1]{\node at (-0.5*\xr,1.25*\yr)[anchor=east]{(#1)};}

		\newcommand{\currH}{\nfourgadget{nw/ne,ne/se,se/sw,sw/nw,sw/ne}{}{}}

		\begin{scope}[xshift=0.0*\xsh*\xr cm,yshift=-1*\ysh*\yr cm]
			\currH{}
			\node (yA) at ($(nw)!0.5!(sw)+(-0.5,-0.0)$)[xnodeA,label=-90:{$y$}]{};
			\node (yB) at ($(ne)!0.5!(se)+(+0.5,-0.0)$)[xnodeB,label=-90:{$y'$}]{};

			\foreach \x in {nw,ne,se,sw}{
				\draw[xedgeA] (yA) -- (\x);
			}
			\foreach \x in {nw,se}{
				\draw[xedgeB] (yB) -- (\x);
			}
			\node at (0,-1*\yr)[]{$\downarrow$};
			\node at (0,-1*\yr)[color=red,thick]{$\times$};
		\end{scope}

		\begin{scope}[xshift=1*\xsh*\xr cm,yshift=-1*\ysh*\yr cm]
			\nfourgadget{nw/ne,sw/nw,sw/ne}{ne/se,se/sw}{}
			\node (yA) at ($(nw)!0.5!(sw)+(-0.5,-0.0)$)[xnodeA,label=-90:{$y$}]{};
			\node (yB) at ($(ne)!0.5!(se)+(+0.5,-0.0)$)[xnodeB,label=-90:{$y'$}]{};

			\foreach \x in {nw,ne,sw}{
				\draw[xedgeA] (yA) -- (\x);
			}
			\foreach \x in {nw,se}{
				\draw[xedgeB,ultra thick] (yB) -- (\x);
			}
			\draw[gray,dotted] (yA) -- (yB);
			\node at (0,-1*\yr)[]{$\downarrow$};
		\end{scope}

		\begin{scope}[xshift=1*\xsh*\xr cm,yshift=-2*\ysh*\yr cm]
			\nfourgadget{nw/ne,sw/nw,sw/ne}{ne/se,se/sw}{}
			\node (yA) at ($(nw)!0.5!(sw)+(-0.5,-0.0)$)[xnodeA,label=90:{$y$}]{};
			\node (yB) at ($(ne)!0.5!(se)+(+0.5,-0.0)$)[xnodeB,label=90:{$y'$}]{};

			\foreach \x in {nw,ne,sw}{
				\draw[xedgeA] (yA) -- (\x);
			}
			\foreach \x in {nw,se}{
				\draw[xedgeB,ultra thick] (yB) -- (\x);
			}
			\node (z) at ($(sw)!0.5!(se)+(0.5,-0.6*\yr)$)[xnodeC,label=0:{$y''$}]{};
			\draw[xedgeC,ultra thick] (z) -- (se);
			\draw[xedgeC,ultra thick] (yA) to [out=-90,in=180](z);
			\draw[gray,dashed] (yB) -- (z);
		\end{scope}

		\begin{scope}[xshift=2*\xsh*\xr cm,yshift=-1*\ysh*\yr cm]
			\currH{}
			\node (yA) at ($(nw)!0.5!(sw)+(-0.5,-0.0)$)[xnodeA,label=-90:{$y$}]{};
			\node (yB) at ($(ne)!0.5!(se)+(+0.5,-0.0)$)[xnodeB,label=-90:{$y'$}]{};

			\foreach \x in {nw,se,sw}{
				\draw[xedgeA] (yA) -- (\x);
			}
			\foreach \x in {nw,se}{
				\draw[xedgeB] (yB) -- (\x);
			}
			\draw[gray,dotted] (yA) -- (yB);
			\node at (0,-1*\yr)[]{$\downarrow$};
		\end{scope}

		\begin{scope}[xshift=2*\xsh*\xr cm,yshift=-2*\ysh*\yr cm]
			\currH{}
			\node (yA) at ($(nw)!0.5!(sw)+(-0.5,-0.0)$)[xnodeA,label=-90:{$y$}]{};
			\node (yB) at ($(ne)!0.5!(se)+(+0.5,-0.0)$)[xnodeB,label=-90:{$y'$}]{};

			\foreach \x in {nw,se,sw}{
				\draw[xedgeA] (yA) -- (\x);
			}
			\foreach \x in {nw,se}{
				\draw[xedgeB, ultra thick] (yB) -- (\x);
			}
			\node (z) at ($(nw)!0.5!(ne)+(0.5,0.6*\yr)$)[xnodeC,label=0:{$y''$}]{};
			\draw[xedgeC, ultra thick] (z) -- (ne);
			\draw[xedgeC, ultra thick] (yA) to [out=90,in=180](z);

			\draw[gray, dashed] (z) -- (yB);
		\end{scope}

		\begin{scope}[xshift=3*\xsh*\xr cm,yshift=-1*\ysh*\yr cm]
			\currH{}
			\node (yA) at ($(nw)!0.5!(sw)+(-0.5,-0.0)$)[xnodeA,label=-90:{$y$}]{};
			\node (yB) at ($(ne)!0.5!(se)+(+0.5,-0.0)$)[xnodeB,label=-90:{$y'$}]{};

			\foreach \x in {nw,se,sw}{
				\draw[xedgeA] (yA) -- (\x);
			}
			\foreach \x in {nw,se,ne}{
				\draw[xedgeB] (yB) -- (\x);
			}
			\draw[gray,dotted] (yA) -- (yB);
			\node at (0,-1*\yr)[]{$\downarrow$};
			\node at (0,-1*\yr)[color=red,thick]{$\times$};
		\end{scope}

	\end{tikzpicture}
	\caption{($K_4-e$)-Intersections.}
	\label{fig:H5D4:P:Intersect:K4e}
	\end{figure}

	\xcase{1}{$N_{G[H]}(y)=\{u,v,w,x\}$}
	Then~$N_{G[H]}(y')=\{x,v\}$.
	Hence,
	$y'$ is the only vertex of degree smaller than four in~$G[H\cup H']$,
	and thus~$\bnd(\{H,H'\})=\emptyset$.

	\xcase{2}{$N_{G[H]}(y)=\{x,u,w\}$ and~$N_{G[H]}(y')=\{x,v\}$}
	In this case,
	the edges~$\{v,w\},\{v,u\},\{x,y'\},\{v,y'\}$ are forced.
	Let~$H''\in\bnd(\{H, H'\})$.
	If~$\{y,y'\}\in E$,
	then~$H''\setminus(H\cup H')$ must be adjacent with~$v,y'$,
	but then \cref{obs:H5D4:P:K3free} applies.
	Let~$\{y,y'\}\not\in E$.
	Note that~$|H''\setminus(H\cup H')|<3$,
	since the only edge it could intersect with is~$\{v,y'\}$,
	which is forced.
	Let~$|H''\setminus(H\cup H')|=2$.
	$H''\setminus(H\cup H')$ cannot be adjacent with~$y,v$ both would have degree two in $G[H'']$.
	Thus,
	$H''$ must contain~$y'$,
	and exactly one of~$y,v$.
	If~$H''$ contains~$v$,
	then,
	since~$v$ and~$y'$ have no common neighbor,
	$H''$ contains at least two degree-at-most-2 vertices.
	If $H''$ contains~$y$,
	then it contains~$x$,
	all of which have degree two.

	Let~$y''=H''\setminus(H\cup H')$.
	If~$y''$ is adjacent exactly with~$v,y'$ or with~$y,y'$,
	then \cref{obs:H5D4:P:P3free} or \cref{obs:H5D4:P:K3free} applies.
	Let~$y''$ be adjacent exactly with~$y,v$.
	Let~$H^*\in\bnd(\{H, H', H''\})$.
	Then~$H^*\setminus (H\cup H'\cup H'')$ must be adjacent with~$y',y''$.
	If~$|H^*\setminus (H\cup H'\cup H'')|=1$,
	then \cref{obs:H5D4:P:K3free} or \cref{obs:H5D4:P:P3free} applies.
	If~$|H^*\setminus (H\cup H'\cup H'')|\geq 2$,
	then there is no unforced edge to intersect with.

	\xcase{3}{$N_{G[H]}(y)=\{x,v,w\}$ and~$N_{G[H]}(y')=\{x,v\}$}
	Let~$H''\in\bnd(\{H,H'\})$.
	If~$\{y,y'\}\in E$,
	then~$H''\setminus (H\cup H')$ is adjacent with~$u,y'$
	(each of which is of degree three in~$G[H\cup H']$ in this case),
	and hence,
	$|H''\setminus(H\cup H')|=1$ and \cref{obs:H5D4:P:P3free} applies.
	Let~$\{y,y'\}\not\in E$.
	Assume that $H''\setminus (H\cup H')$ is adjacent with~$y,u,y'$.
	If~$|H''\setminus(H\cup H')|=2$,
	then~$y,u$ are of degree one.
	Let~$y''=H''\setminus(H\cup H')$.
	Then,
	since only one of~$x,v$ must and can be in~$H''$,
	both~$y,y''$ are of degree two.
	Assume that $H''\setminus (H\cup H')$ is adjacent with~$u,y'$ resp.\ with~$y,y'$.
	Note that~$|H''\setminus(H\cup H')|<3$.
	If~$|H''\setminus(H\cup H')|=2$,
	then one of $H''\setminus(H\cup H')$ and~$u$ resp.\ $y$ is of degree two.
	If $|H''\setminus(H\cup H')|=1$,
	then~\cref{obs:H5D4:P:P3free} applies.
	Finally, assume that $H''\setminus (H\cup H')$ is adjacent with~$y,u$.
	Note that~$y''=H''\setminus(H\cup H')$.
	Assume that there is~$H^*\in\bnd(\{H,H',H''\})$.
	Then~$H^*\setminus (H\cup H'\cup H'')$ must be adjacent with~$y',y''$,
	and~$|H^*\setminus (H\cup H'\cup H'')|=1$ since~$y',y''$ have no common neighbor.
	Thus,
	\cref{obs:H5D4:P:K3free} or \cref{obs:H5D4:P:P3free} applies.

	\xcase{4}{$N_{G[H]}(y)=\{x,v,w\}$ and~$N_{G[H]}(y')=\{x,u,v\}$}
	Note that if~$\{y,y'\}\in E$,
	all vertices have degree four and no third habitat with vertex outside $H\cup H'$ can intersect.
	Thus,
	let $\{y,y'\}\not\in E$.
	Assume that there is~$H''\in\bnd(\{H,H'\})$.
	Since all vertices are of degree at least three,
	we have $|H''\setminus(H\cup H')|=1$,
	and thus,
	\cref{obs:H5D4:P:P3free} applies;
	a contradiction.

	\xcase{5}{$N_{G[H]}(y)=N_{G[H]}(y')=\{x,v\}$}
	In this case,
	the edges~$\{y,x\}$, $\{y,v\}$, $\{y',x\}$, and $\{y',v\}$ are forced.
	Let~$H''\in\bnd(\{H, H'\})$.
	We distinguish whether edge~$\{y,y'\}$ is in~$E$
	(see \cref{fig:H5D4:P:Intersect:K4e:22}).
	\begin{figure}[t]
	\centering
	\begin{tikzpicture}
		\def\xr{1}
		\def\yr{1}
		\def\xsh{2.75}
		\def\ysh{2.33}
		\tikzpreamble{}

		\newcommand{\nfourgadget}[3]{%
			\node (ne) at (0.25*\xr,0.4*\yr)[xnode,label={[label distance=-1pt](90):{$u$}}]{};
			\node (se) at (0.25*\xr,-0.4*\yr)[xnode,label={[label distance=-1pt](-90):{$v$}}]{};
			\node (sw) at (-0.25*\xr,-0.4*\yr)[xnode,label={[label distance=0pt](-90):{$w$~~}}]{};
			\node (nw) at (-0.25*\xr,0.4*\yr)[xnode,label={[label distance=0pt](90):{$x$~~}}]{};

			\foreach \x/\y in {#1}{\draw[xedge] (\x) to (\y);}
			\foreach \x/\y in {#2}{\draw[xedge,ultra thick] (\x) to (\y);}
			\foreach \x/\y in {#3}{\draw[xsedge] (\x) to (\y);}
		}
		\newcommand{\xlabel}[1]{\node at (-0.5*\xr,1.25*\yr)[anchor=east]{(#1)};}

		\newcommand{\currH}{\nfourgadget{nw/ne,ne/se,se/sw,sw/nw,sw/ne}{}{}}

		\begin{scope}[xshift=3*\xsh*\xr cm,yshift=-1*\ysh*\yr cm]
		  \xlabel{b}
			\currH{}
			\node (yA) at ($(nw)!0.5!(sw)+(-0.5,-0.0)$)[xnodeA,label=-90:{$y$}]{};
			\node (yB) at ($(ne)!0.5!(se)+(+0.5,-0.0)$)[xnodeB,label=-90:{$y'$}]{};

			\foreach \x in {nw,se}{
				\draw[xedgeA,ultra thick] (yA) -- (\x);
			}
			\foreach \x in {nw,se}{
				\draw[xedgeB,ultra thick] (yB) -- (\x);
			}
			\draw[gray] (yA) -- (yB);
			\node (z) at ($(yA)!0.5!(nw)+(0,0.6*\yr)$)[xnodeC,label=180:{$y''$}]{};
			\draw[xedgeC] (yA) -- (z) -- (ne);
			\draw[xedgeC] (sw) -- (z);
		\end{scope}

		\begin{scope}[xshift=4*\xsh*\xr cm,yshift=-1*\ysh*\yr cm]
			\currH{}
			\node (yA) at ($(nw)!0.5!(sw)+(-0.5,-0.0)$)[xnodeA,label=-90:{$y$}]{};
			\node (yB) at ($(ne)!0.5!(se)+(+0.5,-0.0)$)[xnodeB,label=-90:{$y'$}]{};

			\foreach \x in {nw,se}{
				\draw[xedgeA,ultra thick] (yA) -- (\x);
			}
			\foreach \x in {nw,se}{
				\draw[xedgeB,ultra thick] (yB) -- (\x);
			}
			\draw[gray] (yA) -- (yB);
			\node (z) at ($(nw)!0.5!(ne)+(0,0.6*\yr)$)[xnodeC,label=00:{$y''$}]{};
			\draw[xedgeC] (yA) -- (z) -- (yB);
			\draw[xedgeC] (ne) -- (z);
		\end{scope}

		\begin{scope}[xshift=1*\xsh*\xr cm,yshift=-1*\ysh*\yr cm]
			\xlabel{a}
			\currH{}
			\node (yA) at ($(nw)!0.5!(sw)+(-0.5,-0.0)$)[xnodeA,label=-90:{$y$}]{};
			\node (yB) at ($(ne)!0.5!(se)+(+0.5,-0.0)$)[xnodeB,label=-90:{$y'$}]{};

			\foreach \x in {nw,se}{
				\draw[xedgeA,ultra thick] (yA) -- (\x);
			}
			\foreach \x in {nw,se}{
				\draw[xedgeB,ultra thick] (yB) -- (\x);
			}
			\node (z) at ($(yA)!0.5!(nw)+(0,0.6*\yr)$)[xnodeC,label=180:{$y''$}]{};
			\draw[xedgeC] (yA) -- (z) -- (ne);
			\draw[xedgeC] (sw) -- (z);
		\end{scope}

		\begin{scope}[xshift=2*\xsh*\xr cm,yshift=-1*\ysh*\yr cm]
			\currH{}
			\node (yA) at ($(nw)!0.5!(sw)+(-0.5,-0.0)$)[xnodeA,label=-90:{$y$}]{};
			\node (yB) at ($(ne)!0.5!(se)+(+0.5,-0.0)$)[xnodeB,label=-90:{$y'$}]{};

			\foreach \x in {nw,se}{
				\draw[xedgeA,ultra thick] (yA) -- (\x);
			}
			\foreach \x in {nw,se}{
				\draw[xedgeB,ultra thick] (yB) -- (\x);
			}
			\node (z) at ($(yA)!0.5!(nw)+(0,0.6*\yr)$)[xnodeC,label=90:{$y''$}]{};
			\draw[xedgeC] (yA) -- (z) -- (ne);
			\draw[xedgeC] (sw) -- (z);
			\draw[gray] (yB) to [out=120,in=0](z);
		\end{scope}

	\end{tikzpicture}
	\caption{($K_4-e$)-Intersections where~$\deg_{G[H]}(y)=\deg_{G[H']}(y')=2$.}
	\label{fig:H5D4:P:Intersect:K4e:22}
	\end{figure}
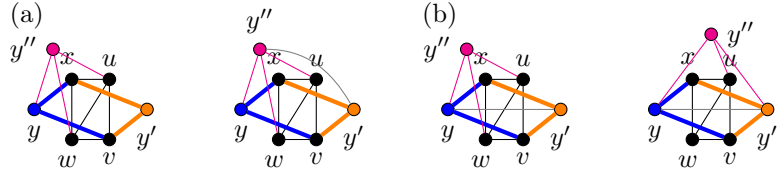

	\xsubcase{a}{$\{y,y'\}\not\in E$ (see~\cref{fig:H5D4:P:Intersect:K4e:22}(a))}
	Note that~$H''\setminus(H\cup H')$ cannot be adjacent with all four vertices~$y,y',u,w$
	since both~$y,y'$ would have degree one.
	If~$H''\setminus(H\cup H')$ is adjacent with exactly three,
	then~$|H''\setminus(H\cup H')|=1$,
	and exactly one of~$x$ or~$v$ must and can be in~$H''$
	(since otherwise there are degree-1 vertices).
	But then,
	$H''\setminus(H\cup H')$ cannot be adjacent with~$y,y',u$,
	since then~$y,y'$ are of degree two.

	Assume that $H''\setminus(H\cup H')$ is adjacent with~$y,u,w$.
	Note that $|H''\setminus(H\cup H')|<2$
	since~$u$ and~$w$ are of degree three in~$G[H\cup H']$ and not neighboring~$y$.
	Let~$y''=H''\setminus(H\cup H')$.
	Let~$H^*\in\bnd(\{H,H',H''\})$.
	Note that~$|H^*\setminus(H\cup H'\cup H'')|<2$
	since no two of $y,y',y''$ have a common neighbor and only~$y'$ is of degree two.
	Let $y^*=H^*\setminus(H\cup H'\cup H'')$.
	Let $\{y',y''\}\not\in E$.
	$y^*$ cannot be adjacent with~$y,y',y''$,
	since each choice of one of~$x,u,w,v$ results in a degree-1 vertex.
	If~$y^*$ is adjacent with~$y',y''$ or with~$y,y'$,
	then \cref{obs:H5D4:P:P3free} applies.
	If~$y^*$ is adjacent with~$y,y''$,
  then \cref{obs:H5D4:P:K3free} applies.
  Let~$\{y',y''\}\in E$.
	Then~$y^*$ must be adjacent with~$y,y'$,
	and hence,
	\cref{obs:H5D4:P:P3free} applies.
  Either way yields a contradiction.

	Assume that $H''\setminus(H\cup H')$ is adjacent with~$y,y'$.
	Note that $|H''\setminus(H\cup H')|<3$,
	since~$\{y,y'\}\not\in E$.
	Let~$|H''\setminus(H\cup H')|=2$.
	Then, either~$x\in H''$ or~$v\in H''$.
	Either case,
	$H''$ only intersects on forced edges;
	a contradiction.
	Let~$|H''\setminus(H\cup H')|=1$.
	Then,
	\cref{obs:H5D4:P:P3free} applies;
	a contradiction.

	Assume that $H''\setminus(H\cup H')$ is adjacent with~$y,u$.
	Note that $|H''\setminus(H\cup H')|<2$,
	since otherwise~$u$ and a vertex from $H''\setminus(H\cup H')$
	are of degree two in~$G[H'']$.
	If~$|H''\setminus(H\cup H')|=1$,
	then \cref{obs:H5D4:P:P3free} applies;
	a contradiction.

	Assume that $H''\setminus(H\cup H')$ is adjacent with~$u,w$.
	Note that since~$u$ and~$w$ have degree three in~$G[H\cup H']$,
	it holds that~$|H''\setminus(H\cup H')|=1$.
	In this case,
	\cref{obs:H5D4:P:K3free} applies;
	a contradiction.

	\xsubcase{b}{$\{y,y'\}\in E$ (see~\cref{fig:H5D4:P:Intersect:K4e:22}(b))}
	Note that~$H''\setminus(H\cup H')$ cannot be adjacent with all four vertices~$y,y',u,w$
	since then both~$y,y'$ would have degree two.
	If~$H''\setminus(H\cup H')$ is adjacent with three,
	then~$|H''\setminus(H\cup H')|=1$,
	and exactly one of~$x$ or~$v$ must and can be in~$H''$
	(since otherwise there are degree-1 vertices).

	Assume that $H''\setminus(H\cup H')$ is adjacent with~$y,y',u$.
	Note that $|H''\setminus(H\cup H')|<2$
	since~$u$ is of degree three in~$G[H\cup H']$ and not neighboring~$y$ or~$y'$.
	Let~$H''\setminus(H\cup H')=y''$.
	Assume that there is~$H^*\in\bnd(\{H, H', H''\})$.
	Then~$H^*\setminus(H\cup H')$ must be adjacent with~$w,y''$,
	and~$|H^*\setminus (H\cup H'\cup H'')|=1$ since both~$w,y''$ are of degree three in~$G[H\cup H'\cup H'']$.
	Thus,
	\cref{obs:H5D4:P:P3free} applies;
	a contradiction.

	Assume that $H''\setminus(H\cup H')$ is adjacent with~$y,u,w$.
	Note that $|H''\setminus(H\cup H')|<2$,
	since~$y$ is of degree three in $G[H\cup H']$ and not neighboring~$u$ or~$w$.
	Let~$H''\setminus(H\cup H')=y''$.
	Assume that there is~$H^*\in\bnd(\{H, H', H''\})$.
	Then~$H^*\setminus(H\cup H')$ must be adjacent with~$y',y''$,
	and~$|H^*\setminus (H\cup H'\cup H'')|=1$ since both~$y',y''$ are of degree three in~$G[H\cup H'\cup H'']$.
	Thus,
	\cref{obs:H5D4:P:P3free} applies;
	a contradiction.

	Assume that $H''\setminus(H\cup H')$ is adjacent with~$y,y'$.
	Since~$y$ and~$y'$ have degree three in~$G[H\cup H']$,
	it holds that~$|H''\setminus(H\cup H')|=1$.
	Then,
	\cref{obs:H5D4:P:K3free} applies;
	a contradiction.

	Assume that $H''\setminus(H\cup H')$ is adjacent with~$y,u$.
	Since~$y$ and~$u$ have degree three in~$G[H\cup H']$,
	it holds that~$|H''\setminus(H\cup H')|=1$.
	Then,
	\cref{obs:H5D4:P:P3free} applies;
	a contradiction.

	Assume that $H''\setminus(H\cup H')$ is adjacent with~$u,w$.
	Since~$u$ and~$w$ have degree three in~$G[H\cup H']$,
	it holds that~$|H''\setminus(H\cup H')|=1$.
	Then,
	\cref{obs:H5D4:P:K3free} applies;
	a contradiction.
\end{proof}

\begin{lemma}\label{lem:H5D4:P:K4}
 If~$\Gcap$ is isomorphic to the~$K_4$,
 then the union of habitats in the connected component in~$\calG_{G,F^*\calH}$ containing~$H$ and~$H'$ is of order at most 6.
\end{lemma}

\begin{proof}
 Let~$\{u,v,w,x\}$ form the~$K_4$.
 Let~$y= H\setminus H'$ and~$y'=H'\setminus H$.
 We know that each of~$y,y'$ is adjacent to exactly of $\{u,v,w,x\}$,
 and~$y$ and~$y'$ do not have a common neighbor in $\{u,v,w,x\}$.
 By symmetry,
 we can assume that~$N_{G[H]}(y)=\{x,w\}$ and $N_{G[H]}(y')=\{u,v\}$.
 (see~\cref{fig:H5D4:P:Intersect:K4}).
	\begin{figure}[t]
	\centering
	\begin{tikzpicture}
		\def\xr{1}
		\def\yr{1}
		\def\xsh{3}
		\def\ysh{2.5}
		\tikzpreamble{}

		\newcommand{\nfourgadget}[3]{%
			\node (ne) at (0.25*\xr,0.4*\yr)[xnode,label={[label distance=-1pt](90):{$u$}}]{};
			\node (se) at (0.25*\xr,-0.4*\yr)[xnode,label={[label distance=-1pt](-90):{$v$}}]{};
			\node (sw) at (-0.25*\xr,-0.4*\yr)[xnode,label={[label distance=0pt](-90):{$w$~~}}]{};
			\node (nw) at (-0.25*\xr,0.4*\yr)[xnode,label={[label distance=0pt](90):{$x$~~}}]{};

			\foreach \x/\y in {#1}{\draw[xedge] (\x) to (\y);}
			\foreach \x/\y in {#2}{\draw[xfedge] (\x) to (\y);}
			\foreach \x/\y in {#3}{\draw[xsedge] (\x) to (\y);}
		}
		\newcommand{\xlabel}[1]{\node at (-0.5*\xr,1.25*\yr)[anchor=east]{(#1)};}

		\newcommand{\currH}{
		\nfourgadget{ne/se,sw/nw,nw/ne,ne/sw,nw/se}{}{}
		\node (yA) at ($(nw)!0.5!(sw)+(-0.5*\xr,-0.0)$)[xnodeA]{};
		\node (yB) at ($(ne)!0.5!(se)+(+0.5*\xr,-0.0)$)[xnodeB]{};
		}

		\begin{scope}[xshift=4*\xsh*\xr cm,yshift=-1*\ysh*\yr cm]
		\nfourgadget{ne/se,se/sw,sw/nw,nw/ne,ne/sw,nw/se}{}{}
		\node (yA) at ($(nw)!0.5!(sw)+(-0.5*\xr,-0.0)$)[xnodeA,label=90:{$y$}]{};
		\node (yB) at ($(ne)!0.5!(se)+(+0.5*\xr,-0.0)$)[xnodeB,label=90:{$y'$}]{};

		\foreach \x in {sw,nw}{
			\draw[xedgeA,ultra thick] (yA) -- (\x);
		}
		\foreach \x in {se,ne}{
			\draw[xedgeB,ultra thick] (yB) -- (\x);
		}
		\draw[gray,dotted] (yA) -- (yB);
		\end{scope}

	\end{tikzpicture}
	\caption{$K_4$-Intersection.}
	\label{fig:H5D4:P:Intersect:K4}
	\end{figure}
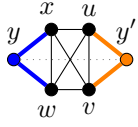
	Assume that there is~$H''\in\bnd(\{H,H'\})$.
	Note that~$\{y,y'\}\subseteq H''$.
	Note that~$|H''\setminus(H\cup H')|\leq 1$,
	since if~$|H''\setminus(H\cup H')|\geq 2$,
	then~$y,y'$ have degree at most two
	(recall that~$y,y'$ have no common neighbor in~$G[H\cup H']$).
	Let~$y''=H''\setminus(H\cup H')$.
	Then,
	$H''\cap \{u,v,w,x\} = 2$,
	and these two have degree two.
\end{proof}

We are finally set to prove our main result.

\begin{proof}[Proof of \cref{prop:H5D4:P}]
	Let~$\calG_{G,F^*,H}$ be the habitat graph.
	For each connected component~$C$ containing a habitat of size five.
	Due to~\cref{lem:H5D4:P:K5,lem:H5D4:P:C4,lem:H5D4:P:Claw,lem:H5D4:P:K4e,lem:H5D4:P:K4},
	we know that the order of the union of the habitats in the connected component is at most 8.
	Hence,
	we can compute in polynomial time
	via brute force a minimum-cost solution
	on each of the corresponding edge subsets.

  After this,
  each connected component contains only habitats of size four.
  Due to \cref{prop:H4D4:P},
  we know that we can solve this case in polynomial time.
\end{proof}
\section{Habitats of size six}
\label{sec:six}
	\subsection{\(\eta=6\) and \(\Delta\geq 5\)}
\begin{proposition}
	\label{prop:H6D5:NPhard}
 \brgbpAcr{} is \NP-hard even if~$\eta=6$ and~$\Delta\geq 5$.
\end{proposition}
We give a polynomial-time many-one reduction from \hcvcAcr{} (\probref{ttts}).
The intuition for the reduction is similar to~\cref{int:H4D7:NPhard}.
\begin{construction}\label{constr:H6D5:NPhard}
 Let~$I=(G,p,C)$ be an instance of HCVC with~$G=(V,E)$ and let $V=\{v_1,\dots,v_n\}$ be enumerated such that~$C=(V,\bigcup_{i=1}^{n-1}\{\{v_i,v_{i+1}\}\}\cup \{\{v_n,v_1\}\})$ forms a Hamiltonian cycle in~$G$,
 and let~$m\ceq |E|$.
 Let~$X=X'\cup X''$ with~$X'=\{(i,i+1)\mid i\in\set{n-1}\}\cup\{(n,1)\}$ and~$X''=\{(i,j)\mid \{v_i,v_j\}\in E\land i<j\}\setminus(X'\cup \{(1,n)\})$
 (note that~$|X|=|E|$).
 Construct an instance~$I'=(G',\calH,k)$ with~$G'=(V',E')$ of \brgbpAcr{} as follows
 (see \cref{fig:H6D5:NPhard} for an illustration).
 \begin{figure}[t]
  \centering
	\begin{tikzpicture}
   \def\xr{1}
   \def\yr{1}
   \tikzpreamble{}
   \def\teps{0.2*\xr*\yr}

   \newcommandx{\nodegadget}[7][6=0]{%
		\node (#1sw) at (#2*\xr-0.5*\xr,#3*\yr-0.4*\yr)[xnode,label={[label distance=-5pt](-135+#6):{$\ell_{#7}^1$}}]{};
		\node (#1se) at (#2*\xr+0.5*\xr,#3*\yr-0.4*\yr)[xnode,label={[label distance=-4pt](-45+#6):{$b_{#7}^2$}}]{};
		\node (#1nw) at (#2*\xr-0.5*\xr,#3*\yr+0.4*\yr)[xnode,label={[label distance=-5pt](135+#6):{$\ell_{#7}^2$}}]{};
		\node (#1ne) at (#2*\xr+0.5*\xr,#3*\yr+0.4*\yr)[xnode,label={[label distance=-4pt](45+#6):{$t_{#7}^2$}}]{};

		\node (#1n) at (#2*\xr,#3*\yr+1*\yr)[xnode,label={[label distance=-2pt](90+#6):{$t_{#7}^1$}}]{};
		\node (#1s) at (#2*\xr,#3*\yr-1*\yr)[xnode,label={[label distance=-2pt](-90+#6):{$b_{#7}^1$}}]{};

		\foreach \x/\y in {%
			se/ne,
			sw/nw,
			se/sw,
			ne/nw,
			se/s,s/sw,
			n/ne,n/nw,n/s}{\draw[xedge] (#1\x) to (#1\y);}

		\foreach \x/\y in {#4}{\draw[xfedge] (#1\x) to (#1\y);}
		\foreach \x/\y in {#5}{\draw[xsedge] (#1\x) to (#1\y);}
	 }

	 \newcommandx{\nodegadgetS}[8][6=0]{%
		\node (#1sw) at (#2*\xr-0.5*\xr,#3*\yr-0.4*\yr)[xnode,label={[label distance=-3pt,yshift=+2pt](-135):{$\ell_{#7}^1$}}]{};
		\node (#1se) at (#2*\xr+0.5*\xr,#3*\yr-0.4*\yr)[xnode,label={[label distance=-3pt,yshift=+2pt](-45+#6):{$b_{#7}^2$}}]{};
		\node (#1nw) at (#2*\xr-0.5*\xr,#3*\yr+0.4*\yr)[xnode,label={[label distance=-3pt,yshift=-2pt](135):{$\ell_{#7}^2$}}]{};
		\node (#1ne) at (#2*\xr+0.5*\xr,#3*\yr+0.4*\yr)[xnode,label={[label distance=-3pt,yshift=-2pt](45+#6):{$t_{#7}^2$}}]{};

		\node (#1n) at (#2*\xr,#3*\yr+1*\yr)[xnode,label={[label distance=-1pt,yshift=1pt]0:{$t_{#7}^1$}}]{};
		\node (#1s) at (#2*\xr,#3*\yr-1*\yr)[xnode,label={[label distance=-1pt,yshift=-1pt]0:{$b_{#7}^1$}}]{};

		\foreach \x/\y in {%
			se/ne,
			sw/nw,
			se/sw,
			ne/nw,
			se/s,s/sw,
			n/ne,n/nw,n/s}{\draw[xedge] (#1\x) to (#1\y);}

		\foreach \x/\y in {#4}{\draw[xfedge] (#1\x) to (#1\y);}
		\foreach \x/\y in {#5}{\draw[xsedge] (#1\x) to (#1\y);}
		\foreach \x/\y in {#8}{\draw[xbedge] (#1\x) to (#1\y);}
	 }

	 \newcommand{\edgegadget}[5]{%
		\node (#1sw) at (#2*\xr-1*\xr,#3*\yr-0.5*\yr)[xnode]{};
		\node (#1se) at (#2*\xr+1*\xr,#3*\yr-0.5*\yr)[xnode]{};
		\node (#1nw) at (#2*\xr-1*\xr,#3*\yr+0.5*\yr)[xnode]{};
		\node (#1ne) at (#2*\xr+1*\xr,#3*\yr+0.5*\yr)[xnode]{};
		\node (#1nc) at (#2*\xr,#3*\yr+0.2*\yr)[xnode]{};
		\node (#1sc) at (#2*\xr,#3*\yr-0.2*\yr)[xnode]{};

		\foreach \x/\y in {%
			se/ne,
			sw/nw,
			se/sw,
			ne/nw,
			se/sc,sc/sw,
			nc/ne,nc/nw,nc/sc}{\draw[xedge] (#1\x) to (#1\y);}

		\foreach \x/\y in {#4}{\draw[xfedge] (#1\x) to (#1\y);}
		\foreach \x/\y in {#5}{\draw[xsedge] (#1\x) to (#1\y);}
	 }

	 \newcommand{\edgegadgetL}[5]{%
		\node (#1sw) at (#2*\xr-1*\xr,#3*\yr-0.5*\yr)[xnode,label=180:{$b_i^1$}]{};
		\node (#1se) at (#2*\xr+1*\xr,#3*\yr-0.5*\yr)[xnode,label=0:{$\ell_i^1$}]{};
		\node (#1nw) at (#2*\xr-1*\xr,#3*\yr+0.5*\yr)[xnode,label=90:{$b_i^2$}]{};
		\node (#1ne) at (#2*\xr+1*\xr,#3*\yr+0.5*\yr)[xnode,label=90:{$\ell_i^2$}]{};
		\node (#1nc) at (#2*\xr,#3*\yr+0.2*\yr)[xnode,label=90:{$y_{i,i+1}$}]{};
		\node (#1sc) at (#2*\xr,#3*\yr-0.2*\yr)[xnode,label=-90:{$z_{i,i+1}$}]{};

		\foreach \x/\y in {%
			se/ne,
			sw/nw,
			se/sw,
			ne/nw,
			se/sc,sc/sw,
			nc/ne,nc/nw,nc/sc}{\draw[xedge] (#1\x) to (#1\y);}

		\foreach \x/\y in {#4}{\draw[xfedge] (#1\x) to (#1\y);}
		\foreach \x/\y in {#5}{\draw[xsedge] (#1\x) to (#1\y);}
	 }

	 \newcommand{\edgegadgetM}[5]{%
		\node (#1sw) at (#2*\xr-1*\xr,#3*\yr-0.5*\yr)[xnode,label=180:{$t_i^2$}]{};
		\node (#1se) at (#2*\xr+1*\xr,#3*\yr-0.5*\yr)[xnode,label=0:{$t_j^1$}]{};
		\node (#1nw) at (#2*\xr-1*\xr,#3*\yr+0.5*\yr)[xnode,label=90:{$t_i^1$}]{};
		\node (#1ne) at (#2*\xr+1*\xr,#3*\yr+0.5*\yr)[xnode,label=90:{$t_j^2$}]{};
		\node (#1nc) at (#2*\xr,#3*\yr+0.2*\yr)[xnode,label=90:{$y_{i,j}$}]{};
		\node (#1sc) at (#2*\xr,#3*\yr-0.2*\yr)[xnode,label=-90:{$z_{i,j}$}]{};
		\node (#1c) at (#2*\xr-0.4*\xr,#3*\yr)[xnode,label=180:{$d_{i,j}$}]{};

		\foreach \x/\y in {%
			c/sc,c/nc,
			se/ne,
			sw/nw,
			se/sw,
			ne/nw,
			se/sc,sc/sw,
			nc/ne,nc/nw,nc/sc}{\draw[xedge] (#1\x) to (#1\y);}

		\foreach \x/\y in {#4}{\draw[xfedge] (#1\x) to (#1\y);}
		\foreach \x/\y in {#5}{\draw[xsedge] (#1\x) to (#1\y);}
	 }

	 \newcommand{\edgegadgetR}[5]{
		\node (#1nc) at (#2*\xr,#3*\yr+0.2*\yr)[xnode]{};
		\node (#1sc) at (#2*\xr,#3*\yr-0.2*\yr)[xnode]{};

		\foreach \x/\y in {nc/sc}{\draw[xedge] (#1\x) to (#1\y);}

		\foreach \x/\y in {#4}{\draw[xfedge] (#1\x) to (#1\y);}
		\foreach \x/\y in {#5}{\draw[xsedge] (#1\x) to (#1\y);}
	 }

	 \newcommandx{\xlabel}[3][1=-0.5,2=1.2]{\node at (#1*\xr,#2*\yr)[anchor=east]{(#3)};}

   \begin{scope}[xshift=-3*\xr cm, yshift=-3.0*\yr cm]
		\xlabel{a}
			\nodegadgetS{X}{0}{0}{n/nw,n/s,ne/se,s/sw}{nw/ne,se/sw}{i}{}
			\nodegadgetS{Y}{2.0}{0}{n/nw,n/s,ne/se,s/sw}{}{i}{n/ne,s/se,sw/nw}
	 \end{scope}

   \begin{scope}[xshift=-2*\xr cm, yshift=-5.75*\yr cm]
    \xlabel[-1.25][1]{b}
    \edgegadgetL{C}{0}{0}{nc/sc,nc/ne,ne/nw,nw/nc,sc/se,se/sw,sw/sc}{}

    \draw[xhabA] \xne{Cne} -- \xnw{Cnw} -- \xsw{Cnw} -- \xs{Cnc} -- \xse{Cne} -- cycle;
    \draw[xhabB] \xse{Cse} -- \xsw{Csw} -- \xnw{Csw} -- \xn{Csc} -- \xne{Cse} -- cycle;
	 \end{scope}

	 \begin{scope}[xshift=-2*\xr cm, yshift=-7.75*\yr cm]
	  \xlabel[-1.25][1]{c}
    \edgegadgetM{D}{0}{0}{nc/sc,nc/ne,ne/nw,nw/nc,sc/se,se/sw,sw/sc,c/sc,c/nc}{}
    \draw[xhabC] \xne{Dnc} -- \xnw{Dc} -- \xsw{Dc} -- \xse{Dsc} -- cycle;
   \end{scope}

   \begin{scope}[xshift=5.25*\xr cm,yshift=-6.5*\yr cm]
    \xlabel[-4.15][4.4]{d}
    \nodegadget{A}{-3.5}{0}{n/nw,n/s,ne/se,s/sw}{}{i-1}
    \nodegadget{B}{0}{0}{n/nw,n/s,ne/se,s/sw}{}{i}
    \nodegadget{C}{4.75}{0}{n/nw,n/s,ne/se,s/sw}{}{i+1}
		\begin{scope}[rotate around={180:(4*\xr,3*\yr)}]
			\nodegadget{D}{4}{3}{n/nw,n/s,ne/se,s/sw}{}[180]{j}
		\end{scope}

		\begin{scope}[rotate around={180:(-1*\xr,3*\yr)}]
			\nodegadget{E}{-1}{3}{n/nw,n/s,ne/se,s/sw}{}[180]{j+1}
		\end{scope}

		\begin{scope}[rotate around={7:(-1.75*\xr,-0.33*\yr)}]
			\edgegadgetR{AB}{-1.75}{-0.33}{nc/sc}{}
			\node at (ABnc)[label=0:{$y_{i-1,i}$}]{};
			\node at (ABsc)[label={[label distance=-1pt,xshift=3*\xr pt]-90:{$z_{i-1,i}$}}]{};
		\end{scope}

		\begin{scope}[rotate around={7:(2.375*\xr,-0.34*\yr)}]
			\edgegadgetR{BC}{2.375}{-0.34}{nc/sc}{}
			\node at (BCnc)[label=0:{$y_{i,i+1}$}]{};
			\node at (BCsc)[label=-90:{$z_{i,i+1}$}]{};
		\end{scope}

		\begin{scope}[rotate around={15:(2*\xr,1.5*\yr)}]
				\edgegadgetR{BD}{2}{1.5}{nc/sc}{}
				\node at (BDnc)[label=0:{$y_{i,j}$}]{};
				\node at (BDsc)[label=0:{$z_{i,j}$}]{};
		\end{scope}
		\node (d) at ($(BDnc)-(0.5*\xr,0.4*\yr)$)[xnode,label=180:{$d_{i,j}$}]{};

		\begin{scope}[rotate around={187:(1.5*\xr,3.35*\yr)}]
			\edgegadgetR{DE}{1.5}{3.35}{nc/sc}{}
			\node at (DEnc)[label=-90:{$y_{j,j+1}$}]{};
			\node at (DEsc)[label=90:{$z_{j,j+1}$}]{};
		\end{scope}

		\foreach\x/\y in {%
		d/BDnc,d/BDsc,
		As/ABsc,As/Bsw,Ase/ABnc,Ase/Bnw,Bnw/ABnc,Bsw/ABsc,
		Bs/BCsc,Bs/Csw,Bse/BCnc,Bse/Cnw,Cnw/BCnc,Csw/BCsc,
		Bn/BDnc,Bn/Dne,Bne/BDsc,Bne/Dn,Dn/BDsc,Dne/BDnc,
		Ds/DEsc,Ds/Esw,Dse/DEnc,Dse/Enw,Enw/DEnc,Esw/DEsc%
		}{\draw[xfedge] (\x) -- (\y);}
		\draw[xfedge,ultra thick] (BCnc) to (BDsc);
		\draw[xfedge,ultra thick] (DEnc) to (BDnc);
		\draw[xhabA] \xne{Bn} -- \xnw{ABnc} -- \xsw{ABsc} -- \xse{Bs} -- cycle;
		\draw[xhabB] \xnw{Bne} -- \xne{BDsc} -- \xse{BCnc} -- \xsw{Bse} -- cycle;
		\draw[xhabC] \xne{Dse} -- \xse{Dne} -- \xs{BDnc} -- \xsw{BDnc} -- \xnw{DEnc} -- cycle;
	\end{scope}

  \end{tikzpicture}
	\caption{Illustration to~\cref{constr:H6D5:NPhard} (red edges are in every solution).
	(a) The two minimal solutions for~$H_i$ (either green or blue).
	(b) Habitats~$H_{i,j}^y$ (cyan) and~$H_{i,j}^z$ (orange) for some~$(i,j)\in X'$.
	(c) Habitat~$H_{i,j}^d$ (magenta) for some~$(i,j)\in X''$.
	(d) An excerpt from~$G'$ around~$V_i$, with habitats~$H_i^\dagger$ (cyan),
	$H_i^\ddagger$ (orange), and~$H_j^\ddagger$ (magenta).
	}
	\label{fig:H6D5:NPhard}
 \end{figure}
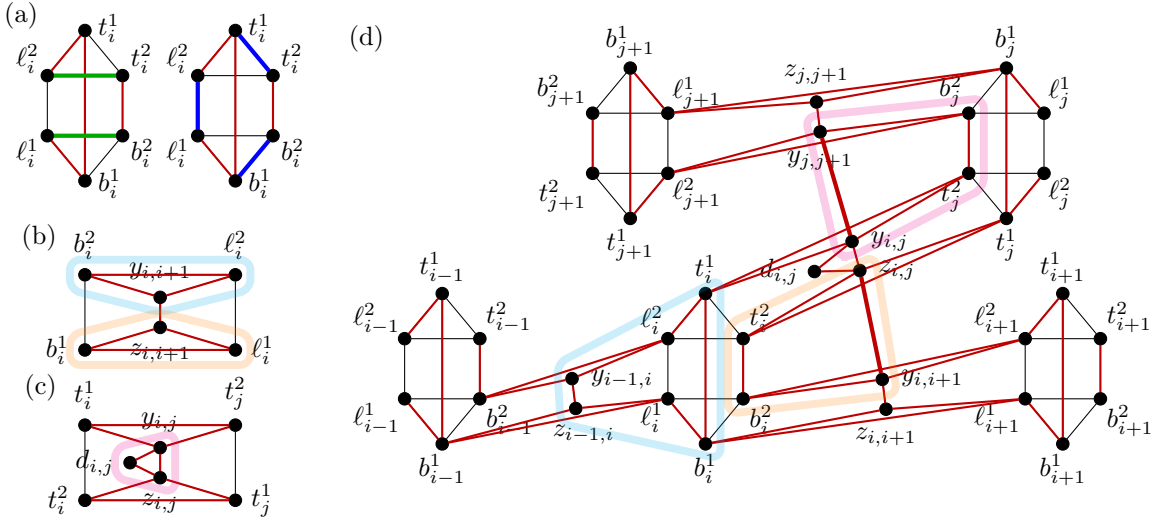
 For each~$i\in\set{n}$,
 add the vertex set~$V_i=\{t_i^z,b_i^z,\ell_i^z\mid z\in\{1,2\}\}$
 and the edge set~$E_i=E_i'\cup E_i^{\rm in} \cup E_i^{\rm out}$,
 where~$E_i'=\{\{\ell_i^2,t_i^1\},\{\ell_i^1,b_i^{1}\},\{t_i^1,b_i^{1}\},\{b_i^{2},t_i^2\}\}$,
 $E_i^{\rm in} = \{\{t_i^1,t_i^2\},\{b_i^{2},b_i^{1}\},\{\ell_i^1,\ell_i^2\}\}$ and~$E_i^{\rm out}=\{\{\ell_i^1,b_i^{2}\},\{\ell_i^2,t_i^2\}\}$.
 For each~$(i,j)\in X'$,
 add two vertices~$y_{i,j}$ and~$z_{i,j}$ and the edge set~$E_{i,j}=\{\{y_{i,j},z_{i,j}\}\}\cup E_{i,j}'$ with $E_{i,j}'\ceq \{\{b_i^1,\ell_j^1\},\{b_i^1,z_{i,j}\},\{z_{i,j},\ell_j^1\}\}\cup\{\{b_i^2,\ell_j^2\},\{b_i^2,y_{i,j}\},\{y_{i,j},\ell_j^2\}\}$.
 For each $(i,j)\in X''$,
 add three vertices~$y_{i,j},z_{i,j},d_{i,j}$ and the edge set $E_{i,j}=E_{i,j}' \cup \{\{y_{i,j},z_{i,j}\},\{y_{i,j},d_{i,j}\},\{d_{i,j},z_{i,j}\}\}$ with $E_{i,j}'\ceq\{\{t_i^1,t_j^2\},\{t_i^1,y_{i,j}\},\{y_{i,j},t_j^2\}\}\cup\{\{t_i^2,t_j^1\},\{t_i^2,z_{i,j}\},\{z_{i,j},t_j^1\}\}$.
 Add the edges~$\{z_{i,j},y_{i,i+1}\}$ and~$\{y_{i,j},y_{j,j'}\}$,
 where $j'=j+1$ if~$j<n$, and~$j'=1$ otherwise.
 For each~$i\in\set{n}$,
 add the habitat~$H_i=V_i$.
 For each~$(i,j)\in X'$,
 add the habitats~$H_{i,j}^y = \{b_i^2,y_{i,j},\ell_j^2\}$,
 $H_{i,j}^z = \{b_i^1,z_{i,j},\ell_j^1\}$,
 and~$H_{i,j} = H_{i,j}^y\cup H_{i,j}^z$,
 as well as~$H_j^\dagger = \{y_{i,j},z_{i,j},\ell_j^1,b_j^1,t_j^1,\ell_j^2\}$.
 For each~$(i,j)\in X''$,
 add~$H_{i,j}^y = \{t_i^1,y_{i,j},t_j^2\}$,
 $H_{i,j}^z = \{t_i^2,z_{i,j},t_j^1\}$,
 $H_{i,j}^d = \{y_{i,j},z_{i,j},d_{i,j}\}$
 and~$H_{i,j} = H_{i,j}^y\cup H_{i,j}^z$,
 as well as the habitats~$H_i^\ddagger=\{t_i^2,z_{i,j},y_{i,i+1},b_i^2\}$
 and $H_j^\ddagger=\{t_j^2,y_{i,j},y_{j,j'},b_j^2\}$,
 where~$j'=j+1$ if~$j<n$, and~$j'=1$ otherwise.
 Set~$k=(7n+\frac{11}{2}n)+(4n+2n+p)$.
 \cqed
\end{construction}
\begin{lemma}%
 \label{lem:H6D5:NPhard:included}
 Let~$I'$ be a \yes-instance.
 Then,
 for every solution~$F$ the following holds:
 \begin{inparaenum}[(i)]
	\item For every~$i\in\set{n}$,
		we have~$E_i'\subseteq F$.
  \item For every~$(i,j)\in X$,
		we have~$E_{i,j}\subseteq F$
	\item For every~$(i,j)\in X''$,
		we have~$\{\{z_{i,j},y_{i,i+1}\},\{y_{i,j},y_{j,j'}\}\}\subseteq F$,
		where~$j'=j+1$ if~$j<n$, and~$j'=1$ otherwise.
 \end{inparaenum}
\end{lemma}
\begin{proof}
 (i) Since the vertex~$t_i^1$ neighbors only~$\ell_i^2$ and~$b_i^1$ in~$G'[H_i^\dagger]$,
 we have $\{\{\ell_i^2,t_i^1\},\{b_i^1,t_i^1\}\}\subseteq F$,
 Moreover,
 since vertex~$b_i^1$ neighbors only~$\ell_i^1$ and~$t_i^1$ in~$G'[H_i^\dagger]$,
 we have $\{\{\ell_i^1,b_i^1\},\{b_i^1,t_i^1\}\}\subseteq F$.
 Since $b_i^2$ and~$t_i^2$ are neighbors and each of degree two in~$G'[H_i^\ddagger]$,
 we have~$\{b_i^2,t_i^2\}\in F$.

 (ii) Let $(i,j)\in X'$.
 Since each of~$H_{i,j}^y$ and $H_{i,j}^z$ forms a triangle,
 all edges in~$E_{i,j}\setminus \{\{y_{i,j},z_{i,j}\}\}$ are in~$F$.
 Since in~$G'[H_j^\dagger]$ the vertices $y_{i,j}$ and $z_{i,j}$ are neighboring and each has degree two,
 we have~$\{y_{i,j},z_{i,j}\}\in F$.
 Let~$(i,j)\in X''$.
 Since each of~$H_{i,j}^y,H_{i,j}^z,H_{i,j}^d$ form triangles,
 all edges in~$E_{i,j}$ are in~$F$.

 (iii)
 Let~$(i,j)\in X''$.
 Since each of $G'[H_i^\ddagger]$ and $G'[H_j^\ddagger]$ is a cycle
 where $z_{i,j}$ neighbors~$y_{i,i+1}$ and~$y_{i,j}$ neighbors~$y_{j,j'}$,
 respectively,
 we have $\{\{z_{i,j},y_{i,i+1}\},\{y_{i,j},y_{j,j'}\}\}\subseteq F$,
 where~$j'=j+1$ if~$j<n$, and~$j'=1$ otherwise.
\end{proof}

\begin{lemma}%
 \label{lem:H6D5:NPhard:sol}
 Let~$I'$ be a \yes-instance.
 Then,
 there is a solution~$F$
 such that for each~$i\in\set{n}$,
 either
 \begin{inparaenum}[(i)]
  \item $E_i^{\rm in}\subseteq F$ and~$E_i^{\rm out}\cap F=\emptyset$,
		or
  \item $E_i^{\rm out}\subseteq F$ and~$E_i^{\rm in}\cap F=\emptyset$.
 \end{inparaenum}
\end{lemma}
\begin{proof}
 Observe that for every~$i\in\set{n}$,
 we have that all but~$t_i^1,b_i^1$ have degree one in~$G'[H_i,E_i']$.
 Hence,
 $G'[H_i,F_i']$ with~$F_i'\subseteq E_i'$ and~$|F_i'|\leq 6$ is 2-connected
 if and only if
 $F_i'=E_i'\cup E_i^{\rm out}$.
 Moreover,
 note that for the same reason,
 $E_i'\cup E_i^{\rm in}$ is minimal in the sense that for every strict subset~$E''\subseteq E_i^{\rm in}$,
 we have that~$G[H_i,E_i'\cup E'']$ is not 2-connected.

 Now suppose towards a contradiction that the statement is false,
 i.e.,
 for every solution~$F$
 there is an~$i\in\set{n}$
 such that
 $E_i^{\rm in}\cap F\neq \emptyset$ and~$E_i^{\rm out}\cap F\neq \emptyset$.
 Let~$F$ be a solution with the minimum number of edges from the sets~$E_1^{\rm out},\dots, E_n^{\rm out}$,
 and let~$i\in\set{n}$ such that $E_i^{\rm in}\cap F\neq \emptyset$ and~$E_i^{\rm out}\cap F\neq \emptyset$.
 We know that there are at least three edges from~$E_i''\ceq E_i^{\rm in}\cup E_i^{\rm out}$ in~$F$.
 We claim that~$F'\ceq (F\setminus E_i^{\rm out}) \cup E_i^{\rm in}$
 is also a solution to~$I'$.
 We have that
 $|F'|=|F'\setminus E_i''| + |F'\cap E_i''| \leq |F\setminus E_i''| + |F\cap E_i''| = |F|$.
 Moreover,
 due to \cref{lem:H6D5:NPhard:included},
 we only need to check for the habitats~$H_i$ and~$H_{i',j'}$ with $(i',j')\in X$ and~$i\in\{i',j'\}$.
 We know that $G'[H_i,E_i'\cup E_i^{\rm in}]$ is 2-connected and hence
 $G'[H_i,F']$ is still 2-connected.
 For each~$(i',j')\in X$ with~$i\in\{i',j'\}$,
 we have that~$G'[H_{i',j'},F']$ is 2-connected since the edge in $E(G'[H_i,F'])\cap E(G'[H_{i',j'},F'])$ is also in~$F'$.
 Since~$F'$ contains less edges from the sets~$E_1^{\rm out},\dots, E_n^{\rm out}$ than~$F$,
 this contradicts the choice of~$F$.
\end{proof}

\begin{proof}[Proof of \cref{prop:H6D5:NPhard}]
 Let~$I'=(G',\calH,k)$ with~$G'=(V',E')$ be the instance of \brgbpAcr{} obtained from input instance~$I=(G,p,C)$ of \prob{HCVC}
 via~\cref{constr:H6D5:NPhard}.
 We prove that~$I$ is a \yes-instance if and only if~$I'$ is a \yes-instance.

 \RD{} Let~$W\subseteq V$ be a vertex cover of size~$p$.
 Let~$\Gamma\ceq \set{n}$ and~$\Gamma_W\ceq \{i\in\set{n}\mid v_i\in W\}$.
 For each~$i\in \Gamma_W$,
 add the edges~$E_i^{\rm in}$ to~$F$,
 and for each~$i\in \Gamma\setminus \Gamma_W$,
 add the edges~$E_i^{\rm out}$ to~$F$,
 We added~$3p + 2(n-p)=2n+p$ edges in this step.
 Next,
 add all edges according to \cref{lem:H6D5:NPhard:included}:
 These are~$4n$ according to~(i),
 $7\cdot |X'| + 9\cdot |X''| = 7n+\frac{9}{2}n$ according to~(ii),
 and~$2\cdot |X''| = n$ edges according to~(iii).
 Thus,
 it holds that~$|F|\leq k$.
 We claim that~$F$ is a solution.

 For each~$i\in\set{n}$,
 we have that~$G'[H_i,F]$ is 2-connected,
 since each of~$G'[H_i,E_i'\cup E_i^{\rm out}]$ and $G'[H_i,E_i'\cup E_i^{\rm in}]$ is 2-connected.
 Now suppose towards a contradiction that there is~$(i,j)\in X$ such that
 $G'[H_{i,j},F]$ is not 2-connected.
 As the first case,
 let~$(i,j)\in X'$.
 It follows that neither~$\{b_i^1,b_i^2\}$ nor $\{\ell_j^1,\ell_j^2\}$ is in~$F$.
 Thus,
 $\{i,j\}\cap W=\emptyset$,
 and hence,
 edge~$\{v_i,v_j\}\in E$ is not covered by~$W$;
 this contradicts the fact that~$W$ is a vertex cover.
 As the second case,
 let~$(i,j)\in X''$.
 It follows that neither~$\{t_i^1,t_i^2\}$ nor $\{t_j^1,t_j^2\}$ is in~$F$.
 Thus,
 $\{i,j\}\cap W=\emptyset$,
 and hence,
 edge~$\{v_i,v_j\}\in E$ is not covered by~$W$;
 this contradicts the fact that~$W$ is a vertex cover.

 \LD{}
 Let~$F$ be a solution of size at most~$k$ such that for each~$i\in\set{n}$,
 either $E_i^{\rm in}\subseteq F$ and~$E_i^{\rm out}\cap F=\emptyset$,
 or $E_i^{\rm out}\subseteq F$ and~$E_i^{\rm in}\cap F=\emptyset$
 (which exists due to~\cref{lem:H6D5:NPhard:sol}).
 We claim that~$W=\{v_i\mid E_i^{\rm in} \subseteq F\}$ is a vertex cover of size~$p$.
 We have that~$F$ contains all edges according to~\cref{lem:H6D5:NPhard:included};
 these are $4n+7n+\frac{11}{2}n$ edges.
 It follows that at most~$p$ times,
 $E_i^{\rm in}$ can be chosen into~$F$,
 and hence,
 $|W|\leq p$.
 Now suppose towards a contradiction that there is an edge~$e=\{v_i,v_j\}\in E$ with~$e\cap W=\emptyset$.
 Then,
 neither~$E_i^{\rm in}$ nor~$E_j^{\rm in}$ are contained in~$F$.
 This implies that
 if~$(i,j)\in X'$,
 neither~$\{b_i^1,b_i^2\}$ nor $\{\ell_j^1,\ell_j^2\}$ is in~$F$,
 and hence,
 $G'[H_{i,j},F]$ is not 2-connected;
 this contradicts the fact that~$F$ is a solution.
 Similarly,
 if~$(i,j)\in X''$,
 neither~$\{t_i^1,t_i^2\}$ nor $\{t_j^1,t_j^2\}$ is in~$F$,
 and hence,
 $G'[H_{i,j},F]$ is not 2-connected;
 this contradicts the fact that~$F$ is a solution.
\end{proof}

\begin{remark}
 \label{rem:H6D5:NPhard}
 The reduction works also for the edge-variant of \rgbpAcr{}
 due to the following:
 The logic for the habitats~$H_i$ and~$H_{i,j}$ works exactly the same way
 for 2-edge-connectivity.
 For~$H_i$,
 the two possible solutions are still minimal as they reconnect the habitat into cycles.
 For~$H_{i,j}$,
 the edge~$\{y_{i,j},z_{i,j}\}$ is a cut when none of the two unforced edges is contained in the solution.
 All the remaining, auxiliary habitats induce cycles (possibly with a chord)
 and hence enforce exactly the same edge set.
\end{remark}
\subsection{$\eta=6$ and~$\Delta\leq 3$}
\Cref{fig:H6D3:P} shows all pairwise non-isomorphic 2-vertex- and 2-edge-connected size-6 habitats of degree at most~3.
\begin{figure}[t]
 \centering
 \begin{tikzpicture}
  \def\xr{1}
  \def\yr{1}
  \def\xsh{2.225}
  \def\ysh{2.7}
  \tikzpreamble{}

  \newcommand{\nodegadget}[3]{%
		\node (sw) at (-0.5*\xr,-0.4*\yr)[xnode,label={[label distance=-4pt](-135):{$x_5$}}]{};
		\node (se) at (0.5*\xr,-0.4*\yr)[xnode,label={[label distance=-4pt](-45):{$x_3$}}]{};
		\node (nw) at (-0.5*\xr,0.4*\yr)[xnode,label={[label distance=-4pt](135):{$x_6$}}]{};
		\node (ne) at (0.5*\xr,0.4*\yr)[xnode,label={[label distance=-4pt](45):{$x_2$}}]{};
		\node (n) at (0,1*\yr)[xnode,label={[label distance=-1pt](0):{$x_1$}}]{};
		\node (s) at (0,-1*\yr)[xnode,label={[label distance=-1pt](-0):{$x_4$}}]{};

		\foreach \x/\y in {#1}{\draw[xedge] (\x) to (\y);}

		\foreach \x/\y in {#2}{\draw[xfedge] (\x) to (\y);}
		\foreach \x/\y in {#3}{\draw[xsedge] (\x) to (\y);}
	 }
	 \newcommandx{\xlabel}[2][1=1]{\node at (-0.5*\xr,#1*1.25*\yr)[anchor=east]{(#2)};}
	 \newcommand{\xlabelx}[2]{\node at (-#1*0.5*\xr,1.375*\yr)[anchor=east]{(#2)};}

	\begin{scope}[xshift=0*\xsh*\xr cm]
	 \xlabelx{1.66}{a}
	 \xlabel{i}
	 \nodegadget{}{n/ne,ne/se,se/s,s/sw,sw/nw,nw/n}{}
	\end{scope}
	\begin{scope}[xshift=1*\xsh*\xr cm]
	 \xlabel{ii}
	 \nodegadget{}{n/ne,n/nw,nw/s,nw/sw,ne/se,ne/s,se/sw}{}
	\end{scope}
	\begin{scope}[xshift=2*\xsh*\xr cm]
	\xlabel{iii}
	 \nodegadget{nw/ne}{n/ne,n/nw,nw/sw,sw/s,s/se,se/ne}{}
	\end{scope}
	\begin{scope}[xshift=3*\xsh*\xr cm]
	 \xlabel{iv}
	 \nodegadget{n/s}{n/ne,n/nw,s/sw,s/se,se/ne,sw/nw}{}
	\end{scope}
	\begin{scope}[xshift=4*\xsh*\xr cm, yshift=0*\ysh*\yr cm]
	 \nodegadget{ne/nw,se/sw}{n/ne,s/sw,n/nw,s/se,sw/nw,se/ne}{}
	\end{scope}

	\begin{scope}[xshift=0*\xsh*\xr cm, yshift=-1*\ysh*\yr cm]
	 \xlabel[0.9]{b}
	 \nodegadget{n/ne,n/nw,n/s,nw/ne}{s/sw,s/se,se/ne,sw/nw}{}
	\end{scope}
	\begin{scope}[xshift=1*\xsh*\xr cm, yshift=-1*\ysh*\yr cm]
	 \nodegadget{nw/sw,nw/s,se/sw,se/s,s/sw}{n/nw,n/ne,ne/se}{}
	\end{scope}
	\begin{scope}[xshift=2*\xsh*\xr cm, yshift=-1*\ysh*\yr cm]
	 \nodegadget{ne/sw,nw/se,nw/sw,ne/se,ne/sw}{n/nw,n/ne,se/s,sw/s}{}
	\end{scope}
	\begin{scope}[xshift=3.25*\xsh*\xr cm, yshift=-1*\ysh*\yr cm]
	 \xlabel[0.9]{c}
	 \nodegadget{n/nw,n/ne,n/s,ne/sw,nw/se,nw/sw,ne/se,ne/sw,se/s,sw/s}{}{}
	\end{scope}
	\begin{scope}[xshift=4.25*\xsh*\xr cm, yshift=-1*\ysh*\yr cm]
	 \nodegadget{n/nw,n/ne,n/s,nw/ne,sw/se,nw/sw,ne/se,se/s,sw/s}{}{}
	\end{scope}

 \end{tikzpicture}
 \caption{All pairwise non-isomorphic 2-vertex- and 2-edge-connected habitats of size six with degree at most three. Red edges are forced in every reduced instance.}
	\label{fig:H6D3:P}
\end{figure}
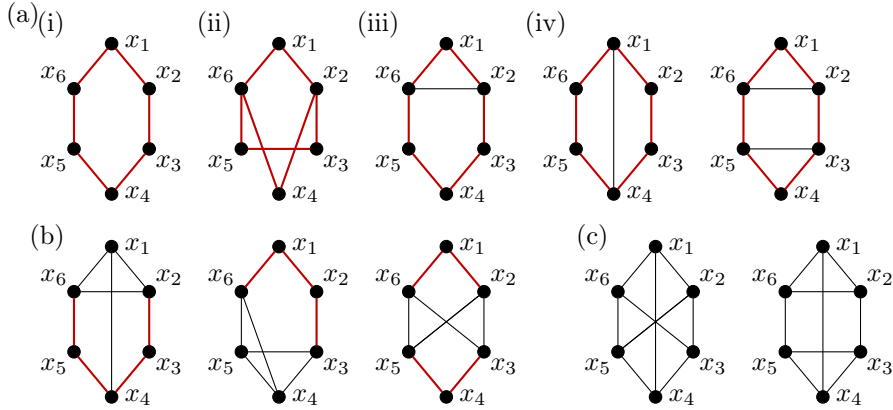
In a reduced instance,
no habitat in (a) appears.
Thus,
the habitats in (b) and (c) remain.
We call them \emph{(6,3)-habitats}.
\begin{proposition}
 \label{prop:H6D3:P}
 \rgbpAcr{} is polynomial-time solvable if~$\eta=6$ and~$\Delta=3$.
\end{proposition}
To prove \Cref{prop:H6D3:P},
we essentially combine \Cref{obs:HG:components,prop:H5D3:P}
with the following.
\begin{lemma}
 \label{lem:H6D3:P}
 Let~$I=(G,c,\Fin,\calH,k)$ be reduced and~$H\in\calH$ be a (6,3)-habitat.
 If~$\{H,H'\}\in E(\calG_{G,\Fin,\calH})$ for some~$H'\in\calH$,
 then~$H'\subseteq H$.
\end{lemma}
\begin{proof}
	We prove that~$\bnd(\{H\})=\emptyset$ for every (6,3)-habitat~$H$.
	Suppose towards a contradiction that there is $H'\in \bnd(\{H\})$.
	We distinguish for all five (6,3)-habitats~$H$
	how $H'$ can intersect with~$H$.
	Let~$H$ be a habitat from \cref{fig:H6D3:P}(c).
	Note that in~$G[H]$,
	all vertices have degree three,
	and hence no neighbor~$H'$ of~$H$ in~$\calG_{G,\Fin,\calH}$ with~$H'\setminus H\neq\emptyset$ can exist.
	Next,
	we distinguish the three (6,3)-habitats from \cref{fig:H6D3:P}(b).
	See \Cref{fig:H6D3:P:Intersect} for an illustration and let~$H=\{x_1,\dots,x_6\}$ as depicted.
	\begin{figure}[t]
	\centering
	\begin{tikzpicture}
		\def\xr{1}
		\def\yr{1}
		\def\xsh{2.45}
		\def\ysh{3.5}
		\tikzpreamble{}

		\newcommand{\nodegadget}[3]{%
			\node (sw) at (-0.5*\xr,-0.4*\yr)[xnode,label={[label distance=-4pt](-135):{$x_5$}}]{};
			\node (se) at (0.5*\xr,-0.4*\yr)[xnode,label={[label distance=-4pt](-45):{$x_3$}}]{};
			\node (nw) at (-0.5*\xr,0.4*\yr)[xnode,label={[label distance=-4pt](135):{$x_6$}}]{};
			\node (ne) at (0.5*\xr,0.4*\yr)[xnode,label={[label distance=-4pt](45):{$x_2$}}]{};
			\node (n) at (0,1*\yr)[xnode,label=(90):{$x_1$}]{};
			\node (s) at (0,-1*\yr)[xnode,label=(-90):{$x_4$}]{};

			\foreach \x/\y in {#1}{\draw[xedge] (\x) to (\y);}

			\foreach \x/\y in {#2}{\draw[xfedge] (\x) to (\y);}
			\foreach \x/\y in {#3}{\draw[xsedge] (\x) to (\y);}
		}
		\newcommand{\xlabel}[1]{\node at (-0.5*\xr,1.25*\yr)[anchor=east]{(#1)};}

		\begin{scope}[xshift=0*\xsh*\xr cm, yshift=0*\ysh*\yr cm]
			\xlabel{a}
			\nodegadget{n/ne,n/nw,n/s,nw/ne}{s/sw,s/se,se/ne,sw/nw}{}
		\end{scope}
		\begin{scope}[xshift=-0.5*\xsh*\xr cm, yshift=-1*\ysh*\yr cm]
			\nodegadget{n/ne,n/nw,n/s,nw/ne}{s/sw,s/se,se/ne,sw/nw}{}
			\node (y) at ($(s)+(0,-0.8*\yr)$)[xnode,color=blue,label={[label distance=-2pt](0):{$y$}}]{};
			\draw[xfedge] (sw) -- (y) -- (se);
		\end{scope}
		\begin{scope}[xshift=0.5*\xsh*\xr cm, yshift=-1*\ysh*\yr cm]
			\nodegadget{n/ne,n/nw,n/s,nw/ne}{s/sw,s/se,se/ne,sw/nw}{}
			\node (y) at ($(s)+(-0.5*\xr,-0.8*\yr)$)[xnode,color=blue,label={[label distance=-2pt](180):{$y$}}]{};
			\node (z) at ($(s)+(0.5*\xr,-0.8*\yr)$)[xnode,color=blue,label={[label distance=-2pt](0):{$z$}}]{};
			\draw[xfedge] (sw) -- (y) -- (z) -- (se);
		\end{scope}
		\begin{scope}[xshift=1.5*\xsh*\xr cm, yshift=0*\ysh*\yr cm]
		  \xlabel{b}
			\nodegadget{nw/sw,nw/s,se/sw,se/s,s/sw}{n/nw,n/ne,ne/se}{}
		\end{scope}
		\begin{scope}[xshift=1.5*\xsh*\xr cm, yshift=-1*\ysh*\yr cm]
			\nodegadget{nw/sw,nw/s,se/sw,se/s,s/sw}{n/nw,n/ne,ne/se}{}
			\node (y) at ($(ne)+(0.2*\xr,0.6*\yr)$)[xnode,color=blue,label={[label distance=-2pt](0):{$y$}}]{};
			\draw[xedge] (n) -- (y) -- (ne);
		\end{scope}
		\begin{scope}[xshift=3*\xsh*\xr cm, yshift=0*\ysh*\yr cm]
			\xlabel{c}
			\nodegadget{ne/sw,nw/se,nw/sw,ne/se,ne/sw}{n/nw,n/ne,se/s,sw/s}{}
		\end{scope}
		\begin{scope}[xshift=2.5*\xsh*\xr cm, yshift=-1*\ysh*\yr cm]
			\nodegadget{ne/sw,nw/se,nw/sw,ne/se,ne/sw}{n/nw,n/ne,se/s,sw/s}{}
			\node (yaux) at ($(ne)!0.5!(se)$)[]{};
			\node (y) at ($(yaux)+(0.4*\xr,0.0*\yr)$)[xnode,color=blue,label={[label distance=-2pt](0):{$y$}}]{};
			\draw[xfedge] (n) to [out=0,in=90](y);
			\draw[xfedge] (s) to [out=0,in=-90](y);
		\end{scope}
		\begin{scope}[xshift=3.5*\xsh*\xr cm, yshift=-1*\ysh*\yr cm]
			\nodegadget{ne/sw,nw/se,nw/sw,ne/se,ne/sw}{n/nw,n/ne,se/s,sw/s}{}
			\node (y) at ($(ne)+(0.4*\xr,0.0*\yr)$)[xnode,color=blue,label={[label distance=-2pt](0):{$y$}}]{};
			\node (z) at ($(se)+(0.4*\xr,0.0*\yr)$)[xnode,color=blue,label={[label distance=-2pt](0):{$z$}}]{};
			\draw[xfedge] (n) to [out=0,in=90](y);
			\draw[xfedge] (s) to [out=0,in=-90](z);
			\draw[xfedge] (y) -- (z);
		\end{scope}

		\node at (0.25*\xsh*\xr,-0.5*\ysh*\yr)[]{$\searrow$};
		\node at (-0.25*\xsh*\xr,-0.5*\ysh*\yr)[]{$\swarrow$};
		\node at (1.5*\xsh*\xr,-0.5*\ysh*\yr)[]{$\downarrow$};
		\node at (3.25*\xsh*\xr,-0.5*\ysh*\yr)[]{$\searrow$};
		\node at (2.75*\xsh*\xr,-0.5*\ysh*\yr)[]{$\swarrow$};
	\end{tikzpicture}
	\caption{Possible intersections of~$H$ and~$H'$ of size five with degree at most three.}
	\label{fig:H6D3:P:Intersect}
	\end{figure}
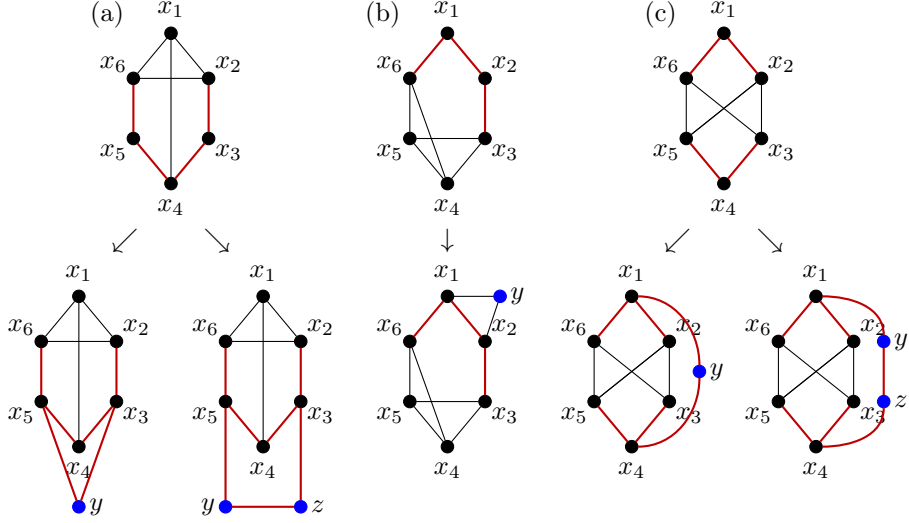

	(a):
	Any vertex in~$H'\setminus H$ can only be adjacent with~$x_5$ or~$x_3$,
	and hence~$x_3,x_5\in H'$.
	Since~$H'$ must intersect with~$H$ in an unforced edge,
	$|H'\setminus H|\leq 2$,
	and any of~$\{x_1,x_2\}$, $\{x_1,x_6\}$, $\{x_2,x_6\}$, and~$\{x_1,x_4\}$ must be contained in~$H'$.

	\xcase{1}{$H'\setminus H=\{y\}$}
	Let~$H''=\{y,x_3,x_5\}$.
	If $H'\setminus H''=\{x_2,x_6\}$,
	then $H'$ induces a $C_5$.
	If~$H'\setminus H''=\{x_1,x_2,x_6\}$,
	then~$H'$ induces \cref{fig:H6D3:P}(iii).
	If~$H'\setminus H''=\{x_4,x_2,x_6\}$,
	then~$H'$ induces \cref{fig:H6D3:P}(ii).
	Finally,
	if~$H'\setminus H''\in\{\{x_4,x_1,x_2\},\{x_4,x_1,x_6\}\}$,
	then~$H'$ induces \cref{fig:H6D3:P}(iv).
	Either way,
	this contradicts the fact that~$I$ is reduced.

	\xcase{2}{$H'\setminus H=\{y,z\}$}
	Let~$H''=\{z,y,x_3,x_5\}$.
	Then,
	$H'\setminus H''$ must be either~$x_4$ or~$\{x_2,x_6\}$,
	and hence,
	$H'$ induces a~$C_5$ or \cref{fig:H6D3:P}(i);
	a contradiction to the fact that~$I$ is reduced.

	(b):
	Any vertex in~$H'\setminus H$ can only be adjacent with~$x_1$ or~$x_2$.
	Since~$H'$ must intersect with~$H$ in an unforced edge,
	$|H'\setminus H|=1$,
	and~$x_3,x_6\in H'$.
	Then,
	$H'$ contains either~$x_4$ or~$x_5$;
	in both cases,
	$H'$ induces a~\cref{fig:H6D3:P}(iii);
	a contradiction to the fact that~$I$ is reduced.

	(c):
	Any vertex in~$H'\setminus H$ can only be adjacent with~$x_1$ or~$x_4$.
	Since~$H'$ must intersect with~$H$ in an unforced edge,
	$|H'\setminus H|\leq 2$,
	and any of~$\{x_2,x_3\}$, $\{x_2,x_5\}$, $\{x_3,x_6\}$, and~$\{x_5,x_6\}$ must be contained in~$G[H']$.
	We distinguish into the size of $H'\setminus H$.

	\xcase{1}{$H'\setminus H=\{y\}$}
	Let~$H''=\{y,x_1,x_4\}$.
	We exploit the symmetry of \cref{fig:H6D3:P:Intersect}(c).
	If~$|(H'\setminus H'')\cap \{x_2,x_3\}|=|(H'\setminus H'')\cap \{x_5,x_6\}|=1$,
	then~$H'$ induces a~$C_5$.
	If~$|(H'\setminus H'')\cap \{x_2,x_3,x_5,x_6\}|=3$,
	then $H'$ induces a~\cref{fig:H6D3:P}(ii);
	a contradiction to the fact that~$I$ is reduced.

	\xcase{2}{$H'\setminus H=\{y,z\}$}
	Let~$H''=\{y,z,x_1,x_4\}$.
	Then~$|(H'\setminus H'')\cap \{x_2,x_3\}|=|(H'\setminus H'')\cap \{x_5,x_6\}|=1$.
	Either way,
	$H'$ induces \cref{fig:H6D3:P}(i);
	a contradiction to the fact that~$I$ is reduced.
\end{proof}
\begin{proof}[Proof of \cref{prop:H6D3:P}]
 Let~$I$ be a reduced instance.
 Compute the basic habitat graph~$\calG_{G,F^*,\calH}$ in polynomial time.
 Due to~\cref{obs:HG:components},
 we know that we can independently solve each component optimally.
 Due to~\cref{lem:H6D3:P},
 we know that if a component contains a (6,3)-habitat,
 then the union of all habitats in this component has cardinality six.
 Hence,
 in polynomial-time,
 we can compute a minimum-cost solution among the at most~$2^{\binom{6}{2}}$ possible solutions or report that~$I$ is a \no-instance.
 Every other component that includes only habitats of size at most five are polynomial-time solvable due to \cref{prop:H5D3:P}.
\end{proof}
\section{Degree Three and Four}
\label{sec:deg34}

We show that for maximum degree three and four,
the problem is \NP-hard for fairly large habitats of size twenty-two and thirteen, respectively.

\begin{proposition}
 \label{prop:D3D4:NPhard}
 \brgbpAcr{} is \NP-hard even if
 \begin{inparaenum}[(i)]
		\item $\eta=13$ and~$\Delta\geq 4$, or\label{item:H13D4:NPhard}
		\item $\eta=22$ and~$\Delta\geq 3$.\label{item:H22D3:NPhard}
 \end{inparaenum}
\end{proposition}
Both reductions follow the same conceptional idea by Herkenrath~\etal~\cite[Constr.~3]{HerkenrathFGK22}.
We give a polynomial-time many-one reduction from \hcvcAcr~(\cref{prob:hcvc}).

\begin{construction}
 \label{constr:H22D3:NPhard}
 Let~$I=(G,p,C)$ be an instance of HCVC with~$G=(V,E)$ and let $V=\{v_1,\dots,v_n\}$ be enumerated such that~$C=(V,\bigcup_{i=1}^{n-1}\{\{v_i,v_{i+1}\}\}\cup \{\{v_n,v_1\}\})$ forms a Hamiltonian cycle in~$G$.
 Construct an instance~$I'=(G',\calH,k)$ of~\brgbpAcr{} as follows
 (see \cref{fig:H22D3:NPhard} for an illustration).
 \begin{figure}[t]
	\centering
	\begin{tikzpicture}
	 \def\xr{1}
	 \def\yr{1}
	 \def\ysq{0}
	 \tikzpreamble{}

	 \newcommand{\nodegaget}[4]{%
	  \ifstrequal{#1}{i}{\def\ncolA{blue}\def\ncolB{magenta}}{\ifstrequal{#1}{j}{\def\ncolA{blue}\def\ncolB{magenta}}{\def\ncolA{black}\def\ncolB{black}}}
	  \ifstrequal{#1}{i}{\def\ncolC{yellow}}{\def\ncolC{black}}

		\node (#1sa) at (#2*\xr+0*\xr,#3*\yr+0.5*\yr-\ysq*\yr)[xnode,label=0:{$s_{#4}^a$},color=\ncolA]{};
		\node (#1sb) at (#2*\xr+0*\xr,#3*\yr-0.5*\yr+\ysq*\yr)[xnode,label=0:{$s_{#4}^b$},color=\ncolB]{};
		\node (#1ra) at (#2*\xr-0.5*\xr,#3*\yr+1*\yr-\ysq*\yr)[xnode,label=0:{$r_{#4}^a$},color=\ncolA]{};
		\node (#1rb) at (#2*\xr-0.5*\xr,#3*\yr-1*\yr+\ysq*\yr)[xnode,label=0:{$r_{#4}^b$},color=\ncolB]{};
		\node (#1ta) at (#2*\xr+0.5*\xr,#3*\yr+1*\yr-\ysq*\yr)[xnode,label=0:{$t_{#4}^a$},color=\ncolA]{};
		\node (#1tb) at (#2*\xr+0.5*\xr,#3*\yr-1*\yr+\ysq*\yr)[xnode,label=0:{$t_{#4}^b$},color=\ncolB]{};
		\node (#1qa) at (#2*\xr-0.5*\xr,#3*\yr+1.5*\yr-\ysq*\yr)[xnode,label=0:{$q_{#4}^a$},color=\ncolC]{};
		\node (#1qb) at (#2*\xr-0.5*\xr,#3*\yr-1.5*\yr+\ysq*\yr)[xnode,label=0:{$q_{#4}^b$},color=\ncolC]{};
		\node (#1ua) at (#2*\xr+0.5*\xr,#3*\yr+1.5*\yr-\ysq*\yr)[xnode,label=0:{$u_{#4}^a$}]{};
		\node (#1ub) at (#2*\xr+0.5*\xr,#3*\yr-1.5*\yr+\ysq*\yr)[xnode,label=0:{$u_{#4}^b$}]{};

		\foreach \c in {a,b}{
			\foreach \x/\y in {q/r,r/s,s/t,t/u}{
					\draw[xfedge] (#1\x\c) -- (#1\y\c);
			}
		}
		\draw[xedge] (#1sa) -- (#1sb);
	 }

	 \nodegaget{im}{-3}{0}{i-1}
	 \nodegaget{i}{0}{0}{i}
	 \nodegaget{ip}{3}{0}{i+1}
	 \nodegaget{j}{7}{0}{j}
	 \node at (4.5*\xr,0*\yr)[]{$\cdots$};

	 \node (ximia) at ($(imsa)!0.5!(isa)$)[xnode,label=0:{$x_{i-1,i}^a$},color=yellow]{};
	 \node (ximib) at ($(imsb)!0.5!(isb)$)[xnode,label=0:{$x_{i-1,i}^b$},color=yellow]{};
	 \draw[xfedge] (ximia) -- (ximib);
	 \node (xiipa) at ($(isa)!0.5!(ipsa)$)[xnode,label=0:{$x_{i,i+1}^a$},color=yellow]{};
	 \node (xiipb) at ($(isb)!0.5!(ipsb)$)[xnode,label=0:{$x_{i,i+1}^b$},color=yellow]{};
	 \draw[xfedge] (xiipa) -- (xiipb);
	 \node (xija) at ($(isa)!0.75!(jsa)$)[xnode,label=0:{$x_{i,j}^a$},color=blue]{};
	 \node (xijb) at ($(isb)!0.75!(jsb)$)[xnode,label=0:{$x_{i,j}^b$},color=magenta]{};
	 \draw[xfedge] (xija) -- (xijb);

	 \foreach \x/\y in {im/i,i/ip}{
		 \draw[xfedge] (\x ua) to [out=25,in=155](\y ua);
		 \draw[xfedge] (\x ub) to [out=-25,in=-155](\y ub);
		 \draw[xfedge] (\x qa) to [out=75,in=90](x\x\y a);
		 \draw[xfedge] (x\x\y a) to (\y qa);
		 \draw[xfedge] (\x qb) to [out=-75,in=-90](x\x\y b);
		 \draw[xfedge] (x\x\y b) to (\y qb);
	 }

		\draw[xfedge] (ita) to [out=25,in=155](jta);
		\draw[xfedge] (itb) to [out=-25,in=-155](jtb);
		\draw[xfedge] (ira) to [out=45,in=135](xija);
		\draw[xfedge] (xija) to (jra);
		\draw[xfedge] (irb) to [out=-45,in=-135](xijb);
		\draw[xfedge] (xijb) to (jrb);

		\draw[xhabA] \xnw{imqa} -- \xsw{imra} -- \xsw{imsa} -- \xse{isa} -- \xse{ita} -- \xne{iua} -- cycle;
		\draw[xhabB] \xsw{imqb} -- \xnw{imrb} -- \xnw{imsb} -- \xne{isb} -- \xne{itb} -- \xse{iub} -- cycle;
		\def\teps{0.3*\xr*\yr}
		\draw[xhabC] \xnw{imqa} -- \xsw{imqb} -- \xse{iub} -- \xne{iua} -- cycle;
		\def\teps{0.35*\xr*\yr}

		\draw[xhabF, line width=4pt] (ximia) to (iqa) to (ira) to [out=45,in=135](xija) to (xijb) to [in=-45,out=-135](irb) to (iqb) to (ximib) to (ximia);

	\end{tikzpicture}
	\caption{Illustration to \cref{constr:H22D3:NPhard}: An excerpt around~$V_i$ where~$(i,j)\in X''$, with habitats $H_{i,i+1}^a$ (cyan), $H_{i,i+1}^b$ (orange), and $H_{i,i+1}$ (magenta), as well as $H_{i,j}^\dagger$ (green), $H_{i,j}^a$ (blue vertices),
	$H_{i,j}^b$ (magenta vertices), and~$H_i^\dagger$ (yellow vertices).}
	\label{fig:H22D3:NPhard}
 \end{figure}
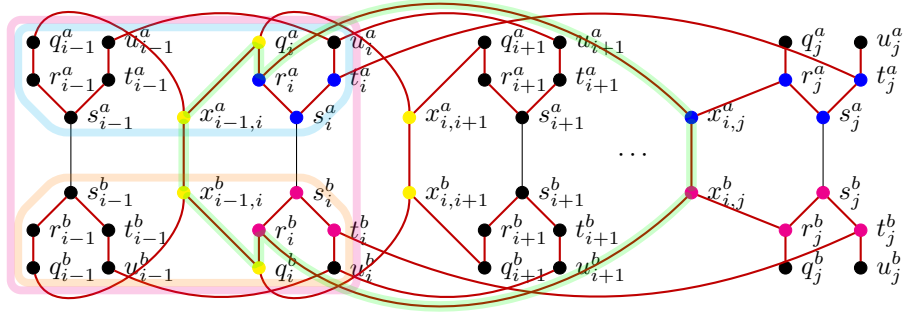
 For each~$i\in\set{n}$,
 add the vertex set~$V_i=V_i^a\cup V_i^b$ with~$V_i^c=\{q_i^c,r_i^c,s_i^c,t_i^c,u_i^c\}$ for~$c\in\{a,b\}$
 and the edge set~$E_i=E_i^a\cup E_i^b\cup \{e_i^*\}$ with~$E_i^c=\{\{q_i^c,r_i^c\},\{r_i^c,s_i^c\},\{s_i^c,t_i^c\},\{t_i^c,u_i^c\}\}$ for~$c\in\{a,b\}$,
 and~$e_i^*=\{s_i^a,s_i^b\}$.
 Let~$X=X'\cup X''$ with~$X'=\{(i,i+1)\mid i\in\set{n-1}\}\cup\{(n,1)\}$
 and~$X''=\{(i,j)\mid \{v_i,v_j\}\in E\land i<j\}\setminus(X'\cup \{(1,n)\})$
 (note that~$|X|=|E|$).
 For each~$(i,j)\in X$,
 add vertices~$x_{i,j}^a$ and $x_{i,j}^b$,
 and the edge~$\{x_{i,j}^a,x_{i,j}^b\}$.
 For each~$(i,j)\in X'$,
 add the edge set~$E_{i,j}=E_{i,j}^a\cup E_{i,j}^b$ with~$E_{i,j}^c=\{\{q_i^c,x_{i,j}^c\},\{q_j^c,x_{i,j}^c\},\{u_i^c,u_j^c\}\}$ for~$c\in\{a,b\}$;
 add the habitats~$H_{i,j}^c=V_i^c\cup V_j^c\cup\{x_{i,j}^c\}$ for each~$c\in\{a,b\}$,
 and the habitat~$H_{i,j}=H_{i,j}^a\cup H_{i,j}^b$.
 For each~$(i,j)\in X''$,
 add the edge set~$E_{i,j}=E_{i,j}^a\cup E_{i,j}^b$
 with~$E_{i,j}^c=\{\{r_i^c,x_{i,j}^c\},\{r_j^c,x_{i,j}^c\},\{t_i^c,t_j^c\}\}$ for~$c\in\{a,b\}$;
 add the habitats~$H_{i,j}^c=(V_i^c\setminus\{q_i^c,u_i^c\})\cup (V_j^c\setminus\{q_j^c,u_j^c\})\cup\{x_{i,j}^c\}$ for each~$c\in\{a,b\}$,
 and the habitat~$H_{i,j}=H_{i,j}^a\cup H_{i,j}^b$.
 Moreover,
 add the habitat~$H_{i,j}^\dagger=\{x_{i',i}^1,q_i^a,r_i^a,x_{i,j}^a,x_{i,j}^b,r_i^b,q_i^b,x_{i',i}^b\}$,
 where~$i'=i-1$ if~$i>1$ and~$i'=n$ otherwise.
 Finally,
 for each~$i\in\set{n}$,
 add habitat~$H_i^\dagger=\{x_{i',i}^a,q_i^a,x_{i,i''}^a,x_{i,i''}^b,q_i^b,x_{i',i}^b\}$,
 where~$i'=i-1$ if~$i>1$ and~$i'=n$ otherwise,
 and~$i''=i+1$ if~$i<n$ and~$i''=1$ otherwise.
 Set~$k=|E'|-n+p$.
 \cqed
\end{construction}

\begin{lemma}\label{lem:H22D3:NPhard}
 Let~$I'$ be a \yes-instance.
 Then,
 for every solution~$F$ it holds that if~$e\in E'\setminus\bigcup_{i=1}^n \{\{s_i^a,s_i^b\}\}$,
 then~$e\in F$.
\end{lemma}

\begin{proof}
 Since for each~$(i,j)\in X$,
 we have that each of~$H_{i,j}^a$ and $H_{i,j}^b$ induce a cycle,
 It remains to show that for all~$(i,j)\in X$ edge~$\{x_{i,j}^a,x_{i,j}^b\}$ is contained in~$F$.
 If~$(i,j)\in X'$,
 then
 the habitat~$H_i^\dagger$ induces a cycle and enforces the edge.
 If~$(i,j)\in X''$,
 then
 the habitat~$H_{i,j}^\dagger$ induces a cycle and enforces the edge.
\end{proof}

\begin{proof}[Proof of \cref{prop:D3D4:NPhard}\eqref{item:H22D3:NPhard}]
 Let~$I'$ be the instance obtained from input instance~$I=(G,p,C)$ of \prob{HCVC}
 via~\cref{constr:H22D3:NPhard}.
 We prove that~$I$ is a \yes-instance if and only if~$I'$ is a \yes-instance.

 \RD{}
 Let~$W\subseteq V$ be a vertex cover with~$|W|=p$ and let~$\calI_W\ceq \{i\in \set{n}\mid v_i\in W\}$ and~$\calJ_W\ceq\set{n}\setminus \calI_W$.
 Let~$F\ceq E'\setminus \bigcup_{i\in \calJ_W} \{\{s_i^a,s_i^b\}\}$.
 It holds that~$|F|=|E'|-n+p$.
 We claim that~$F$ is a solution.
 Observe that for each~$i\in\set{n}$,
 $G[H_i^\dagger,F]$ is 2-connected since every edge in~$G[H_i^\dagger]$ is contained in~$F$.
 With the same argument,
 it holds that each of $G[H_{i,j}^\dagger,F]$, $G[H_{i,j}^a,F]$, and $G[H_{i,j}^b,F]$ is 2-connected for every~$(i,j)\in X$.
 Suppose towards a contradiction that there is~$(i,j)\in X$ such that~$G[H_{i,j},F]$ is not 2-connected.
 This implies that neither~$\{s_i^a,s_i^b\}$ nor~$\{s_j^a,s_j^b\}$ are contained in~$F$
 (note that all remaining edges in~$G[H_{i,j}]$ are contained in~$F$).
 Thus,
 neither~$v_i$ nor~$v_j$ are in~$W$,
 and hence the edge~$\{v_i,v_j\}\in E$ is not covered by~$W$;
 this contradicts the fact that~$W$ is a vertex cover.

 \LD{}
 Let~$F$ be a solution to~$I'$.
 Let~$W\ceq \{v_i\in V\mid \{s_i^a,s_i^b\}\in F\}$.
 We claim that~$W$ is a vertex cover in~$G$ of size at most~$p$.
 Due to~\cref{lem:H22D3:NPhard},
 we know that~$|W|= |F| - |F\cap (E'\setminus\bigcup_{i=1}^n \{\{s_i^a,s_i^b\})| \leq |E'|-n+p - (|E'|-n) = p $.
 Suppose towards a contradiction that there is an edge~$e=\{v_i,v_j\}\in E$ with~$e\cap W=\emptyset$.
 Then,
 neither $\{s_i^a,s_i^b\}$ nor~$\{s_j^a,s_j^b\}$ are contained in~$F$.
 since $F$ contains all edges from~$G[H_{i,j}]$ except for $\{s_i^a,s_i^b\}$ and~$\{s_j^a,s_j^b\}$,
 we have that the habitat~$G[H_{i,j},F]$ is not 2-connected,
 with separators~$x_{i,j}^a$ and $x_{i,j}^b$;
 a contradiction to the fact that~$F$ is a solution.
\end{proof}

\begin{remark}
 \label{rem:H22D3:NPhard}
 The reduction works also for the edge-variant of \rgbpAcr{}
 due to the following:
 The logic for the habitats~$H_{i,j}$ works exactly the same way
 for 2-edge-connectivity,
 since the edge~$\{x_{i,j}^a,x_{i,j}^b\}$ is a cut when none of the two unforced edges~$\{s_i^a,s_i^b\}$, $\{s_j^a,s_j^b\}$ is contained in the solution.
 All the remaining, auxiliary habitats induce cycles
 and hence enforce exactly the same edge set.
\end{remark}

We give a polynomial-time many-one reduction from \cvcTsc{} (\cvcAcr),
which denotes the problem \hcvcAcr{} when the Hamiltonian cycle is excluded from the input.

\begin{construction}\label{constr:H13D4:NPhard}
Let~$I=(G,p)$ be an instance of \cvcAcr{} with~$G=(V,E)$ and let $V=\{v_1,\dots,v_n\}$ be enumerated.
 We construct an instance~$I'=(G',k,\calH)$ of~\rgbpAcr{} as follows
 (see \cref{fig:H13D4:NPhard} for an illustration).
 \begin{figure}[t]
	\centering
	\begin{tikzpicture}
	 \def\xr{1}
	 \def\yr{1}
	 \tikzpreamble{}

	 \newcommand{\nodegaget}[4]{%
	  \ifstrequal{#1}{i}{\def\ncolC{yellow}}{\def\ncolC{black}}

		\node (#1sa) at (#2*\xr+0*\xr,#3*\yr+0.5*\yr)[xnode,label=0:{$s_{#4}^a$}]{};
		\node (#1sb) at (#2*\xr+0*\xr,#3*\yr-0.5*\yr)[xnode,label=0:{$s_{#4}^b$}]{};
		\node (#1ra) at (#2*\xr-0.5*\xr,#3*\yr+1*\yr)[xnode,label=0:{$r_{#4}^a$}]{};
		\node (#1rb) at (#2*\xr-0.5*\xr,#3*\yr-1*\yr)[xnode,label=0:{$r_{#4}^b$}]{};
		\node (#1ta) at (#2*\xr+0.5*\xr,#3*\yr+1*\yr)[xnode,label=0:{$t_{#4}^a$}]{};
		\node (#1tb) at (#2*\xr+0.5*\xr,#3*\yr-1*\yr)[xnode,label=0:{$t_{#4}^b$}]{};

		\foreach \c in {a,b}{
			\foreach \x/\y in {r/s,s/t}{
					\draw[xfedge] (#1\x\c) -- (#1\y\c);
			}
		}
		\draw[xedge] (#1sa) -- (#1sb);
	 }

	 \nodegaget{im}{-3}{0}{i-1}
	 \nodegaget{i}{0}{0}{i}
	 \nodegaget{ip}{3}{0}{i+1}
	 \nodegaget{j}{7}{0}{j}
	 \node at (4.5*\xr,0*\yr)[]{$\cdots$};

	 \node (ximia) at ($(imsa)!0.5!(isa)$)[]{};
	 \node (ximib) at ($(imsb)!0.5!(isb)$)[]{};
	 \node (ximi) at ($(ximia)!0.5!(ximib)$)[xnode,label=0:{$x_{i-1,i}$}]{};
	 \node (xiipa) at ($(isa)!0.5!(ipsa)$)[]{};
	 \node (xiipb) at ($(isb)!0.5!(ipsb)$)[]{};
	 \node (xiip) at ($(xiipa)!0.5!(xiipb)$)[xnode,label=0:{$x_{i,i+1}$}]{};
	 \node (xija) at ($(isa)!0.75!(jsa)$)[]{};
	 \node (xijb) at ($(isb)!0.75!(jsb)$)[]{};
	 \node (xij) at ($(xija)!0.75!(xijb)$)[xnode,label=0:{$x_{i,j}$}]{};

	 \foreach \x/\y in {im/i,i/ip}{
		 \draw[xfedge] (\x ta) to [out=25,in=155](\y ta);
		 \draw[xfedge] (\x tb) to [out=-25,in=-155](\y tb);
		 \draw[xfedge] (\x ra) to [out=75,in=90](x\x\y);
		 \draw[xfedge] (x\x\y) to (\y ra);
		 \draw[xfedge] (\x rb) to [out=-75,in=-90](x\x\y);
		 \draw[xfedge] (x\x\y) to (\y rb);
	 }

		\draw[xfedge] (ita) to [out=25,in=155](jta);
		\draw[xfedge] (itb) to [out=-25,in=-155](jtb);
		\draw[xfedge] (ira) to [out=45,in=115](xij);
		\draw[xfedge] (xij) to (jra);
		\draw[xfedge] (irb) to [out=-45,in=-115](xij);
		\draw[xfedge] (xij) to (jrb);

		\draw[xhabA] \xnw{imra} -- \xsw{imra} -- \xsw{imsa} -- \xsw{ximi} -- \xse{ximi} -- \xse{isa} -- \xse{ita} -- \xne{ita} -- cycle;
		\draw[xhabB] \xsw{imrb} -- \xnw{imrb} -- \xnw{imsb} -- \xnw{ximi} -- \xne{ximi} -- \xne{isb} -- \xne{itb} -- \xse{itb} -- cycle;
		\def\teps{0.3*\xr*\yr}
		\draw[xhabC] \xnw{imra} -- \xsw{imrb} -- \xse{itb} -- \xne{ita} -- cycle;
		\def\teps{0.35*\xr*\yr}
	\end{tikzpicture}
	\caption{Illustration to \cref{constr:H13D4:NPhard}: An excerpt around~$V_i$ where~$(i,j)\in X''$, with habitats $H_{i,i+1}^a$ (cyan), $H_{i,i+1}^b$ (orange), and $H_{i,i+1}$ (magenta).}
	\label{fig:H13D4:NPhard}
 \end{figure}
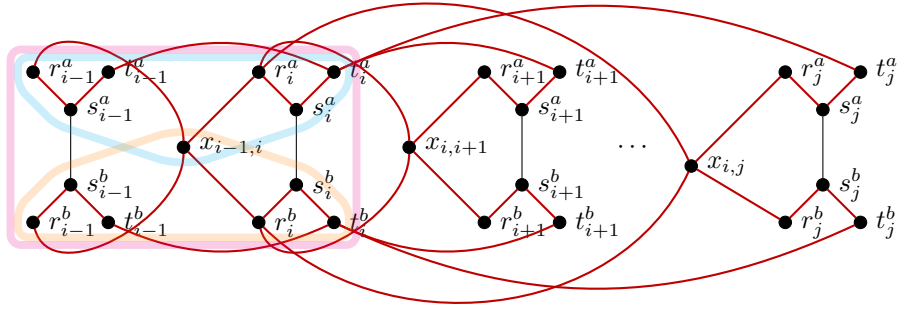
 For each~$i\in\set{n}$,
 add the vertices~$V_i=V_i^a\cup V_i^b$ with~$V_i^c=\{r_i^c,s_i^c,t_i^c\}$ for~$c\in\{a,b\}$
 and the edges~$E_i=E_i^a\cup E_i^b\cup \{e_i^*\}$ with~$E_i^c=\{\{r_i^c,s_i^c\},\{s_i^c,t_i^c\}\}$ for each~$c\in\{a,b\}$, and~$e_i^*=\{s_i^a,s_i^b\}$.
 For each~$(i,j)\in X\ceq \{(i,j)\mid \{v_i,v_j\}\in E\land i<j\}$,
 add~$x_{i,j}$.
 For each~$c\in\{a,b\}$,
 add the edges~$E_{i,j}=E_{i,j}^a\cup E_{i,j}^b$ with~$E_{i,j}^c=\{\{r_i^c,x_{i,j}\},\{r_j^c,x_{i,j}\},\{t_i^c,t_j^c\}\}$ for~$c\in\{a,b\}$.
 For each~$(i,j)\in X$,
 add the habitats~$H_{i,j}^c=V_i^c\cup V_j^c\cup\{x_{i,j}\}$ for~$c\in\{a,b\}$,
 and the habitat~$H_{i,j}=H_{i,j}^a\cup H_{i,j}^b$.
 Set~$k=|E'|-n+p$.
 \cqed
\end{construction}

\begin{lemma}\label{lem:H13D4:NPhard}
 Let~$I'$ be a \yes-instance.
 Then,
 for every solution~$F$ it holds that if~$e\in E'\setminus\bigcup_{i=1}^n \{\{s_i^a,s_i^b\}\}$,
 then~$e\in F$.
\end{lemma}

\begin{proof}
 For each~$(i,j)\in X$ and~$c\in\{a,b\}$,
 $G[H_{i,j}^c]$ induces a cycle on seven vertices.
 Thus,
 every edge in~$E_i^c$, $E_j^c$, and $E_{i,j}^c$
 is contained in~$F$.
\end{proof}

\begin{proof}[Proof of \cref{prop:D3D4:NPhard}\eqref{item:H13D4:NPhard}]
 Let~$I'$ be the instance obtained from input instance~$I=(G,p,C)$ of \prob{HCVC}
 via~\cref{constr:H13D4:NPhard}.
 We prove that~$I$ is a \yes-instance if and only if~$I'$ is a \yes-instance.

 \RD{}
 Let~$W\subseteq V$ be a vertex cover with~$|W|=p$ and let~$\calI_W\ceq \{i\in \set{n}\mid v_i\in W\}$ and~$\calJ_W\ceq\set{n}\setminus \calI_W$.
 Let~$F\ceq E'\setminus \bigcup_{i\in \calJ_W} \{\{s_i^a,s_i^b\}\}$.
 It holds that~$|F|=|E'|-n+p$.
 We claim that~$F$ is a solution.
 Observe that for every~$(i,j)\in X$,
 each of $G[H_{i,j}^a,F]$ and $G[H_{i,j}^b,F]$ is 2-connected since every edge in~$G[H_{i,j}^a]$ and $G[H_{i,j}^b]$ is contained in~$F$.
 Suppose towards a contradiction that there is~$(i,j)\in X$ such that~$G[H_{i,j},F]$ is not 2-connected.
 This implies that neither~$\{s_i^a,s_i^b\}$ nor~$\{s_j^a,s_j^b\}$ are contained in~$F$
 (note that all remaining edges in~$G[H_{i,j}]$ are contained in~$F$).
 Thus,
 neither~$v_i$ nor~$v_j$ are in~$W$,
 and hence the edge~$\{v_i,v_j\}\in E$ is not covered by~$W$;
 this contradicts the fact that~$W$ is a vertex cover.

 \LD{}
 Let~$F$ be a solution to~$I'$.
 Let~$W\ceq \{v_i\in V\mid \{s_i^a,s_i^b\}\in F\}$.
 We claim that~$W$ is a vertex cover in~$G$ of size at most~$p$.
 Due to~\cref{lem:H13D4:NPhard},
 we know that~$|W|= |F| - |F\cap (E'\setminus\bigcup_{i=1}^n \{\{s_i^a,s_i^b\}\})| \leq |E'|-n+p - (|E'|-n) = p $.
 Suppose towards a contradiction that there is an edge~$e=\{v_i,v_j\}\in E$ with~$e\cap W=\emptyset$.
 Then,
 neither $\{s_i^a,s_i^b\}$ nor~$\{s_j^a,s_j^b\}$ are contained in~$F$.
 Since $F$ hence contains all edges from~$G[H_{i,j}]$ except for $\{s_i^a,s_i^b\}$ and~$\{s_j^a,s_j^b\}$,
 we have that the habitat~$G[H_{i,j},F]$ is not 2-connected,
 with separator~$x_{i,j}$;
 a contradiction to the fact that~$F$ is a solution.
\end{proof}

\begin{remark}\label{rem:H14D4:NPhard}
 Note that the reduction does not work for 2-edge-connectivity.
 Yet,
 with a small adaptation the reduction also works for 2-edge-connectivity,
 i.e.,
 when replacing each~$x_{i,j}$ with two adjacent vertices~$x_{i,j}^a$ and~$x_{i,j}^b$
 (see \cref{fig:H14D4:NPhard} for an illustration).
 \begin{figure}[t]
	\centering
	\begin{tikzpicture}
	 \def\xr{1}
	 \def\yr{1}
	 \tikzpreamble{}

	 \newcommand{\nodegaget}[4]{%
	  \ifstrequal{#1}{i}{\def\ncolC{yellow}}{\def\ncolC{black}}

		\node (#1sa) at (#2*\xr+0*\xr,#3*\yr+0.5*\yr)[xnode,label=0:{$s_{#4}^a$}]{};
		\node (#1sb) at (#2*\xr+0*\xr,#3*\yr-0.5*\yr)[xnode,label=0:{$s_{#4}^b$}]{};
		\node (#1ra) at (#2*\xr-0.5*\xr,#3*\yr+1*\yr)[xnode,label=0:{$r_{#4}^a$}]{};
		\node (#1rb) at (#2*\xr-0.5*\xr,#3*\yr-1*\yr)[xnode,label=0:{$r_{#4}^b$}]{};
		\node (#1ta) at (#2*\xr+0.5*\xr,#3*\yr+1*\yr)[xnode,label=0:{$t_{#4}^a$}]{};
		\node (#1tb) at (#2*\xr+0.5*\xr,#3*\yr-1*\yr)[xnode,label=0:{$t_{#4}^b$}]{};

		\foreach \c in {a,b}{
			\foreach \x/\y in {r/s,s/t}{
					\draw[xfedge] (#1\x\c) -- (#1\y\c);
			}
		}
		\draw[xedge] (#1sa) -- (#1sb);
	 }

	 \nodegaget{im}{-3}{0}{i-1}
	 \nodegaget{i}{0}{0}{i}
	 \nodegaget{ip}{3}{0}{i+1}
	 \nodegaget{j}{7}{0}{j}
	 \node at (4.5*\xr,0*\yr)[]{$\cdots$};

	 \node (ximia) at ($(imsa)!0.5!(isa)$)[xnode,label=0:{$x_{i-1,i}^a$}]{};
	 \node (ximib) at ($(imsb)!0.5!(isb)$)[xnode,label=0:{$x_{i-1,i}^b$}]{};
	 \draw[xfedge] (ximia) -- (ximib);
	 \node (xiipa) at ($(isa)!0.5!(ipsa)$)[xnode,label=0:{$x_{i,i+1}^a$}]{};
	 \node (xiipb) at ($(isb)!0.5!(ipsb)$)[xnode,label=0:{$x_{i,i+1}^b$}]{};
	 \draw[xfedge] (xiipa) -- (xiipb);
	 \node (xija) at ($(isa)!0.75!(jsa)$)[xnode,label=0:{$x_{i,j}^a$}]{};
	 \node (xijb) at ($(isb)!0.75!(jsb)$)[xnode,label=0:{$x_{i,j}^b$}]{};
	 \draw[xfedge] (xija) -- (xijb);

	 \foreach \x/\y in {im/i,i/ip}{
		 \draw[xfedge] (\x ta) to [out=25,in=155](\y ta);
		 \draw[xfedge] (\x tb) to [out=-25,in=-155](\y tb);
		 \draw[xfedge] (\x ra) to [out=75,in=90](x\x\y a);
		 \draw[xfedge] (x\x\y a) to (\y ra);
		 \draw[xfedge] (\x rb) to [out=-75,in=-90](x\x\y b);
		 \draw[xfedge] (x\x\y b) to (\y rb);
	 }

		\draw[xfedge] (ita) to [out=25,in=155](jta);
		\draw[xfedge] (itb) to [out=-25,in=-155](jtb);
		\draw[xfedge] (ira) to [out=45,in=135](xija);
		\draw[xfedge] (xija) to (jra);
		\draw[xfedge] (irb) to [out=-45,in=-135](xijb);
		\draw[xfedge] (xijb) to (jrb);

		\draw[xhabA] \xnw{imra} -- \xsw{imra} -- \xsw{imsa} -- \xse{isa} -- \xse{ita} -- \xne{ita} -- cycle;
		\draw[xhabB] \xsw{imrb} -- \xnw{imrb} -- \xnw{imsb} -- \xne{isb} -- \xne{itb} -- \xse{itb} -- cycle;
		\def\teps{0.3*\xr*\yr}
		\draw[xhabC] \xnw{imra} -- \xsw{imrb} -- \xse{itb} -- \xne{ita} -- cycle;
		\def\teps{0.35*\xr*\yr}
		\draw[xhabF, line width=4pt] (ximia) to (ira) to [out=45,in=135](xija) to (xijb) to [in=-45,out=-135](irb) to (ximib) to (ximia);
	\end{tikzpicture}
	\caption{Illustration to the adaptation of \cref{constr:H13D4:NPhard} for 2-edge-connectivity: An excerpt around~$V_i$ where~$(i,j)\in X''$, with habitats $H_{i,i+1}^a$ (cyan), $H_{i,i+1}^b$ (orange), and $H_{i,i+1}$ (magenta). The green-colored habitat enforces the edges~$\{x_{i-1,i}^a,x_{i-1,i}^b\}$ and~$\{x_{i,j}^a,x_{i,j}^b\}$, exemplarily.}
	\label{fig:H14D4:NPhard}
 \end{figure}
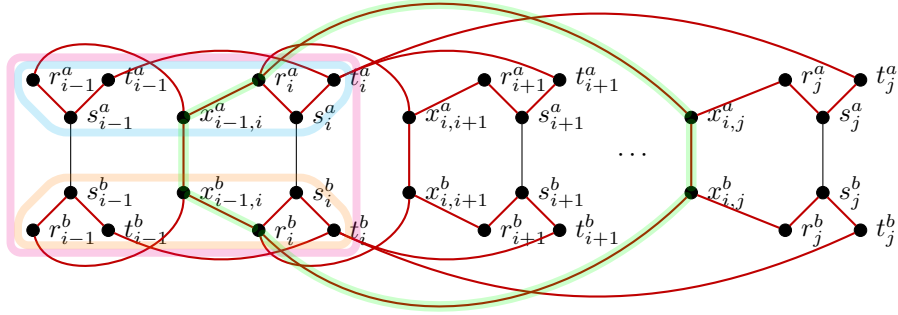
 Note,
 however,
 this adaptation comes with the cost of having habitats of size fourteen.
\end{remark}

\section{Conclusion}

Our results suggest that
some real-world instances of our problem could be tractable.
For example,
graphs derived from Manhattan-like street networks
are likely of maximum degree at most 4
(each bounded face has at most four adjacent faces).
Thus,
when habitats are of size at most 5,
these instances are tractable.
Even if the maximum degree increases from 4 to 5,
for small habitats of size at most 4
we are still in the tractable regime.
Overall,
since all of our \NP-hardness results hold without edge costs or forced edges,
this suggests that these two aspects
play at most a limited role in
the problem’s intractability.

Natural directions for future work include
(i) closing the gaps for~$\eta$ and $\Delta$ in our table and
(ii) $\kappa$-connected habitats with~$\kappa\geq 3$.
As to~(i),
the largest gap is for maximum degree 3: we prove \NP-hardness for habitats of size 22,
but the problem is polynomial-time solvable for habitats of size 6.
In the edge-variant of \rgbpAcr{},
the remaining gap for size-5 habitats is particularly interesting.

{\begingroup
	\let\clearpage\relax
	\renewcommand{\url}[1]{\href{#1}{$\ExternalLink$}}
	\bibliography{gbp-2con-bib}
\endgroup}
\end{document}